\documentclass[A4,11pt]{article}

\usepackage{amsmath, amsthm, amsfonts}
\usepackage{anysize}
\usepackage[utf8]{inputenx}
\usepackage{graphicx}
\usepackage{mathrsfs}
\usepackage{url}
\usepackage[all]{xy}
\usepackage[left=1in,top=1in,right=1in,bottom=1in]{geometry}
\usepackage[activeacute, english]{babel}
\usepackage{float}
\usepackage{mathtools}
\usepackage{listings}
\usepackage{color}
\usepackage[english]{babel}
\usepackage{amssymb}
\usepackage[ruled,vlined]{algorithm2e}
\usepackage{sidecap}
\usepackage[colorlinks=true,linkcolor=blue,citecolor=blue,urlcolor=blue]{hyperref}
\usepackage{cleveref}
\usepackage{eurosym}
\usepackage{subcaption}
\usepackage{fancyhdr}
\usepackage{multirow}
\usepackage{bbm}
\usepackage{enumitem}
\usepackage{dsfont}
\usepackage{lineno}
\usepackage{booktabs}

\DeclareMathOperator*{\sm}{<} 
\DeclareMathOperator*{\bbi}{>} 

\DeclareMathOperator*{\argmax}{argmax}

\newtheorem{theorem}{Theorem}
\newtheorem{proposition}[theorem]{Proposition}%
\newtheorem{lemma}[theorem]{Lemma}
\newtheorem{corollary}[theorem]{Corollary}%

\theoremstyle{definition}{
\newtheorem{example}{Example}%
\newtheorem{remark}{Remark}%
\newtheorem{definition}{Definition}}

\begin{document}
\begin{titlepage}

\begin{center}
\huge{\textbf{General Matching Games}}
\end{center}
\vspace{1cm}

\begin{table}[h]
\centering
\begin{tabular}[t]{lcl}
\textbf{Felipe Garrido-Lucero} & \quad & \textbf{Rida Laraki}\\
IRIT, Université Toulouse Capitole & & Moroccan Center for Game Theory, UM6P\\
Toulouse, France & & Rabat, Morroco \\
felipe.garrido-lucero@ut-capitole.fr & & rida.laraki@um6p.ma
\end{tabular}
\end{table}

\vspace{1cm}

\begin{abstract}
Matching games is a one-to-one two sided market model introduced by Garrido-Lucero and Laraki, in which coupled agents' utilities are endogenously determined as the outcome of a strategic game. They refine the classical pairwise stability by requiring robustness to renegotiation and provide general conditions under which pairwise stable and renegotiation-proof outcomes exist as the limit of a deferred acceptance with competitions algorithm together with a renegotiation process. In this article, we extend their model to a general setting encompassing most of one-to-many matching markets and roommates models and specify two frameworks under which core stable and renegotiation-proof outcomes exist and can be efficiently computed. 
\end{abstract}


\vfill
\small{Preprint. Under review.}
\end{titlepage}

\section{Introduction}

The stable matching problem is a critical research topic in the \textit{Econ-CS} community due to its wide range of applications in both the private and public sectors, spanning online markets \cite{coles_marketplaces_1998}, online advertising \cite{mehta_online_2013}, ride-sharing \cite{banerjee_ride_2019}, the job market \cite{comission_commission_2023}, university admissions \cite{akbarpour_centralized_2022,garrido2025two}, high school teacher assignments \cite{combe_design_2022}, refugee programs \cite{ahani_dynamic_2021}, and even organ transplants \cite{akbarpour_unpaired_2020}. 

The first ones to introduce this problem were Gale and Shapley in their seminal work \cite{gale_college_1962}, where they presented three models: \textit{one-to-one two-sided matching markets} (marriage problem), \textit{one-to-many two-sided matching markets} (college admissions), and \textit{one-to-one non-two-sided matching markets} (roommates). In all problems, agents are endowed with exogenous preference rankings and a \textit{stable partition} of the set of agents is sought, i.e., a partition where no group of agents has an incentive to abandon their partners and match together.

For the marriage problem, Gale and Shapley designed a \textit{deferred-acceptance} algorithm to prove the existence of a \textit{stable matching} for every instance. Their algorithm takes one of the sides of the market, called the proposer-side, and asks its agents to propose to their most preferred option that has not rejected them yet. Agents receiving more than one proposal accept the best one and reject all the others. The algorithm continues until all agents on the proposer side have been accepted or rejected by all. The computation of the stable matching is exact and takes at most $\mathcal{O}(N^2)$ iterations with $N$ being the size of the largest set. An alternative directed graph approach to the marriage problem was designed by Balinski and Ratier \cite{balinski_stable_1997,balinski_graphs_1998}, where they characterized the stable matching polytope in the marriage problem through linear inequalities, proving that any feasible point of the polytope is a stable matching and vice-versa (Rothblum \cite{rothblum_characterization_1992} was the first one to show this characterization).

\textit{One-to-many two-sided matching markets} are the extension of the marriage problem to the case in which agents on one of the sides can be matched with many partners at the same time. Gale and Shapley proved that their same deferred-acceptance algorithm could be applied to the one-to-many setting. Baïou and Balinski \cite{baiou_stable_2000,baiou_student_2004} generalized the graph-theoretic approach and the polytope characterization of the set of stable matching in \cite{balinski_stable_1997,balinski_graphs_1998} to one-to-many matching markets. Echenique and Oviedo \cite{echenique_core_2002} characterized the set of core stable allocations as fixed points of a map. In their model, agents are endowed with strict preferences and their characterization gives an efficient algorithm to compute stable allocations. No extra assumption is required for their characterization, but substitutability is required for the non-emptyness of the core.

Many interesting applications arise from the one-to-many setting. Gimbert et al. \cite{gimbert_constrained_2021} studied a school choice problem with imperfect information in which students reveal only a partial version of their preferences due to a limited number of applications allowed. Correa et al. \cite{correa_school_2019} studied a centralized mechanism to fairly allocate students to schools in Chile giving priority to joined siblings allocation. In France, extensive studies have been done to develop the students' allocation mechanism to universities \textit{Parcoursup} (a french related document can be found \href{https://www.irif.fr/_media/users/claire/presentation_algorithmes_parcoursup.pdf}{here}).

In \textit{one-to-one non-two-sided matching markets}, a partition of the set of agents into pairs must be found. Unlike the previous models, existence of stable allocations cannot be guaranteed, as illustrated by Gale and Shapley who gave a small roommates problem instance without stable allocation. Whenever a solution exists, Knuth \cite{knuth_marriages_1976} proved that roommates problems can have multiple stable allocations and asked 
the question of designing an algorithm to compute one whenever it exists. Irving \cite{irving_efficient_1985} answered with an efficient algorithm to compute a solution if it exists or to report the non-existence of stable allocations for those instances without a solution. Tan \cite{tan_necessary_1991} found necessary and sufficient conditions for the existence of stable matchings in the roommates problem: the non-existence of \textit{stable partitions} with \textit{odd parties}.

Since the seminal paper of Gale and Shapley, several articles extended the different models to the \textit{endogenous preferences} setting. The first extension of the marriage problem is the \textit{assignment game} of Shapley and Shubik \cite{shapley_assignment_1971} in which agents within the same couple can make monetary transfers. The leading example is a \textit{housing market} where buyers and sellers have quasi-linear utilities. Allocations in the Shapley-Shubik model are stable if there is no unmatched pair buyer-seller and no transaction price such that both agents end up strictly better off by trading. Exploiting the linearity of the payoff functions on the monetary transfers, Shapley and Shubik found stable solutions for their problem using linear programming where a pair primal-dual gives, respectively, the matching and the utility vectors. Remark the polynomial complexity of solving the assignment game thanks to the linear programming approach.

The assignment game belongs to the class of \textit{cooperative games with transferable utility} as agents within the same couple have to split their \textit{worth} in such a way nobody prefers to change their partner. Moreover, Shapley and Shubik proved that the set of stable allocations for their assignment game is exactly the \textit{Core} of the housing market problem seen as a transferable utility cooperative game. Demange and Gale \cite{demange_strategy_1985} considered more general utility functions on monetary transfers (non-quasi-linear) and allowed monetary transfers on both sides (from buyer to seller and vice-versa). Demange et al. \cite{demange_multi-item_1986} designed two \textit{ascending price mechanisms} to compute stable allocations of the \textit{matching with transfers} model in \cite{demange_strategy_1985}. For integer utilities, the first algorithm converges in a bounded number of iterations to an exact solution. For continuous payments, the second algorithm converges to an $\varepsilon$-stable solution in a bounded number of iterations $T \propto \frac{1}{\varepsilon}$.

As Shapley and Shubik in the one-to-one case, Crawford and Knoer \cite{crawford_job_1981} extended the model of one-to-many two-sided matching markets to the linear monetary transfer setting. Kelso Jr. and Crawford \cite{kelso_jr_job_1982} went further in the extension by considering any kind of transferable utility. Their \textit{job matching model} considers workers and firms that get matched and simultaneously determine salaries to be paid to the workers. The authors proved the existence of stable allocations for any setting in which workers are \textit{gross substitutes} for the firms: increasing the salary of a set of workers can never cause a firm to withdraw an offer from a worker whose salary has not been risen.

A seminal paper in one-to-many matching markets was written by Hatfield and Milgrom \cite{hatfield_matching_2005}, the \textit{matching with contracts} model, that extends the model of Kelso and Crawford by allowing doctors and hospitals (instead of workers and firms) to sign \textit{contracts} from a finite set of possible contracts in the market. Contracts are \textit{bilateral} so each of them relates one doctor with one hospital. Agents are endowed with preference orderings that define \textit{choice functions}. Given a set of possible contracts, the choice functions output the most preferred contract for each doctor, and the most preferred subset of contacts for each hospital. Hatfield and Milgrom proved that the set of stable allocations is a non-empty lattice and that a \textit{cumulative offer mechanism} reaches the extremes of the lattice thanks to Tarski's fixed point theorem. The main assumption behind this result is \textit{substitutability} for hospitals, i.e., no previously rejected contract can be chosen by a hospital because of the broadening of the set of contracts. Substitutability has been proved to be sufficient but not necessary for the existence of stable allocations in the matching with contracts model and many authors have worked to find weaker assumptions \cite{aygun_matching_2012,hatfield_substitutes_2010,hatfield_stability_2013,hatfield_stability_2021}. Aygün and Sönmez \cite{aygun_matching_2012} exposed that different models are obtained if agents' choice functions are treated as primitives or they are induced from preference rankings in the matching with contracts model. Hatfield and Milgrom's model belongs to the second type, however, they treated their choice functions as primitives. To truly guarantee the existence of stable allocations, an extra assumption, namely, the \textit{irrelevance of rejected contracts}, is required as well.

\textit{Roommates with transferable utility} has received attention lately \cite{andersson_competitive_2014,chiappori_roommate_2014,eriksson_stable_2001,klaus_consistency_2010,talman_model_2011}. Andersson et al. \cite{andersson_competitive_2014} designed a \textit{price adjustment process} that computes, under integral payments, a stable allocation or disproves its existence in finite time. Shioura \cite{shioura_partnership_2017} made the connection between roommates with transferable utility and the assignment game of Shapley and Shubik \cite{shapley_assignment_1971}. More precisely, Shioura reduced his problem to a particular assignment game in an auxiliary bipartite graph and proposed an extension of the algorithm of Andersson et al. \cite{andersson_competitive_2014} to compute a stable allocation in case of existence. Up to our knowledge, the only extension of the roommates problem to the \textit{non-transferable utility} domain has been made by Alkan and Tuncay \cite{alkan_pairing_2014}. 

A clear comparison between the models with transferable and non-transferable utility in different settings was made by Echenique and Galichon \cite{echenique_ordinal_2017}. Other extensions have considered \textit{imperfect transferable utility} due to, for example, the presence of taxes in the transfers. Galichon et al. \cite{galichon_costly_2019} algorithmically proved the existence of stable solutions in this setting. 

Recently, Garrido-Lucero and Laraki \cite{garrido2025stable} introduced \textit{matching games}, a novel one-to-one matching model where doctors and hospitals are matched and agents’ outcomes within couples result from playing strategic two-player games, simultaneously to the moment of getting matched. Matching games encompass many of the studies in the stable matching literature and, unlike most utility-driven approaches, they analyzed the strategies that support stable outcomes. By running a \textit{deferred-acceptance with competitions} algorithm, an adaptation of Gale-Shapley's where agents proposing to the same partner compete as in a second-price auction, Garrido-Lucero and Laraki proved the existence of Nash-pairwise stable allocations for matching games without commitment and pairwise stable allocations in matching games with commitment, both under mild assumptions over the agents' strategy sets and payoff functions. 

Among their key concepts and results, Garrido-Lucero and Laraki introduced a refinement of pairwise stable allocations called \textit{renegotiation-proofness}. Their existence is not always guaranteed. Garrido-Lucero and Laraki introduced the concept of \textit{feasible games}, showed its necessity for existence,  proved that several standard classes of two-player strategic games are feasible, and designed a \textit{renegotiation process} (similar to the one by Rochford in \cite{rochford_symmetrically_1984}) that converges to a pairwise stable and renegotiation-proof allocation whenever all the pairwise games are feasible.

The present article is devoted to extend the model of matching games to a general matching games model, able to cover the work in \cite{garrido2025stable} as well as both the one-to-many two-sided setting and the one-to-one non-two-sided setting.

General matching games present new challenges for the existence of stable allocations and the design of algorithms to compute them. However, in addition to the models already captured by one-to-one matching games, general matching games get a much broader class of models in stable matching markets, being among the most important, the \textit{matching job market} model of Kelso and Crawford \cite{kelso_jr_job_1982}, the \textit{matching with contracts} model of Hatfield and Milgrom \cite{hatfield_matching_2005}, \textit{hedonic games} \cite{dreze_hedonic_1980}, the \textit{roommates problem} of Gale and Shapley \cite{gale_college_1962,irving_efficient_1985,knuth_marriages_1976}, the \textit{roommates problem with transferable utility} \cite{andersson_competitive_2014,chiappori_roommate_2014,eriksson_stable_2001,klaus_consistency_2010,talman_model_2011}, and the \textit{roommates problem with non-transferable utility} of Alkan and Tuncay \cite{alkan_pairing_2014}.
\medskip

\noindent\textbf{Contributions}. We introduce the general matching games model where several doctors can be allocated to the same hospital and doctors' payoffs may depend on both colleagues (other doctors at the same hospital) and the hospital. In order to establish existence of stable allocations, we focus on two submodels, namely, the one-to-many additive separable setting and the roommates setting. For each submodel, we give sufficient and necessary conditions for \textit{core stable} and \textit{renegotiation-proof} allocations to exist. In particular, we show how to map the work of Hatfield and Milgrom \cite{hatfield_matching_2005} to our fist submodel, showing how the way we treat choice functions impact the posterior required assumptions for existence of stable outcomes. Finally, we perform the complexity study of all the algorithms whenever players play finite games in mixed strategies, showing they are polynomial on the number of agents and pure strategies per player. 
\medskip

\noindent\textbf{Outline}. The article is structured as follows. \Cref{sec:model_one_to_many_matching_games} presents the general matching games model, gives several examples, and introduces the notion of core stability. \Cref{sec:matching_with_contracts_matching_games} studies the existence of core stable allocations in the first submodel and shows how to map the \textit{matching with contracts} model to ours. \Cref{sec:roommates_matching_games} studies roommates matching games and establish necessary conditions for core stable allocations to exist. \Cref{sec:internal_stability_one_to_many} extends renegotiation proofness to the same two submodels and uses the renegotiation process of \cite{garrido2025stable} to prove their existence. \Cref{sec:complexity} introduces \textit{bi-matrix matching games}m exposes the main computational challenges of the different algorithms applied in the two submodels, and discuss the complexity of finding stable allocations in the matching games without commitment model also introduced by Garrido-Lucero and Laraki. \Cref{sec:complexity_zero_sum_games} conducts the complexity study for zero-sum matching games. \Cref{sec:conclusions} concludes the article. Appendix \ref{sec:complexity_infinitely_repeated_games} conducts the complexity study for infinitely repeated matching games while Appendix \ref{sec:complexity_strictly_competitive_games} for strictly competitive matching games.

\section{General matching games}\label{sec:model_one_to_many_matching_games}

In this section, we present the mathematical model of general matching games, give several examples of models that can be mapped into a one-to-many matching game, and introduce the notion of Core stability.

\subsection{Mathematical model}

Consider two finite sets of agents $D$ and $H$ that we refer to as \textbf{doctors} and \textbf{hospitals}, respectively. Simultaneously to getting matched, agents may play actions from a given set of actions and receive some utilities depending on the matching and the action profiles chosen. Formally, every doctor $d \in D$ and every hospital $h \in H$ is endowed with a \textbf{set of strategies} $X_d$ and $Y_h$, respectively, subsets of a topological space. Given $I \subseteq D$, we denote $X_I := \prod_{d \in I} X_d$ and $Y^{|I|}_h$ to $h$'s strategy power set. For the sick of simplicity, we will omit the parenthesis $|\cdot|$ and write $Y^{I}_h$ to refer to the power set.

Doctors choose only one strategy from their strategy sets (e.g. the number of weekly hours they want to work), while hospitals pick as many strategies as assigned doctors have (e.g. the salary of each doctor). The payoff of a hospital depends on the strategies played by all its doctors as well as the ones played by the hospital itself. Doctors' payoffs, in exchange, depend on their strategy, the particular strategy that their hospital plays against them, and the strategies of all the other doctors allocated in the same hospital (from now on, their colleagues). Formally, given a set of doctors $I$ allocated in a hospital $h$, the agents' payoff functions are, for any $d \in I$, $f_{d,I,h} : X_d \times X_I \times Y_h \to \mathbb{R}$ and $g_{I,h} : X_{I} \times Y_h^{I} \to \mathbb{R}$. We remark the redundancy on $d$'s payoff function as $d$ belongs to $I$. To avoid the overcharged but more accurate notation $f_{d,I\setminus\{d\},h}$, we prefer to make an abuse of notation.

Players are endowed with \textbf{individually rational payoffs} (IRPs). Formally, for every doctor $d \in D$ (resp. hospital $h \in H$), we consider a value $\underline{f}_d \in \mathbb{R}$ (resp. $\underline{g}_h \in \mathbb{R}$), representing the utility for being unmatched. 

\begin{definition}\label{def:matching_game_one_to_many}
A (one-to-many) \textbf{matching game} is given by
$$\Gamma = \left(D,H,(X_d)_{d \in D}, (Y_h)_{h \in H}, (\underline{f}_d)_{d \in D}, (\underline{g}_h)_{h \in H}\right),$$
and payoff functions given as above.
\end{definition}

We extend the notion of allocation given in \cite{garrido2025stable} to the one-to-many setting. In this new setting, a matching $\mu$ corresponds to any mapping between $D$ and $H$. We extend the sets $D_0 := D \cup \{d_0\}$ and $H_0 := H \cup \{h_0\}$ by adding the empty players $d_0$ and $h_0$, such that, any agent matched with one of them gets her IRP as payoff. 

\begin{definition}\label{def:allocation_one_to_many}
An \textbf{allocation} corresponds to any triplet $\pi = (\mu,\vec{x},\vec{y})$ in which, $\mu : D \to H$ is a \textbf{matching}, $\vec{x} \in X_D$ is a \textbf{doctors' strategy profile}, and $\vec{y} := (\vec{y}_h)_{h \in H}$ is a profile of \textbf{hospitals' strategy profiles}, where each $\vec{y}_h \in Y_h^{\mu(h)}$.
For a coalition $(I,h) \in \mu$, where $I \subseteq D$ and $h \in H$, we write indistinctly $I = \mu(h)$ and $\mu(d) = h$, for any $d \in I$. For a doctor $d$ allocated in hospital $h$, we use $\mu(h)$ to denote $d$'s colleagues (we remark, again, the abuse of notation as we consider $d$ inside of her colleagues set).
\end{definition}

Given an allocation $\pi = (\mu,\vec{x}, \vec{y})$, $h \in H$, and $d \in \mu(h)$, their utilities at $\pi$ correspond to,
\begin{align*}
    f_d(\pi) &:= f_{d,\mu(h),h}(x_d,\vec{x}_{\mu(h)},y_{d,h}) = f_{d,\mu(h),h}(x_d,(x_{d'})_{d' \in \mu(h)}, y_{d,h}),\\
    g_h(\pi) &:= g_{\mu(h),h}(\vec{x}_{\mu(h)}, \vec{y}_{\mu(h),h}) = g_{\mu(h),h}((x_{d'})_{d' \in \mu(h)}, (y_{d',h})_{d' \in \mu(h)}).
\end{align*}

Doctors' utility depends on the identity of their hospital, the identity of their colleagues, their own strategy, the strategies of their colleagues, and the particular strategy played by the hospital against them. Hospitals' utility depends on the identity of their doctors, the strategies of these doctors, and the strategies played by the hospitals against each of their doctors. Let us illustrate the model with the following examples.

\begin{example}\label{ex:segregation_problem}
Consider a set of $2n$ doctors $D = \{1,...,2n\}$, each of them with a ranking $r_d \in \mathbb{N}$, and suppose that $r_1 \bbi r_2 \bbi ... \bbi r_{2n}$ (the higher the ranking, the better the doctor). Consider a set of two hospitals $H = \{h_1,h_2\}$, and that each of them has a prestige $p_{h} \in \mathbb{N}$, with $p_1 \bbi p_2$ (the higher the prestige, the better the hospital). Take all strategy sets empty. Given a matching $\mu$, $h \in H$ and $d \in \mu(h)$, agents' payoffs are given by $f_d(\mu) := p_{h} + \sum_{d' \in \mu(h)} r_{d'}$ and $g_h(\mu) = \sum_{d' \in \mu(h)} r_{d'}\cdot \mathbb{I}\{|\mu(h)| \leq q_h\} + \infty\cdot \mathbb{I}\{|\mu(h)| \bbi q_h\}$, where $q_1,q_2 \in \mathbb{N}$ are fixed and known values, representing the hospitals' capacities. In words, doctors' payoffs correspond to the aggregated ranking of their colleagues and their own ranking (the higher, the better) plus the prestige of their hospital. For hospitals, their utility is given by the aggregated ranking of their doctors up to a capacity $q_h$. \qed
\end{example}

\begin{example}\label{ex:roommates}
Consider a set $D$ of $n$ doctors, a set $H$ of $m$ hospitals, with $m \bbi \lfloor n/2 \rfloor$, suppose that all strategy sets are empty (as in the previous example), and all agents have null IRPs. Given $I \subseteq D$, $d \in D$, and $h \in H$, suppose the payoff functions are given by $g_{I,h} \equiv 0$ and $f_{d,I,h} = v_{d,d'} \cdot \mathbb{I}\{I = \{d'\}\} - M \cdot \mathbb{I}\{|I| \bbi 1\}$, where $(v_{d,d'}, d,d' \in D)$ are fixed positive real values and $M \gg 1$. In words, doctors receive a positive utility if they have only one colleague and a negative utility otherwise. In addition, hospitals have null payoff functions independent of the doctors assigned to them, and doctors' utilities do not depend on the hospital they are assigned to. Since there are more hospitals than half of the doctors, doctors in groups of more than two agents can always opt by matching in couples in order to end up better off. The resulting matching game corresponds to a classical \textit{roommates problem with exogenous preferences}. \qed
\end{example}

\begin{example}\label{ex:hedonic_game_example} 
Consider a set of three doctors $D = \{1,2,3\}$ and two hospitals $H = \{a,b\}$, all agents with empty strategy sets. \Cref{tab:hedonic_game_example} shows the utility of each doctor for belonging to each possible subset of doctors, independent of the assigned hospital (doctors get utility only in the coalitions they appear in).
\vspace{-0.5cm}
\begin{table}[H]
    \centering
    \begin{tabular}{cc|ccccccc}
    &\multicolumn{8}{c}{Subsets of doctors}\\\noalign{\vskip 2mm}
    & & \{1\} & \{2\} & \{3\} & \{1,2\} & \{1,3\} & \{2,3\} & \{1,2,3\}\\
    \cline{2-9}
    \multirow{3}{*}{Doctors} & 1 & 0 & - & - & 1 & -2 & - & -1 \\
    & 2 & - & 0 & - & 1 & - & -2 & -1 \\
    & 3 & - & - & 0 & - & 1 & 2 & 3
    \end{tabular}
    \caption{Hedonic game.}
     \label{tab:hedonic_game_example}
\end{table}
Endow hospitals with null payoff functions. Remark that agents $1$ and $2$ prefer to be together rather than being single, and none of them wants to be with $3$. The resulting matching game corresponds to an \textit{hedonic game}. \qed
\end{example}

\subsection{Core stability}\label{sec:external_stability_one_to_many}

This section is devoted to extend the pairwise stability notion studied in \cite{garrido2025stable} to the model of general matching games. As pairwise stability for one-to-one matching games, \textit{Core stability} will generalize the stability notions from the literature on one-to-many matching markets, in particular, the one of Hatfield and Milgrom \cite{hatfield_matching_2005}. We consider the $\varepsilon$-version as it will be useful for the subsequent complexity study.

\begin{definition}\label{def:external_stability_one_to_many}
Let $\varepsilon \geq 0$ be fixed. An allocation $\pi = (\mu,\vec{x},\vec{y})$ is $\varepsilon$-\textbf{blocked} by a coalition $(I,h)$, with $I \subseteq D$ and $h \in H$, if there exist $(\vec{w}_I,\vec{z}_{I,h}) \in X_I \times Y_h^I$, such that, for any $d \in I$, $f_{d,I,h}(w_d,\vec{w}_I,z_{d,h}) \bbi f_d(\pi) + \varepsilon$, and $g_{I,h}(\vec{w}_I,\vec{z}_{I,h}) \bbi g_h(\pi) + \varepsilon$.
The allocation $\pi$ is $\varepsilon$-\textbf{core stable} if it is $\varepsilon$-\textbf{individually rational} (no agent gets $\varepsilon$ less than her IRP) and it is not $\varepsilon$-blocked. 
\end{definition}

Most of the discussions of \Cref{sec:matching_with_contracts_matching_games,sec:roommates_matching_games} will consider $0$-core stable allocations, simply called, core stable, unless precised otherwise.

An allocation is core stable if (1) agents' IRPs are satisfied, so nobody prefers to abandon their partners and become single, and (2) we cannot find a set of doctors and a hospital who prefer to abandon their assigned partners and matching together as they end up better off. 
The coalition $(I,h)$ in the previous definition is called \textbf{blocking coalition}.


Core stability captures \textit{stability} in the matching with contracts \cite{hatfield_matching_2005} and roommates setting \cite{irving_efficient_1985}, \textit{pairwise stability} in the stable marriage problem \cite{gale_college_1962} and one-to-one matching games \cite{garrido2025stable}, and \textit{core stability} in the assignment game \cite{shapley_assignment_1971}, matching with transfers \cite{demange_strategy_1985}, and hedonic games settings \cite{dreze_hedonic_1980}. Core notions from economy \cite{hildenbrand_core_1982} are also captured by our stability notion when fixing the matching $\mu$. Let us compute the core stable allocations of our examples.
\vspace{0.5cm}

\noindent\textbf{\Cref{ex:segregation_problem}.} Recall the example with $2n$ doctors ordered by ranking and two hospitals ordered by prestige. Suppose that hospitals have capacities $q_1 = q_2 = n$ and that agents have low IRPs, so all agents prefer to be matched rather than being single. Let $\pi$ be an individually rational allocation such that a doctor $d \in \{n+1,...,2n\}$ is matched with $h_1$. As the capacities of both hospitals are equal to $n$, there must be a doctor $d' \in \{1,...,n\}$ matched with $h_2$. Note that the coalition $(I,h_1)$, with $I = \{1,...,n\}$, blocks $\pi$ as all doctors in $I$ increase strictly their payoffs if they match with $h_1$, as well as $h_1$. Thus, for $\pi$ to be core stable, it must hold that the best $n$ doctors are assigned to $h_1$, and the $n$ worst to $h_2$. Therefore, we obtain that the only core stable allocation for this matching game is the one that \textit{segregates} the doctors by their rankings. \qed
\medskip

\noindent\textbf{\Cref{ex:roommates}.} Recall the roommates example in which a set of $n$ doctors seek to get matched in couples. A matching $\mu$ is core stable if there exists no pair of doctors $(d_1,d_2) \in D \times D$, such that, $v_{d_1,d_2} \bbi f_{d_1}(\mu) \text{ and } v_{d_2,d_1} \bbi f_{d_2}(\mu)$. In particular, as there are more hospitals than half of the doctors, any pair of doctors can form a couple in case of being matched with a bigger group. A matching $\mu$ is core stable if and only if it is stable for the roommates problem (assuming, without loss of generality, that $v_{d,d'} \neq v_{d,d''}$, for any $d,d',d''$ in $D$).\qed
\medskip

\noindent\textbf{\Cref{ex:hedonic_game_example}.} Recall the hedonic game example in which three doctors seek to get matched in coalitions and get payoffs as shown in \Cref{tab:hedonic_game_example}. A partition $\mu$ of the set of doctors is core stable if there is no set of doctors $I$, not matched between them, such that matching together, all of them end up better off. Therefore, the core stable allocations are the partitions $\mu_1 = (\{1,2,a\},\{3,b\})$ and $\mu_2 = (\{1,2,b\},\{3,a\})$, equivalent in terms of utility, which correspond to the \textit{core stable partition} of the hedonic game when dropping the hospitals from the allocation. \qed
\medskip

In this article we aim to design efficient algorithms to find core stable and renegotiation proof allocations, if they exist, or to report the non-existence of these allocations when the problem does not allow them. In order to have any hope of designing tractable algorithms, we will focus on two particular models: \textit{(one-to-many) additive separable matching games} and \textit{roommates matching games}, being both submodels of general matching games.

\section{1-to-Many additive separable matching games}\label{sec:matching_with_contracts_matching_games}

The additive separable matching games submodel rises as a generalization to Hatfield and Milgrom's work \cite{hatfield_hidden_2015}. Therefore, we begin this section by recalling their model and results.

\subsection{Matching with contracts model}

Matching with contracts, defined by Hatfield and Milgrom (H\&M) \cite{hatfield_matching_2005}, considers two finite sets $D$ and $H$ of doctors and hospitals, respectively, and a finite set of contracts $X$. Contracts are \textit{bilateral}, so each $x \in X$ is related to only one doctor $x_D \in D$ and one hospital $x_H \in H$. Agents have \textit{choice functions} $(C_d, C_h, \forall d \in D, \forall h \in H)$, such that, given a set of contracts $X' \subseteq X$, $C_d(X')$ outputs the most preferred contract of $d$ among those in $X'$, and $C_h(X')$ outputs the most preferred subset of contracts of $h$ among those in $X'$.

A set of contracts is stable if no agent prefers to abandon her contract and no coalition of doctors can get matched with a hospital weakly improving everybody's allocation with at least one doctor and the hospital ending up strictly better off.

\begin{definition}\label{def:stable_allocation_matching_with_contracts}
\textup{A set of contracts $X' \subseteq X$ is \textbf{stable} if, $\bigcup_{d \in D}$ $C_d(X') = \bigcup_{h \in H} C_h(X') = X'$ and there is no $h \in H$ and $X'' \neq C_h(X')$, such that, $X'' = C_h(X'\cup X'') \subseteq \bigcup_{d \in D} C_d(X' \cup X'')$.}
\end{definition}

Due to the way that H\&M treated the hospitals' choice functions, these must satisfy two assumptions for the existence of stable allocations: \textit{substitutability} (\Cref{def:substitutability_matching_with_contracts} or page 918 \cite{hatfield_matching_2005}) and \textit{irrelevance of rejected contracts} (IRC) (\Cref{def:substitutability_matching_with_contracts}, or page 6 \cite{aygun_matching_2012}). 

\begin{definition}\label{def:substitutability_matching_with_contracts}
\textup{Contracts are \textbf{substitutes} for hospital $h$ if for any $X' \subseteq X$ and any pair of different contracts $x,x' \in X$, if $x \notin C_h(X')$, then $x \notin C_h(X' \cup \{x'\})$. The choice function of hospital $h$ satisfies the \textbf{irrelevance of rejected contracts} (IRC) if for any $Y \subseteq X$, and any $z \notin Y$, if $z \notin C_h(Y \cup \{z\})$ then $C_h(Y) = C_h(Y \cup z)$.}
\end{definition}

Using fixed-point techniques from lattice theory that allowed to guarantee the convergence of a \textit{cumulative offer mechanism} (COM), H\&M proved that the set of stable allocations is a non-empty lattice. We claim that an easier existence proof can be conducted under substitutability, as a deferred-acceptance-like of algorithm is enough to compute a stable allocation. Our claim lies in the following result, unobserved in \cite{hatfield_matching_2005}.

\begin{proposition}\label{prop:stability_matching_with_contracts_under_substitutability}
Let $X' \subseteq X$ be an allocation and suppose that contracts are substitutes for hospitals and the IRC property is satisfied. $X'$ is stable if and only if it is individually rational and there is no doctor $d$, hospital $h$, and contract $x \in X\setminus X'$, with $x_D = d, x_H = h$, such that $x = C_d(X' \cup\{x\})$ and $x \in C_h(X' \cup \{x\})$.
\end{proposition}

Remark that the second property of \Cref{prop:stability_matching_with_contracts_under_substitutability} corresponds to \textbf{pairwise stability}, adapted to the matching with contracts setting. 

\begin{proof}
Suppose $X'$ is stable. Let $d$ be a doctor and $x,x' \in X$ be two contracts related to doctor $d$ ($x_D = x'_D = d$), being $x'$ her contract in the allocation ($x' \in X'$) and $x$ any other contract ($x \in X \setminus X'$), such that $d$ prefers $x$ to $x'$ ($x = C_d(X' \cup \{x\})$). Let $h = x_H$ be the hospital related to contract $x$. Set $X'' = (X'\setminus\{x'\}) \cup\{x\}$ as the allocation obtained when replacing $x'$ by $x$ and consider $X''|_h := \{y \in X'' : y_H = h\}$. Note that $X''|_h = X'|_h \cup \{x\}$. It follows that $X''|_h \subseteq \bigcup_{d \in D} C_d(X' \cup X''|_h)$ as the only change is the new contract $x$, the one is preferred by $d$ over $x'$. Since $X'$ is stable, it follows that,
\begin{align*}
 X''|_h \neq C_h(X' \cup X''|_h) &\Longleftrightarrow [X'|_h \cup \{x\}] \neq C_h(X' \cup \{x\}) \Longleftrightarrow  x \notin C_h(X' \cup \{x\}).
\end{align*}

Conversely, suppose $X'$ is individually rational and there is no doctor $d$, hospital $h$, and contract $x \in X\setminus X'$, with $x_D = d, x_H = h$, such that $x = C_d(X' \cup\{x\})$ and $x \in C_h(X' \cup \{x\})$. Suppose $X'$ is not stable, so there exist $h$ a hospital and a set of contracts $X'' \neq C_h(X')$ such that, $X'' = C_h(X'\cup X'') \subseteq \bigcup_{d \in D} C_d(X' \cup X'')$. Let $x \in X'' \setminus C_h(X')$ and $d$ the related doctor to $x$ ($x_D = d$). Since $x \in C_h(X' \cup X'')$, we have that $x$ belongs to $C_h(Z)$ for any $Z\subseteq X' \cup X''$, such that $x \in Z$ (substitutability + IRC). In particular, $x \in C_h(X' \cup \{x\})$. Since $x = C_d(X' \cup \{x\})$, as doctors choose only one contract, we obtain a contradiction. 
\end{proof}

Pairwise stable allocations can be computed by a deferred-acceptance-like algorithm, such as \Cref{Algo:Deferred_acceptance_HM}. In the pseudo-code, we have extended the set of contracts $X$ to $X_0 := X \cup \{\varnothing\}$, where $\varnothing$ represents being unmatched. In addition, given a set of contracts $Y\subseteq X$, and $k$ an agent, we denote by $Y|_k$ the subset of contracts in $Y$ that are related to $k$.

\begin{algorithm}[h]
\textbf{Input}: $(D,H,X_0)$ a matching with contracts instance.

Set $D' \leftarrow D$ as the set of unmatched doctors.

Set $Y,Z \leftarrow \varnothing$ as the sets of accepted and rejected contracts, respectively.

\While{$D' \neq \emptyset$}{
Let $d \in D'$ and $x \in \argmax C_d(X_0 \setminus Z)$. 

If $x = \varnothing$, update $D' \longleftarrow D' \setminus \{d\}$ and start again. Otherwise, let $h  = x_H$ be the concerned hospital.

Let $W = C_h(Y \cup \{x\})$ and $\overline{W} = (Y|_h \cup \{x\}) \setminus W$. Update $Y \leftarrow (Y \setminus \overline{W}) \cup W$ and $Z \leftarrow Z \cup \overline{W}$.

Update $D'$ as it corresponds.
}
Output $Y$
\caption{Deferred-acceptance with contracts algorithm.}
\label{Algo:Deferred_acceptance_HM}
\end{algorithm}

\Cref{Algo:Deferred_acceptance_HM} takes an unmatched doctor and asks her to propose her most preferred contract among all those that have not been already rejected. The concerned hospital then chooses its most preferred subset of contracts among all the proposed ones. The rejected contracts are stored, the set of unmatched doctors is updated accordingly, and a new iteration starts. The algorithm keeps running until all doctors have been allocated to a hospital or the non-allocated ones prefer to remain unmatched.

\Cref{Algo:Deferred_acceptance_HM} always converges as the set of contracts is finite and the set of rejected contracts is increasing. Moreover, the output $Y$ is \textbf{always pairwise stable} as for any doctor $d$ and any contract $x \in X\setminus Y$, with $x = C_d(Y \cup \{x\})$, necessarily $x$ was rejected along the algorithm. Therefore, $x \notin C_h(Y\cup \{x\})$, where $h = x_H$. From \Cref{prop:stability_matching_with_contracts_under_substitutability}, we conclude that \Cref{Algo:Deferred_acceptance_HM} outputs a  stable allocation whenever hospitals have substitute contracts and satisfy the irrelevance of rejected contracts property. Observe that, even if one of the conditions fails to be satisfied, a pairwise stable allocation always exists for any instance of the H\&M model. 

\subsection{Additive separable matching games}\label{sec:model}

Let $\Gamma$ be a matching game (\Cref{def:matching_game_one_to_many}) and relax the doctors' payoff functions dependence on colleagues, that is, redefine the payoff functions as $f_{d,h} : X_d \times Y_h \to \mathbb{R}$, for any $d \in I$, and $g_{I,h} : X_I \times Y_h^{I} \to \mathbb{R}$, where $I \subseteq D$ and $h \in H$. The relaxation from the general matching game model to the additive separable matching game submodel corresponds to Kelso and Crawford's \cite{kelso_jr_job_1982} main assumption.

We will endow hospitals with \textbf{quotas}, representing the maximum number of doctors they can receive. Formally, we say hospital $h\in H$ has quota $q_h \in \mathbb{N}$, if for any $I \subseteq D$ and strategy profiles $\vec{x}_I \in X_I, \vec{y}_h \in Y_h^I$, their payoff function becomes:
\begin{align*}
    \Tilde{g}_{I,h}(\vec{x}_I,\vec{y}_h) = \left\{\begin{array}{cc}
      g_{I,h}(\vec{x}_I,\vec{y}_h)  & \text{if } |I|\leq q_h  \\
      -\infty   & \text{otherwise}
    \end{array} \right.
\end{align*}

As individual rationality of the hospitals will exclude allocations that don't respect the quotas, by abuse notation, \textbf{we denote $\Tilde{g}$ also by $g$}.

\begin{remark}
A model without quotas can be easily obtained by taking $\vec{q} \equiv |D|$ for all hospitals. Similarly, taking $\vec{q} \equiv 1$ we recover the one-to-one matching models in the literature \cite{demange_multi-item_1986,gale_college_1962,shapley_assignment_1971}.
\end{remark}

Next, we introduce the notion of additive separability.

\begin{definition}\label{def:add_separable_with_quotas}
Let $h$ be a hospital. We say that $h$'s payoff function is \textbf{additive separable} if for every doctor $d \in D$ there exists a function $g_{d,h}: X_d \times Y_h \to \mathbb{R}$, such that
\begin{align*}
    g_{h}(\pi) = \left\{\begin{array}{cc}
    \sum_{d \in \mu(h)} g_{d,h}(x_d,y_{d,h}) & \text{ if } |\mu(h)|\leq q_h,  \\
    -\infty     & \text{ otherwise},
    \end{array} \right.
\end{align*}

\noindent for $\pi = (\mu,\vec{x},\vec{y})$ any allocation.
\end{definition}

Note that in additive separable matching games, every tuple 
$$G_{d,h}:=(X_d,Y_h,f_{d,h},g_{d,h}),$$
defines a two-player game (therefore, capturing the model of one-to-one matching games under commitment studied in \cite{garrido2025stable}). Strategy sets do not depend on potential partners. This is done without loss of generality and to ease the notation. In other words, agents are assumed to be able to play the same strategies against all possible partners. However, \textbf{two-player games are couple-dependent}, as agents may play different games against different partners. This is in line with the bilateral contracts of Hatfield and Milgrom \cite{hatfield_matching_2005}.

The next example shows how a matching with linear transfer problem can be mapped to an additive separable matching game (more precisely, poly-matrix constant-sum games).

\begin{example}\label{ex:multi_auction}
A set of $H$ buyers bid for the items offered by $D$ sellers. Buyers are allowed to buy as many items as they prefer while sellers have only one indivisible item to sell. Buyers $h \in H$ and sellers $d \in D$ have non-negative valuations for the items $u_{h,d}$ and $v_d$, respectively. If a buyer $h$ buys the items of a set of sellers $J \subseteq D$, paying $\vec{p}_h := (p_{h,d})_{h \in J}$, agents' payoffs are, 
$$f_{d,h}(p_{h,d}) = p_{h,d} - v_d, \text{ for any } d \in J, \text{ and } g_{J,h}(\vec{p}_h) = \sum_{d \in J} (u_{h,d} - p_{h,d}).$$
Remark that only buyers are strategic in this example while sellers only care about the highest bidder. \qed
\end{example}


\subsection{Link between matching games and the H\&M model}\label{sec:mapping_HM_and_MG}

Now that we have some examples at hand, we can formally show the links between matchings games to that of H\&M and anticipate the advantages of working with the matching games model instead. Interestingly, there are many mappings to associate a H\&M model to a matching games model, while the converse is unique.

Consider a matching with contracts model $(D,H,X)$. Endow all agents with strategy sets: for $d\in D$ and $h\in H$, we let $X_d=\{x \in X: x|_D  = d\}$ and $Y_h=\{x \in X: x|_H  = h\}$.  Under an extra assumption on the preferences in the H\&M model (such as the strong axiom of revealed preferences), there exists a utility representation \cite{chambers_revealed_2016} for each agent over $X$: $u_d$ for a doctor $d$ and $v_h$ for a hospital $h$, that we can without loss of generality assume positive. Given $I \subseteq D$, $h \in H$, and strategy profiles $(\vec{w}_I,\vec{z}_{I,h}) \in X_I \times Y_h^I$, define the payoff function of a doctor to be,
\begin{align*}
    f_{d,h}(w_d,z_{d,h}) = \left\{ \begin{array}{cl}
    u_d(w_d) & \text{if } w_d = z_{d,h} \\
    -1& \text{otherwise}
    \end{array} \right.
\end{align*}

The mapping means that if both parts agree on a contract that is available to them, they get the (positive) utility of that contract; otherwise, they get -1.  We do similarly for the hospitals.


Conversely, given a matching game in which all agents are endowed with strategy sets $(X_d,Y_h)$ and payoff functions, consider the set of contracts $Z := \bigcup_{(d,h) \in D \times H} (X_d \times Y_h)$ and associate the agents' IRPs to an empty contract $\varnothing$, included in $Z$. Then, the matching game can be uniquely mapped into a matching with contracts, where the choice functions are defined, for any $Z' \subseteq Z$ as,
\begin{align*}
    &z = C_d(Z') \text{ if and only if } f_{d,h}(z) = \max\{f_{d,h}(z') : z' \in Z', h = z_H\}, \\
    &Z'' = C_h(Z') \text{ if and only if } g_{I,h}(Z'') = \max\{g_{I,h}(Z''') : Z''' \subseteq Z', I = \{z_D \in D, z \in Z'''\} \}.
\end{align*}

Hence, computing the choice functions of hospitals may not be so easy, even in the separable additive case (they are the Argmax of the utility functions), and requires an exponential number of resources as we need to specify them for each coalition of doctors. In addition, choice functions are black boxes and do not allow us to understand the incentives behind agents' choices of contracts. 

Working with strategies and utilities solves these issues. In addition, (1) the \textit{irrelevance to rejected contracts} assumption is automatically satisfied (thanks to the strong axiom of weak preferences), (2) we can work with infinitely many contracts (finiteness is crucial in H\&M), (3) we can refine the set of stable solutions using renegotiation-proofness, (4) we can identify interesting matching problems where computing renegotiation-proof stable allocations is polynomial.

\subsection{Core stability in the 1st submodel}\label{sec:external_stability_submodel_contracts}

The core is a classical solution concept from cooperative game theory and has been widely studied \cite{aumann_core_1961,gilles_core_2010}. The non-emptyness of the core cannot be always guaranteed without extra assumptions on the game. We aim to establish sufficient conditions for the set of core stable allocations in our model to be non-empty.

Strict preferences over a finite set of contracts is a common assumption \cite{hatfield_matching_2008,hatfield_substitutes_2010,hatfield_stability_2013,hatfield_stability_2021,hatfield_matching_2005} in matching markets as it makes equivalent the concepts of \textbf{core} (\Cref{def:external_stability_one_to_many}) and \textbf{weak-core} (all agents within a coalition must end up weakly better off and at least one of them strictly better off to block an allocation). This equivalence does not necessarily hold in our model due to the continuum of payoffs. Thus, we focus in studying core stable allocations by proving that any pairwise stable allocation (\Cref{def:pairwise_stable_allocation}) is core stable under additive separability (\Cref{prop:core_stable_and_pairwise_stable_are_equivalent}). 

Remark that, generically, the equivalence holds as we can discretize the set of strategy profiles and perturb the players' payoffs to get strict preferences. Moreover, the DA algorithm (\Cref{Algo:Deferred_acceptance_HM}) can be applied to construct core stable allocations in matching games by (1) mapping the model to a matching with contracts model, (2) discretizing with mesh $\varepsilon$ our continuum set of strategy profiles, (3) perturbing the payoffs up to $\varepsilon$ to avoid ties in the preferences, (4) applying \Cref{Algo:Deferred_acceptance_HM} to find a pairwise stable allocation, and finally (5) tending $\varepsilon$ to zero. Clearly, this is too costly. Additive separability will allow us to reduce the computational complexity of this task by extending the deferred-acceptance with competition (DAC) algorithm from \cite{garrido2025stable}. 

Let $\Gamma$ be a matching game and suppose that all hospitals have additive separable payoff functions and quotas (\Cref{def:add_separable_with_quotas}).

\begin{definition}\label{def:pairwise_stable_allocation}
Let $\varepsilon \geq 0$. An $\varepsilon$-individually rational allocation $\pi = (\mu,\vec{x},\vec{y})$ is $\varepsilon$-\textbf{blocked by a pair} $(d,h)$, if there exists $({w}_d,{z}_{d,h}) \in X_d \times Y_h$, such that $f_{d,h}(w_{d},z_{d,h}) \bbi f_{d}(\pi) + \varepsilon$ and
\begin{align*}
 g_{d,h}(w_d,z_{d,h}) \bbi \left\{\begin{array}{cl}
g_{d,h}(x_d,y_{d,h}) + \varepsilon& \text{if }\mu(d) = h,\\
   \min_{d'\in \mu(h)} g_{d',h}(x_{d'},y_{d',h}) + \varepsilon& \text{if } |\mu(h)| = q_h, \mu(d) \neq h,  \\
    \varepsilon & \text{if } |\mu(h)| \sm q_h,\mu(d) \neq h,
 \end{array} \right.  
\end{align*}
$\pi$ is $\varepsilon$-\textbf{pairwise stable} if it is not $\varepsilon$-blocked by any pair.
\end{definition}

When $d = \mu(h)$, blocking means that the strategy profile used by $(d,h)$ in $\pi$ is not Pareto-optimal in their game, so they can jointly deviate to a strictly better outcome. Note that for unit hospitals quotas $q \equiv 1$, we recover a one-to-one model where pairwise stability and core stability coincide.
\medskip

\textbf{\Cref{ex:multi_auction}}. Recall the multi-item auction example. For simplicity, consider two buyers $H = \{\alpha,\beta\}$, four sellers $D = \{a,b,c,d\}$, all sellers having the same valuation $v = 1$ for their items, and all agents with a null IRP. In addition, suppose that buyers $\alpha$ and $\beta$ have valuations $u_{\alpha} = (10,10,2,2)$ and $u_{\beta} = (2,2,10,10)$, respectively, for the sellers' items. The core stable allocations corresponds to any $\pi = (\mu,(p_{\alpha},p_{\beta}))$ with $\mu = ((\alpha,a),(\alpha,b),(\beta,c),(\beta,d))$ and the buyers' strategy profiles $p_{\alpha} = (x_1,x_2,0,0)$ and $p_{\beta} = (0,0,x_3,x_4)$, with $2 \leq x_i \leq 10$ for $i\in \{1,2,3,4\}$, meaning, for example, that $\alpha$ pays $x_1$ monetary units to $a$ and $x_2$ to $b$. We obtain a continuum of core stable allocations. Remark that the two extremes $p_{\alpha} \equiv p_{\beta} \equiv 2$ and $p_{\alpha} \equiv p_{\beta} \equiv 10$ correspond to the 2nd price auction outcome and 1st price auction outcome, respectively.  \qed
\medskip

We prove, first of all, that any pairwise stable allocation is core stable. The result trivially generalizes to the $\varepsilon$-versions of core and pairwise stability.

\begin{proposition}\label{prop:core_stable_and_pairwise_stable_are_equivalent}
Let $\Gamma$ be an additive separable matching game. Then, any pairwise stable allocation $\pi = (\mu,\Vec{x},\Vec{y})$ is core stable.
\end{proposition}

\begin{proof}
Suppose $\pi$ is pairwise stable but not core stable. Let $(I,h)$ be a blocking coalition, that is, there exist $(\vec{w}_I,\vec{z}_{I,h}) \in X_I \times Y_h^I$, such that, for any $d \in I$, $f_{d,h}(w_d,z_{d,h}) \bbi f_d(\pi)$ and $g_{I,h}(\vec{w}_I,\vec{z}_{I,h}) \bbi g_h(\pi)$. Since $h$ has an additive separable payoff function, it follows,
\begin{align*}
    \sum_{d\in I} g_{d,h}(w_d,z_{d,h}) \bbi \sum_{d\in \mu(h)} g_{d,h}(x_d,y_{d,h}).
\end{align*}

In particular, there must exist $d \in I$ and $d' \in \mu(h)$ such that 
$$g_{d,h}(w_d,z_{d,h}) \bbi g_{d',h}(x_{d'},y_{d,h}).$$ 
Since doctor $d$ increases strictly her payoff with the deviation, $(d,h)$ is a blocking pair of $\pi$, which is a contradiction. 
\end{proof}

\Cref{prop:core_stable_and_pairwise_stable_are_equivalent} states that in order to prove the non-emptyness of the set of core stable allocations, it is enough with finding a pairwise stable allocation. Our characterization of core stable allocations (called stable$^*$ in \cite{echenique_core_2002}) through pairwise stable allocations is in line with the work of Echenique and Oviedo \cite{echenique_core_2002} (called stable in \cite{echenique_core_2002}). However, our algorithm does not require fixed points theorem to obtain its convergence as we do not iterate a mapping but run a deferred-acceptance like algorithm. 

To prove the existence of pairwise stable allocations, we remark the deferred-acceptance with competitions algorithm in \cite{garrido2025stable} can be directly adapted to the one-to-many setting, as illustrated by \Cref{Algo:Propose_dispose_algo_general_case}.

\begin{algorithm}[ht]
\textbf{Input}: $\Gamma$ a matching game, $\varepsilon \bbi 0$,

Set $D' \leftarrow D$ as the set of unmatched doctors and match each hospital to $d_0$

\While{$D' \neq \emptyset$}{
Let $d \in D'$ and $(h,d',x_d,y_{d,h})$ be her optimal proposal, solution of Problem (\ref{eq:optimal_proposal_problem}).

\If{$|\mu(h)| \sm q_h$}{$d$ is accepted}
\Else{$d$ and $d'$ \textbf{compete} for $h$ as in a second-price auction. The winner stays at $h$, goes out of $D'$, and the loser is included in $D'$}}
\caption{Deferred-acceptance with Competitions algorithm.}
\label{Algo:Propose_dispose_algo_general_case}
\end{algorithm}

Unlike the one-to-one setting, optimal proposals seek to replace the doctor contributing the less utility to the hospital, formally computed as,
\begin{align}\label{eq:optimal_proposal_problem}
    \max_{\substack{ h \in H_0 \\ (w,z) \in X_d \times Y_h}}\left[f_{d,h}(w,z): g_{d,h}(w,z) \geq \min_{d' \in \mu(h)} g_{d',h}(x_{d'},y_{d,h}) + \varepsilon\right],
\end{align}
given the current hospitals' payoff profile $g(\pi) = (g_h(\pi))_{h \in H}$.
The solution to Problem (\ref{eq:optimal_proposal_problem}) consists in the doctor's most preferred hospital $h$, $d'$ the doctor that $d$ desires to replace (possibly $d_0$), and $(x_d,y_{d,h})$ the strategy profile that $d$ proposes to play to $h$. 
When the optimal proposal includes $d' = d_0$, the proposer joins the hospital without replacing $d_0$ (except when a hospital reaches its quota). 

If the optimal proposal includes a doctor $d'\neq d_0$, the \textbf{competition phase} starts and both doctors play a \textit{second-price auction}. We define the \textbf{reservation payoff} of doctor $d$, $\beta_d$, (and analogously the one of $d'$) as the optimal value of the problem,
\begin{align}\label{eq:reservation_price}
    \max_{\substack{h' \in H_0\setminus\{h\} \\ (w,z) \in X_d \times Y_{h'}}}\left[f_{d,h'}(w,z): g_{d,h'}(w,z) \geq \min_{d'\in \mu(h')} g_{d',h'}(x_{d'},y_{d',h'}) + \varepsilon\right].
\end{align}
In the case of several solutions, ties are broken by favoring hospitals. 
Given $\beta_d$ the reservation payoff of $d$, her bid $\lambda_d$ (and analogously the one of $d'$) is computed by,
\begin{align}\label{eq:bid_problem}
\max \left[g_{d,h}(w,z) : f_{d,h}(w,z) \geq \beta_d, w \in X_d, z \in Y_h \right].
\end{align}

The winner is the doctor with the highest bid (in case of a tie the winner is the current partner), who reduces her proposal to match the one of the loser. Formally, if $d$ wins, she solves,

\begin{align}\label{eq:final_bid_problem}
\max\left[f_{d,h}(w,z): g_{d,h}(w,z) \geq \lambda_{d'},w \in X_d, z \in Y_h\right].
\end{align}

Ties are broken by choosing the proposal that maximizes $h$'s utility. The loser is included in $D'$ and a new proposer is chosen. The following results come directly from the one-to-one setting (Theorems 2 and 5 \cite{garrido2025stable}). Moreover, the finiteness of the algorithm will be proved as well during the complexity study. For simplicity, we omit its proof.

\begin{theorem}\label{teo:propose_dispose_algo_is_correct}
Let $\varepsilon \bbi 0$ be fixed. Then, the DAC algorithm ends in finite time and outputs an $\varepsilon$-pairwise stable allocation.
\end{theorem}


The existence of $0$-pairwise stable allocations is a consequence of \Cref{teo:propose_dispose_algo_is_correct}, passing through the compactness of the Pareto-optimal strategy sets, the continuity of payoff functions, and the finiteness of players. Therefore, from \Cref{prop:core_stable_and_pairwise_stable_are_equivalent}, we conclude the existence of a core stable allocation.

\section{Roommates matching games}\label{sec:roommates_matching_games}

The roommates problem, defined by Gale and Shapley \cite{gale_college_1962}, consists of a set $D$ of agents, each of them having strict preferences for the rest of the agents in $D$, and seeking to match in couples. Many authors \cite{andersson_competitive_2014,chiappori_roommate_2014,eriksson_stable_2001,klaus_consistency_2010,talman_model_2011} extended the model to the \textit{transferable utility} case and studied the existence of stable allocations. Up to our knowledge, Alkan and Tuncay \cite{alkan_pairing_2014} is the only article to work roommates with \textit{non-transferable utility}. Their model will be the base for our extension. We start by explaining more in detail the model in \cite{alkan_pairing_2014} keeping our language of doctors.

Let $D$ be a finite set of \textit{doctors} and $\underline{f} = (\underline{f}_d)_{d \in D} \in \mathbb{R}^{|D|}$ be their \textit{individually rational payoff} profile. For every potential couple $(d_1,d_2) \in D$, we consider a \textit{partnership function} $u_{1,2} : \mathbb{R} \to \mathbb{R}$ such that $u_{1,2}(f_{d_2})$ is the utility that agent $d_1$ achieves when her partner $d_2$ achieves $f_{d_2}$. In particular, it holds $u_{1,2}= u^{-1}_{2,1}$. The partnership functions are assumed to be \textit{continuous, decreasing}, and $u_{1,2}(\underline{f}_{d_2}) \sm \infty$, for any $d_1,d_2 \in D$.

Unlike the matching games approach that focuses on the players' strategies, Alkan and Tuncay focused on payoff profiles.

\begin{definition}
\textup{A \textbf{payoff profile} is a vector $f = (f_d)_{d \in D} \in \mathbb{R}^{|D|}$. A payoff profile is \textbf{blocked} by a couple $(d_1,d_2)$ if there exists $(f_{d_1}',f_{d_2}')\in \mathbb{R}^2$ such that $f_{d_1}' \bbi f_{d_1}$, $f_{d_2}' \bbi f_{d_2}$, and $f_{d_1}' = u_{1,2}(f_{d_2}')$. An \textbf{individually rational} payoff profile, i.e., $f_{d} \geq \underline{f}_d$ for any doctor $d$, is \textbf{stable} if it cannot be blocked.}
\end{definition}

Given a payoff profile, we focus next on finding a matching that can implement it.

\begin{definition}\label{def:realizable_payoff_profile}
\textup{A \textbf{matching} $\mu$ is a partition of $D$ in pairs and singletons. A payoff profile $f$ is \textbf{realizable} by a matching $\mu$ if $f_{d_1} = u_{1,2}(f_{d_2})$, for any $(d_1,d_2)\in \mu$, and the pair $(\mu,f)$ is called an \textbf{allocation}. Finally, a \textbf{stable allocation}  is any allocation in which the payoff profile is stable.}
\end{definition}

Alkan and Tuncay characterized the stable allocations as any allocation $(\mu,f)$ satisfying,
\begin{align}\label{eq:aspiration}
    f_d = \max\left\{\underline{f}_d, \max_{d' \neq d} u_{d,d'}(f_{d'})\right\}, \text{ for any } d \in D.
\end{align}

A payoff profile satisfying \Cref{eq:aspiration} is called an \textbf{aspiration}. Therefore, to solve the roommates with non-transferable utility problem the goal is to find an aspiration realizable by some matching.

To study when an aspiration is realizable, the authors worked with demand sets. Given a payoff profile $f$ and $d \in D$, we define $d$'s \textbf{demand set} $P_d$ at $f$ as,
\begin{align*}
    P_d(f) := \left\{d' \in D \setminus \{d\} : f_d = u_{d,d'}(f_{d'}) \right\},
\end{align*}

\noindent that is, the set of all agents with who $d$ can achieve the payoff $f_d$. Note that $d_2 \in P_{d_1}(f)$ if and only if $d_1 \in P_{d_2}(f)$. 

We define a \textbf{submarket} at $f$ as any pair of disjoint sets of doctors $(B,S) \subseteq D\times D$ such that,
\begin{itemize}
\item[1.] For any $d \in B$, $f_d \bbi \underline{f}_d$,
\item[2.] The demand set of every $B$-player is in $S$,
\item[3.] There exists a matching $\mu$ such that for every $d \in S$, $\mu(d) \in B \cap P_d(f)$, or $d$ is unmatched.
\end{itemize}
Note that a submarket $(B,S)$ is not asked to be a partition of $D$, as there may be doctors that do not belong to $B$ nor $S$. By (3), it always holds $|S|\leq|B|$. A market in which $|S| = |B|$ is called a \textbf{balanced market}. An aspiration $f$ that generates a balanced market is called a \textbf{balanced aspiration}. The following result states the existence of aspirations with balanced markets (Theorem 1, page 10 \cite{alkan_pairing_2014}). 

\begin{theorem}\label{teo:existence_balanced_aspirations}
The set of balanced aspirations is non-empty.
\end{theorem}

Moreover, Alkan and Tuncay designed a market procedure to compute balanced aspiration that converges in a polynomial number of iterations. A second theorem links balanced aspirations and stable allocations (Theorem 2, page 11 \cite{alkan_pairing_2014}).

\begin{theorem}\label{teo:equivalence_balanced_aspirations_and_stable_allocations}
The set of stable allocations is either empty or equal to the set of balanced aspirations.
\end{theorem}

To be precise, Alkan and Tuncay proved that whenever the set of balanced aspirations does not coincide with the set of stable allocations, it coincides with the set of \textit{semistable allocations}. Semistable allocations are the relaxation of stable allocations in which agents can be matched with two partners at the time and their final payoff is the average payoff obtained with her two partners. Therefore, the set of stable allocations is non-empty if and only if it coincides with the set of balanced aspirations if and only if the set of semistable allocations is empty.

The market procedure in \cite{alkan_pairing_2014} designed to compute balanced aspirations starts from any aspiration, generates a piecewise linear path of aspirations, and stops in a bounded number of steps at a balanced aspiration. At every iteration, given the current aspiration, the mechanism computes the demand sets of every agent. In case there exists a balanced submarket, the mechanism stops. Otherwise, the mechanism identifies a submarket $(B,S)$ with $|S|\sm |B|$ and alters the aspiration continuously along a suitable \textit{direction}. The direction is reset when the submarket changes.

Remark the analogy between the market procedure just explained and the \textit{increasing price mechanism} (IPM) of Andersson et al. \cite{andersson_competitive_2014}. The IPM computes the demand sets of the agents given a payoff profile, computes an over-demanded set (\Cref{def:overdemanded_set}), and increases in one unit the utility of its agents. In the market procedure of Alkan and Tuncay, $S$ is over-demanded by $B$. The suitable \textit{direction} used to alter the aspiration is computed by a second mechanism, called the \textit{direction procedure}. Starting from a submarket $(B,S)$, with $|S|\sm|B|$, the procedure computes $\lambda \bbi 0$ such that increasing the payoffs of the agents in $S$ by $\lambda$, decreasing the payoff of the agents in $B$ by $\lambda$, and letting unchanged the payoffs of the agents outside the submarket, the shifted submarket and the original one remain the same. The following result (Lemma $7$, page 29 \cite{alkan_pairing_2014}) states the correctness of the direction procedure.

\begin{lemma}
If $(B,S)$ is a bipartite submarket at an aspiration $f$, then there exists a direction $\vec{e}$ with,
\begin{align*}
e_d &\sm 0, \text{ for any } d \in B,\\
e_d &\bbi 0, \text{ for any } d \in S,\\
e_d &= 0, \text{ for any } d \notin B\cup S,
\end{align*}
such that $(B,S)$ is a bipartite submarket at $f + \lambda \vec{e}$, for all sufficiently small $\lambda \bbi 0$.
\end{lemma}

Finally, the following result (Theorem 4, page 22 \cite{alkan_pairing_2014}) states the correctness and finiteness of the market procedure.

\begin{theorem}\label{teo:market_procedure_computes_a_balanced_aspiration}
The market procedure reaches a balanced aspiration in a bounded number of steps.
\end{theorem}

The remaining question is whether a balanced aspiration can be implemented. 

\begin{definition}\label{def:overdemanded_set}
Given an aspiration $f$ and a set of doctors $I$, we define $\mathcal{D}(I) := \{d \in D: P_d(f) \subseteq I\}$ as the set of agents that demand $I$. $I$ is \textbf{overdemanded} if $|\mathcal{D}(I)| \bbi |I|$.
\end{definition}

Note that at a balanced aspiration $f$ no set of doctors is over-demanded. Therefore, considering the undirected graph Gr $= (D,E)$, with $D$ the set of doctors and $(d_1,d_2) \in E$ if and only if $d_2 \in P_{d_1}(f)$, Gr can always be decomposed as a disjoint union of cycles. Then, the aspiration $f$ is implementable by a matching whenever there exists a decomposition of Gr including only even-cycles, as for odd-cycles at least one agent will need to be matched with two partners. We remark the connection with Tan \cite{tan_necessary_1991} works in the existence of stable matchings for the roommates problem with endogenous preferences.

We finish this section by extending the model of roommates with non-transferable utility to the setting of matching games. Consider a matching game (\Cref{def:matching_game_one_to_many}),
\begin{align}\label{eq:roommates_matching_game}
    \Gamma = \left(D,H,(X_d)_{d \in D}, (f_{d,d'})_{d,d' \in D},(\underline{f}_d)_{d \in D}\right),
\end{align}
where we have taken null hospitals' payoff functions, empty hospitals' strategy sets, hospitals' quotas $\Vec{q} \equiv 2$, and we have relaxed the doctors' payoff functions dependence on the allocated hospital (as done in \Cref{ex:roommates}). Allocations become pairs $\pi = (\mu,\vec{x})$ where $\mu$ is a partition of the set $D$ into pairs and singletons and $\vec{x} \in X_D$ is a doctors' strategy profile. Given an allocation $\pi$, doctors' payoffs are given by,
\begin{align*}
    f_d(\pi) = \left\{\begin{array}{cc}
        f_{d,\mu(d)}(x_d, x_{\mu(d)}), & \text{ if $d$ is matched},  \\
        \underline{f}_d & \textit{ otherwise}.
    \end{array} \right.
\end{align*}
In case of no confusion, we will omit the set $H$ from the matching game.

\subsection{Core stability in the 2nd submodel}

Consider a roommates matching game (\Cref{eq:roommates_matching_game}). We leverage the work of Alkan and Tuncay to compute core stable allocations of our roommates matching games. Remark that core stable allocations and pairwise stable allocations are trivially equivalent in the roommates setting.


\begin{definition}\label{def:Pareto_optimal_external_stability_roommates}
Let $\pi = (\mu,\vec{x})$ be an individually rational allocation. We say that $\pi$ is \textbf{pairwise stable} if, there is no pair $(d,d')\in D\times D$, and strategy profile $(x_d,x_{d'}) \in X_d \times X_{d'}$, such that, $$f_d(x_d,x_{d'}) \bbi f_d(\pi) \text{ and } f_{d'}(x_{d'},x_d) \bbi f_{d'}(\pi).$$
\end{definition}

Since an already matched couple can block a matching in case they do not play weakly-Pareto optimally, we obtain the monotony of the partnership functions.


From the existence of balanced aspirations (\Cref{teo:market_procedure_computes_a_balanced_aspiration}), we conclude the following result.

\begin{theorem}
Given $\Gamma$ a roommates matching game with continuous payoff functions, let $\overline{f} \in \mathbb{R}^{|D|}$ be the balanced aspiration computed by the market procedure of Alkan and Tuncay. Then, the set of pairwise stable allocations is non-empty if and only if $\overline{f}$ can be implemented by an allocation $(\mu,\vec{x})$.
\end{theorem}

\section{Renegotiation proofness}\label{sec:internal_stability_one_to_many}

This section is devoted to extending renegotiation proofness to the two submodels considered in the previous sections. However, in order to simplify the exposition, we do the formal adaptation only for additive separable matching games. Please remark that to obtain the definitions and results for roommates matching games, it will be enough with considering pairs of doctors instead of pairs doctor-hospital. Finally, since it will be useful for the complexity study, we will focus on $\varepsilon$-renegotiation proofness. 

Given a strategy profile $\vec{x} \in \prod_{d \in D} X_d$, we denote $(\vec{x}_{-d},s_d)$ the strategy profile obtained when player $d$ replaces her strategy $x_d \in \vec{x}$ by $s_d \in X_d$. Analogously, we denote a strategy replacement from a hospital.

\begin{definition}\label{def:eps_internally_stable_matching_profile}
\textup{An $\varepsilon$-\textbf{pairwise stable} allocation $\pi=(\mu,\vec{x},\vec{y})$ is $\varepsilon$-\textbf{renegotiation proof} if for any pair $(I,h)\in \mu$, any $d \in I$, and any $(s_d,t_h)\in X_d \times Y_h$, it holds,
\begin{itemize}
\item[1.] If $f_{d,h}(s_d,y_h) > f_{d,h}(x_d,y_{d,h}) + \varepsilon$ then, $(\mu,(\vec{x}_{-d},s_d),\vec{y})$ is not $\varepsilon$-pairwise stable,
\item[2.] If $g_{d,h}(x_d,t_{d,h}) > g_{d,h}(x_d,y_{d,h}) + \varepsilon$ then, $(\mu,\vec{x},(\vec{y}_{-(d,h)},t_{d,h}))$ is not $\varepsilon$-pairwise stable.
\end{itemize}}
\end{definition}  

$\varepsilon$-Renegotiation proof allocations can be characterized as all those allocations in which all agents play $\varepsilon$-constrained Nash equilibria (\Cref{teo:eps_int_stability_is_equivalent_to_eps_CNEs}). In order to define constrained Nash equilibria, first we introduce the agents' reservation payoffs.

\begin{definition}
Let $\pi = (\mu,\vec{x},\vec{y})$ be am allocation, $(I,h) \in \mu$ be a matched pair, and $d \in I$ a fixed doctor. We define the $\varepsilon$-\textbf{reservation payoffs} of $d$ and $h$ respectively, as,
\begin{align}
    \begin{split}\label{eq:eps_outside_options_d}
    f_d^{\pi}(\varepsilon) := &\max\bigl\{ f_{d,k}(s,t) \mid g_{d,k}(s,t) \bbi \min_{d' \in \mu(k)} g_{d',k}(x_{d'},y_{d,h}) +  \varepsilon, k \in H_0 \setminus h, s \in X_d, t \in Y_k\bigr\},\\ 
    g_h^{\pi}(\varepsilon) := &\max\bigl\{g_{k,h}(s,t) \mid f_{k,h}(s,t) \bbi f_{k,\mu(k)}(x_{k},y_{\mu(k),k}) +  \varepsilon, k \in D_0 \setminus I, s \in X_k, t \in Y_h\bigr\}.
    \end{split}
\end{align}
\end{definition}

It may be intuitive to think that hospitals should have a reservation payoff for each of their doctors. The intuition is correct. However, as reservation payoffs depend on the agents outside of the \textit{couple}, the hospital has exactly the same reservation payoff for each of its doctors. In particular, no doctor should decrease her contribution to the hospital $h$'s payoff below $g_h^{\pi}(\varepsilon)$, otherwise, the hospital will have incentives to replace her. We give next, the definition of $\varepsilon$-constrained Nash equilibria.

\begin{definition}
Given an allocation $\pi = (\mu,\vec{x},\vec{y})$, a pair $(I,h) \in \mu$, a doctor $d \in I$, and their reservation payoffs $(f_d^{\pi}(\varepsilon), g_h^{\pi}(\varepsilon))$, a strategy profile $(x'_d,y'_{d,h}) \in X_d \times Y_h$ is 
\begin{itemize}[leftmargin = *]
\item[1.] $\varepsilon$\textbf{-Feasible} if $f_{d,h}(x'_d,y'_{d,h}) + \varepsilon \geq f_d^{\pi}(\varepsilon)$ and $g_{d,h}(x'_d,y'_{d,h}) + \varepsilon \geq g_h^{\pi}(\varepsilon)$,
\item[2.] An $\varepsilon$-$(f_d^{\pi}(\varepsilon), g_h^{\pi}(\varepsilon))$-\textbf{constrained Nash equilibrium} (CNE) if it is $\varepsilon$-feasible and it satisfies,
\begin{align*}
f_{d,h}(x'_d,y'_{d,h}) + \varepsilon &\geq \max\{f_{d,h}(s,y'_{d,h}) : g_{d,h}(s,y'_{d,h}) + \varepsilon \geq g_h^{\pi}(\varepsilon), s \in X_d\},\\
g_{d,h}(x'_d,y'_{d,h}) + \varepsilon &\geq \max\{g_{d,h}(x'_d,t) : f_{d,h}(x'_d,s) + \varepsilon \geq f_d^{\pi}(\varepsilon), t \in Y_h\}.
\end{align*}

We denote the set of $\varepsilon$-$(f_d^{\pi}(\varepsilon), g_h^{\pi}(\varepsilon))$-CNE by $\varepsilon$-CNE($f_d^{\pi}(\varepsilon), g_h^{\pi}(\varepsilon)$).
\end{itemize}
\end{definition}

We extend the characterization of renegotiation proof allocations through constrained Nash equilibria (Proposition 10 \cite{garrido2025stable}) to the $\varepsilon$-case.

\begin{theorem}\label{teo:eps_int_stability_is_equivalent_to_eps_CNEs}
An $\varepsilon$-pairwise stable allocation $\pi = (\mu,\vec{x},\vec{y})$ is $\varepsilon$-renegotiation proof if and only if for any pair $(I,h) \in \mu$ and $d\in I$, $(x_d,y_{d,h})$ is an $\varepsilon$-$(f_d^{\pi}(\varepsilon), g_h^{\pi}(\varepsilon))$-constrained Nash equilibria, where $(f_d^{\pi}(\varepsilon), g_h^{\pi}(\varepsilon))$ are the agents' reservation payoffs (\Cref{eq:eps_outside_options_d}).
\end{theorem}

\begin{proof}
Suppose that all couples play constrained Nash equilibria. Let $(d,h) \in \mu$ be a couple and $(x_d,y_{d,h})$ be their $\varepsilon$-$(f_d^{\pi}(\varepsilon), g_h^{\pi}(\varepsilon))$-CNE. Suppose there exists $x'_d \in X_d$ such that, $f_{d,h}(x_d', y_{d,h}) \bbi f_{d,h}(x_d,y_{d,h}) + \varepsilon.$ It follows, 
$$f_{d,h}(x_d', y_{d,h}) \bbi \max \{f_{d,h}(s,y_{d,h}) : g_{d,h}(s,y_{d,h}) + \varepsilon \geq g_h^{\pi}(\varepsilon), s \in X_d \}.$$
Thus, $f_{d,h}(x_d',y_{d,h}) + \varepsilon \sm g_h^{\pi}(\varepsilon)$. Let $d'$ be the player that attains the maximum in $g_h^{\pi}(\varepsilon)$. Then, $(d',h)$ is an $\varepsilon$-blocking pair of $\pi$. For player $h$ the proof is analogous.

Conversely, suppose $\pi$ is $\varepsilon$-renegotiation proof. Let $(d,h) \in \mu$ be a couple and $(x_d,y_{d,h})$ be their strategy profile. For any $x_d' \in X_d$ such that $f_{d,h}(x'_d, y_{d,h}) \bbi f_{d,h}(x_d,y_{d,h}) + \varepsilon$, it holds, $g_{d,h}(x'_d, y_{d,h}) + \varepsilon \sm g_h^{\pi}(\varepsilon)$. Thus, 
$$f_{d,h}(x_d,y_{d,h}) + \varepsilon \geq \max\{f_{d,h}(s,y_{d,h}) : g_{d,h}(s,y_{d,h}) + \varepsilon \geq g_h^{\pi}(\varepsilon), s \in X_d\}.$$
For player $h$ the proof is analogous.
\end{proof}

$\varepsilon$-Constrained Nash equilibria are not guaranteed to exist in every two-player game. Due to this, we extend the class of feasible games.

\begin{definition}
A two-player game is called $\varepsilon$-\textbf{feasible} if for any pair of $\varepsilon$-reservation payoffs which admits at least one $\varepsilon$-feasible strategy profile, there exists an $\varepsilon$-constrained Nash equilibrium for the same pair of $\varepsilon$-reservation payoffs.
\end{definition}

The class of $0$-feasible games is proved to contain all zero-sum games with a value, strictly competitive games with an equilibrium, potential games, and infinitely repeated games (Theorem 8 \cite{garrido2025stable}). In particular, for any $\varepsilon \geq 0$, the following result holds.

\begin{proposition}
\label{prop:class_of_eps_feasible_games}
The class of $\varepsilon$-feasible games includes zero-sum games with a value, strictly competitive games with an equilibrium, potential games, and infinitely repeated games.
\end{proposition}

Although \Cref{prop:class_of_eps_feasible_games} can be obtained as a corollary of Theorem 8 in \cite{garrido2025stable}, we will conduct its formal proof during the complexity study in order to show that renegotiation-proof allocations can be computed in polynomial time.

\begin{remark}
Zero-sum matching games, that is, matching games where all couples play zero-sum games, as showed in \cite{garrido2025stable}, capture the general model of Demange and Gale on \textit{matching with transfers}. Our two new submodels come to complete the list by also including the seminal works of Crawford and Knoer \cite{crawford_job_1981} and Kelso Jr. and Crawford \cite{kelso_jr_job_1982}, and the many articles on roommates with transferable utility \cite{andersson_competitive_2014,chiappori_roommate_2014,eriksson_stable_2001,klaus_consistency_2010,talman_model_2011}.
\end{remark}

\Cref{Algo:strategy_profiles_modification_eps_version} shows the pseudo-code of the $\varepsilon$-renegotiation process.

\begin{algorithm}[ht]
\SetKwInOut{Input}{input}\SetKwInOut{Output}{output}
\Input{$\pi = (\mu,\vec{x},\vec{y})$ $\varepsilon$-pairwise stable allocation}
\SetInd{0.2cm}{0.2cm}

$t \longleftarrow 1, \pi(t) \longleftarrow \pi$
 
\While{True}{
\For{$(d,h) \in \mu$}{
Compute the reservation payoffs $f_d^{\pi(t)}$ and $g_h^{\pi(t)}$ (\Cref{eq:eps_outside_options_d}) 

Choose $(x_d^*,y_{d,h}^*) \in$ $\varepsilon$-CNE$(f_d^{\pi}(\varepsilon), g_h^{\pi}(\varepsilon))$ and set $(x_d^{t+1},y_{d,h}^{t+1}) \longleftarrow (x_d^*,y_{d,h}^*)$}
\If{For any $(d,h) \in \mu, (x_d^{t+1},y_{d,h}^{t+1}) = (x_d^t,y_{d,h}^t)$}{Output $\pi(t)$}
$t \longleftarrow t+1$
}
\caption{Renegotiation process.}
\label{Algo:strategy_profiles_modification_eps_version}
\end{algorithm}

\begin{theorem}\label{teo:eps_strag_prof_mod_algo_is_correct}
If \Cref{Algo:strategy_profiles_modification_eps_version} converges, its output is an $\varepsilon$-pairwise stable and $\varepsilon$-renegotiation proof allocation. 
\end{theorem}

\begin{proof}
By construction, the output of \Cref{Algo:strategy_profiles_modification_eps_version} is $\varepsilon$-renegotiation proof (\Cref{teo:eps_int_stability_is_equivalent_to_eps_CNEs}). Regarding $\varepsilon$-pairwise stability, we prove that $\pi$ always remains $\varepsilon$-pairwise stable at every iteration $T$. For $T=0$ it holds as the input of \Cref{Algo:strategy_profiles_modification_eps_version} is $\varepsilon$-pairwise stable. Suppose that for some $T \bbi 0$, $\pi(T)$ is $\varepsilon$-pairwise stable but there exists an $\varepsilon$-blocking pair $(d,h)$ of $\pi(T+1)$. Then, there exists $(x^*,y^*) \in X_d \times Y_h$ such that 
\begin{align*}
f_{d,h}(x^*,y^*) \bbi f_d(\pi(T+1)) + \varepsilon \text{ and } g_{d,h}(x^*,y^*) \bbi \min_{k \in \mu^{T+1}(h)}g_{k,h'}(x_k^{T+1},y_{k,h}^{T+1}) + \varepsilon.
\end{align*}
Necessarily, $d$ or $h$ changed of strategy profile at $T$, otherwise $(d,h)$ would also block $\pi(T)$. Without loss of generality, suppose $d$ did. It follows,
\begin{align*}
f_{d,h}(x^*,y^*) \bbi f_d(\pi(T+1)) + \varepsilon = f_{d,\mu(h)}(x',y') + \varepsilon \geq f_d^{\pi(T)}(\varepsilon) \geq f_{d,h}(x^*,y^*),
\end{align*}
where $f_d^{\pi(T)}(\varepsilon)$ is $d$'s reservation payoffs at time $T$, $(x',y') \in \varepsilon$-CNE($f_d^{\pi(T)}(\varepsilon), g_{\mu(h)}^{\pi(T)}(\varepsilon)$) is the CNE chosen by $(d,\mu(d))$ at time $T$, and the last inequality comes from \Cref{eq:eps_outside_options_d}. We get a contradiction.
\end{proof}

\section{Bi-matrix matching games}\label{sec:complexity}

The deferred-acceptance algorithm of Gale and Shapley \cite{gale_college_1962} is guaranteed to converge in at most $\mathcal{O}(N^2)$ iterations, where $N$ is the size of the biggest market side. Shapley and Shubik \cite{shapley_assignment_1971} also achieved a polynomial complexity when computing their stable allocations thanks to the linearity of the payoff functions over the payments, which allowed them to solve their problem by linear programming. In the roommates problem, Irving \cite{irving_complexity_1986} designed a polynomial algorithm to compute a stable matching of the problem, in case of existence, or to report the non-existence of stable allocations. The first of the two algorithms designed by Demange et al. \cite{demange_multi-item_1986}, an increasing price mechanism with integer payoffs, is also guaranteed to converge in polynomial time to an exact stable allocation. The same holds for Andersson et al. who adapted this last algorithm to roommates with integer transferable utility \cite{andersson_competitive_2014}. 

Computing exact stable solutions for matching problems with non-integer utilities is still an open problem as algorithms depend either on discretizations of the utility spaces or fixed point theorems, leading authors to design approximation schemes for their problems. For example, the second algorithm of Demange et al. \cite{demange_multi-item_1986}, the continuous market procedure of Alkan and Tuncay \cite{alkan_pairing_2014} in the roommates with non-transferable utility setting, and the deferred-acceptance with competitions algorithm of Garrido-Lucero and Laraki, all of them converge in a bounded number of iterations $T$ to an $\varepsilon$-stable solution, with $T \propto \frac{1}{\varepsilon}$.

In this section, we extend the literature results proving that the stable allocation computation algorithms for matching games have not only a bounded number of iterations, but a  polynomial-time complexity whenever players play finite zero-sum games in mixed strategies, finite strictly competitive games in mixed strategies, or infinitely repeated games with finite stage games in mixed strategies. We will focus on the two submodels presented in \Cref{sec:matching_with_contracts_matching_games,sec:roommates_matching_games}. First of all, we formalize the notion of bi-matrix (or finite) game in mixed strategies.

\begin{definition}
A two-player game $G = (X,Y,f,g)$ is called a \textbf{bi-matrix game in mixed strategies} if there exist $S,T$ finite strategy sets such that,
\begin{align*}
X &:= \Delta(S) = \bigl\{x \in [0,1]^{|S|} : \sum_{s \in S} x(s) = 1\bigr\} \text{ and } Y := \Delta(T) = \bigl\{y \in [0,1]^{|T|} : \sum_{t \in T} y(t) = 1\bigr\},
\end{align*}
correspond to the simplex of $S$ and $T$, respectively, and the payoff functions are,
\begin{align*}
f(x,y) &:= xAy = \sum_{s \in S}\sum_{t \in T} A(s,t)x(s)y(t) \text{ and }
g(x,y) := xBy = \sum_{s \in S}\sum_{t \in T} B(s,t)x(s)y(t),
\end{align*}
where $x \in X, y \in Y$, and $A,B \in \mathbb{R}^{|S|\cdot|T|}$ are payoff matrices. With this in mind, we can define a bi-matrix matching game as any matching game in which all two-player games are bi-matrix games in mixed strategies.
\begin{itemize}[leftmargin =*]
\item[1.] An \textbf{additive separable bi-matrix matching game} is any additive separable matching game 
$$\Gamma = (D_0,H_0,(X_d)_{d \in D}, (Y_h)_{h \in H}, (f_{d,h},g_{d,h})_{d\in D, h \in H},(\underline{f}_d)_{d \in D}, (\underline{g}_h)_{h \in H}),$$ 
where each two-player game $G_{d,h} := (X_d,Y_h,f_{d,h},g_{d,h})$ is a finite game in mixed strategies, i.e., there exist finite pure strategy sets $S_d, T_h$ such that $X_d = \Delta(S_d)$, $Y_h = \Delta(T_h)$ and,
\begin{align*}
f_{d,h}(x_d,y_{d,h}) &= x_dA_{d,h}y_{d,h} \text{ and }
g_{d,h}(x_d,y_{d,h}) = x_dB_{d,h}y_{d,h},
\end{align*}
with $A_{d,h}, B_{d,h} \in \mathbb{R}^{|S_d|\cdot|T_h|}$ payoff matrices for all $d \in D$ and $h \in H$.
\item[2.] A \textbf{roommates bi-matrix matching game} is any roommate matching game $$\Gamma = (D_0,(X_d)_{d \in D},(f_{d,d'})_{d,d'\in D},(\underline{f}_d)_{d \in D}),$$
where each two-player game $G_{d,d'} := (X_d,X_{d'},f_{d,d'},f_{d',d})$ is a finite game in mixed strategies, i.e., there exist finite pure strategy sets $S_d, S_{d'}$ such that $X_d = \Delta(S_d)$, $X_d = \Delta(S_{d'})$ and,
\begin{align*}
f_{d,d'}(x_d,x_{d'}) &:= x_dA_{d,d'}x_{d'} \text{ and }
f_{d',d}(x_{d'},x_d) := x_{d'}A_{d',d}x_{d},
\end{align*}
with $A_{d,d'}, A_{d',d'} \in \mathbb{R}^{|S_d|\cdot|S_{d'}|}$ payoff matrices.
\end{itemize}
\end{definition}

For examples of bi-matrix matching games, check \Cref{ex:hedonic_game_example,ex:roommates}. 




\subsection{Deferred-acceptance with competitions algorithm}

Consider an additive separable bi-matrix matching game

$$\Gamma = (D_0,H_0,G_{d,h} = (X_d,Y_h,f_{d,h},g_{d,h})_{d \in D,h \in H}, (\underline{f}_d)_{d \in D}, (\underline{g}_h)_{h \in H}).$$

\Cref{Algo:Propose_dispose_algo_additive_separable_case} states the deferred-acceptance with competitions (DAC) algorithm adapted to the bi-matrix setting.

\begin{algorithm}[H]
\textbf{Input}: $\Gamma$ a matching game, $\varepsilon \bbi 0$,

Set $D' \leftarrow D$ as the set of unmatched doctors

\While{$D' \neq \emptyset$}{
Let $d \in D'$ and $(h,d',x_d,y_{d,h})$ be a solution to,
\begin{align}
    \begin{split}\label{eq:optimal_proposal_problem_additive_separable}
    \max\bigl\{ wA_{d,h}z \mid wB_{d,h}z \geq \min_{d' \in \mu(h)} x_{d'}B_{d',h}y_{d,h} + \varepsilon, h \in H_0, (w,z) \in X_d \times Y_h\bigr\}
    \end{split}
\end{align}
\If{$|\mu(h)| \sm q_h$}{$d$ is accepted}
\Else{$d$ and $d'$ \textbf{compete} for $h$ as in a \textbf{second-price auction}:
let $\beta_d$ be the reservation payoff of $d$ (and analogously $\beta_{d'}$ for $d'$), solution to the following problem,
\begin{align*}
    \max\bigl\{xA_{d,h'}y \mid xB_{d,h'}y \geq \min_{d'\in \mu(h')} xB_{d',h'}y + \varepsilon, h' \in H_0\setminus\{h\}, (x,y) \in X_d \times Y_{h'}\bigr\}.
\end{align*}
Then, $d$'s bid (and analogously for $d'$) is computed by,
\begin{align*}
    \lambda_d := \max\bigl\{ xB_{d,h}y \mid xA_{d,h}y  \geq \beta_d, (x,y) \in X_d\times Y_h \bigr\}.
\end{align*}
The winner is the doctor with the highest bid. Finally the winner, namely $d$, pays the second highest bid. Formally, $d$ solves,
\begin{align*}
    \max\left\{ xA_{d,h}y \mid xB_{d,h}y \geq \lambda_{d'}, (x,y) \in X_d \times Y_h\right\}. 
\end{align*}
}}
\caption{DAC algorithm: Additive separable case.}
\label{Algo:Propose_dispose_algo_additive_separable_case}
\end{algorithm}

As all the games are finite games played in mixed strategies, all agents have compact strategy sets and continuous payoff functions. Therefore, the DAC algorithm is guaranteed to converge to an $\varepsilon$-pairwise stable allocation (\Cref{teo:propose_dispose_algo_is_correct}). Moreover, the convergence is done in a finite number of iterations as we prove now.

\begin{theorem}\label{teo:propose_dipose_algo_number_of_iterations}
The deferred-acceptance with competitions algorithm converges in a bounded number $T\propto \frac{1}{\varepsilon}$ of iterations.
\end{theorem}

\begin{proof}
For every hospital, $h \in H$, consider the value,
\begin{align*}
    G_h := \max\{B_{d,h}(s,t) - \underline{g}_h : d \in D_0, s \in S_d, t \in T_h\},
\end{align*}
and let $G_{\text{max}} := \max_{h\in H} G_h$ be the maximum of them. By construction, \Cref{Algo:Propose_dispose_algo_additive_separable_case} increases hospitals' payoffs at each iteration by at least $\varepsilon$. Therefore, the number of iterations is bounded by $T := \frac{1}{\varepsilon} G_{\text{max}}$.
\end{proof}

Note that $G_{\text{max}}$ does not depend on the number of players nor the number of pure strategies per player but only on the values of the payoff matrices. Therefore, taking bounded payoff matrices, $T$ only depends on the relaxation rate $\varepsilon$. 

We aim to study next under which assumptions the iterations of the DAC algorithm have polynomial complexity. 
Remark that all the optimization problems solved during an iteration of \Cref{Algo:Propose_dispose_algo_additive_separable_case} have a quadratically constrained quadratic programming (QCQP) structure (Problem \ref{eq:optimal_proposal_problem_additive_separable} can be decomposed in $|H|$ QCQP sub-problems.) \cite{anstreicher_convex_2012,daspremont_relaxations_2003,linderoth_simplicial_2005}. Particular complexity issues will arise when solving this kind of optimization problems (check \Cref{sec:qcqp_problems}).

\subsection{Market procedure}

The market procedure of Alkan and Tuncay \cite{alkan_pairing_2014} studied in \Cref{sec:roommates_matching_games} for roommates with non-transferable utility computes a balanced aspiration $f \in \mathbb{R}^{|D|}$. The aspiration $f$ corresponds to the payoff profile that players must have at any pairwise stable allocation. As their existence theorem (\Cref{teo:equivalence_balanced_aspirations_and_stable_allocations}) states, a pairwise stable allocation will exist if and only the output of the market procedure is realizable by a proper allocation (and not a semiallocation allowing half-partnerships).

In this section we will deal with the following complexity issue of the pairwise stable allocations computation for the roommates submodel: Given a couple of doctors $(d_1,d_2) \in D\times D$ and strategy profiles $(f_1,f_2)\in \mathbb{R}^2$ such that, $f_{d_1} = u_{1,2}(f_{d_2})$, we aim to compute $(x_{d_1},x_{d_2}) \in X_{d_1}\times X_{d_2}$ such that $f_{d_1}(x_{d_1},x_{d_2}) = f_{1} \text{ and } f_{d_2}(x_{d_2},x_{d_1}) = f_{2}$. Being able to find the strategy profile that achieves the payoffs $(f_1,f_2)$, for any payoff profile, will allow us to solve two challenges of this model: (1) the computation of the \textit{demand sets} of the players at every iteration of the market procedure, and (2) given a balanced aspiration $f$, the computation of the strategy profile $\vec{x} \in X_D$ such that $(\mu,\vec{x})$ implements $f$, whenever this can be done.

Consider a bi-matrix game $G = (X,Y,A,B)$ with $X = \Delta(S)$ and $Y = \Delta(T)$ being simplex, and $A,B$ payoff matrices. Given $(u,v) \in \mathbb{R}^2$, we aim to find $x \in X, y \in Y$, such that,
\begin{align*}
    xAy &= u \Longleftrightarrow \sum_{s \in S}\sum_{t\in T} A(s,t)x_sy_t = u \text{ and }
    xBy = v \Longleftrightarrow \sum_{s \in S}\sum_{t\in T} B(s,t)x_sy_t = v.
\end{align*}

The previous system of quadratic equations can be seen as a QCQP with a constant objective function. Therefore, being able to solve QCQPs in polynomial time would allow us to compute the strategy profile of a couple that achieves a given payoff profile. We will come back to this problem later.

\subsection{Renegotiation process}

An iteration of the renegotiation process computes the reservation payoffs of all couples and a constrained Nash equilibrium to replace their current strategy profile. For a given couple $(d,h) \in \mu$, both problems are QCQP programs. Indeed, note first that reservation payoffs are computed in the same way as optimal proposal, while constrained Nash equilibria require to solve the multi-optimization problem,
\begin{align*}
    \max \bigl\{ w A_{d,h} z \mid z \in \argmax\{w B_{d,h} z \mid wA_{d,h} z + \varepsilon \geq f_d^\pi(\varepsilon), z \in Y_h \}, wB_{d,h}z + \varepsilon \geq g_h^\pi(\varepsilon), w \in X_d \bigr\}.
\end{align*}

\subsection{Quadratically constrained quadratic programs}\label{sec:qcqp_problems}

The main issue in the complexity study of our algorithms is the presence of quadratically constrained quadratic programming (QCQP) problems \cite{anstreicher_convex_2012,daspremont_relaxations_2003,linderoth_simplicial_2005}. As we have already remarked, the optimization problems solved during an iteration of the deferred-acceptance with competitions algorithm, the demand sets of the market procedure of Alkan and Tuncay, the allocation computation given a balanced aspiration, the computation of the reservation payoffs during the renegotiation process, or even the constrained Nash equilibria computation, all of them have the following structure, 
\begin{align}
    \begin{split}\label{def:general_opt_problem}
    \max\{ x A y \mid xBy \geq c, x \in X, y \in Y\},
    \end{split}
\end{align}
where $A,B$ are real-valued matrices, $c \in \mathbb{R}$, and $X,Y$ are simplex. For negative semi-definite matrices $A$ and $B$, Problem (\ref{def:general_opt_problem}) corresponds to a convex problem and can be solved in polynomial time. However, in its most general case, Problem (\ref{def:general_opt_problem}) is NP-hard. Luckily, for zero-sum games, strictly competitive games, and infinitely repeated games we will manage to reduce these problems to a polynomial number of linear programs.

Linear programming is one of the most useful tools to prove the polynomial complexity of given problems. The first polynomial algorithms for linear programming problems were published by Khachiyan \cite{khachiyan_polynomial_1979} and Karmarkar \cite{karmarkar_new_1984}. For our analysis, we will refer to the complexity result of Vaidya \cite{vaidya_speeding-up_1989}.

\begin{theorem}[Vaidya'89]\label{teo:vaidya_complexity_LP}
Let $P$ be a linear program with $m$ constraints, $n$ variables, and such its data takes $L$ bits to be encoded. Then, in the worst case, $P$ can be solved in $\mathcal{O}((n+m)^{1.5}nL)$ elementary operations.
\end{theorem}

The following section will perform the complexity study for zero-sum matching games. The study for infinitely repeated matching games and strictly competitive matching games are included, respectively, in \Cref{sec:complexity_infinitely_repeated_games} and \Cref{sec:complexity_strictly_competitive_games} in the appendix. The rest of this section is devoted to discuss the complexity of computing pairwise-Nash stable allocations in the matching games without commitment model also introduced in \cite{garrido2025stable}, where agents cannot commit to play their strategies and, therefore, any stable allocation forces the couples to play Nash equilibria.

\subsection{Complexity of matching games without commitment}

Let $\Gamma$ be a matching game and suppose, for simplicity, that hospitals have unit quotas, i.e., $\Gamma$ is a one-to-one matching game. An allocation $\pi = (\mu,\vec{x},\vec{y})$ is called \textbf{pairwise-Nash stable} (Definition 8 \cite{garrido2025stable}) if (1) all couples play a Nash equilibrium of their game and (2) no couple of agents can match together, play a Nash equilibrium in their game, and end up better-off. 

Finding pairwise-Nash stable allocations can be directly done by using the deferred-acceptance with competitions algorithm (\Cref{Algo:Propose_dispose_algo_general_case}) restricting the strategy profiles of the couples to be in their set of Nash equilibria. In particular, the complexity of solving any of the related optimization problems boils down to the complexity of computing Nash equilibria that guarantee a certain given utility to each of the players (the reservation payoffs). The computation of a Nash equilibria in finite games with mixed strategies belongs to the complexity class PPAD \cite{papadimitriou_complexity_2007}. Moreover, when trying to compute a Nash equilibrium that guarantees at least some utility $\delta$ to each of the players, the problem becomes NP-complete. Still, and as we will see in \Cref{sec:complexity_zero_sum_games} (and \Cref{sec:complexity_infinitely_repeated_games,sec:complexity_strictly_competitive_games} in the appendix) for matching games under commitment, the computation of pairwise-Nash stable allocations becomes a polynomial problem. 

\section{Complexity study for Zero-sum matching games}\label{sec:complexity_zero_sum_games}

Consider a matching game $\Gamma$ in which all strategic games are finite zero-sum matrix games in mixed strategies, from now on, a zero-sum matching game. Note that $\Gamma$ can be either a additive separable matching game or roommates matching game. When needed, we will specify the submodel considered.

We study the complexity of the three algorithms recalled in the previous section: \textit{deferred-acceptance with competitions}, \textit{market procedure}, and \textit{renegotiation process}, when the matching game $\Gamma$ is a zero-sum matching game (with the corresponding submodel). Then, \ref{sec:complexity_strictly_competitive_games} will extend the results to strictly competitive matching games. The following subsections will split the analysis for each algorithm. All the presented results will use the following main theorem.

\begin{theorem}\label{teo:complexity_of_solving_the_opt_problem_zero_sum_case}
Let $G = (X,Y,A,B)$ be a finite zero-sum game in mixed strategies, where $X = \Delta(S)$, $Y = \Delta(T)$ are simplexes with $S,T$ pure strategy sets, and $A,B$ are payoff matrices. Given a vector $c$, the QCQP Problem (\ref{def:general_opt_problem}),
\begin{align*}
\max\{ x A y \mid xBy \geq c, x \in X, y \in Y\},
\end{align*}
can be solved in $\mathcal{O}(|S|\cdot|T|)$ comparisons.
\end{theorem}

To prove \Cref{teo:complexity_of_solving_the_opt_problem_zero_sum_case} we need a preliminary result. Note that, first of all, since $G$ is a zero-sum game, the QCQP Problem (\ref{def:general_opt_problem}) can be rewritten as
\begin{align}
\label{def:general_opt_problem_zero_sum}
\max\{ x A y \mid xAy \leq c, x \in X, y \in Y\}.
\end{align}

Therefore, solving the previous optimization problem is equivalent to finding a strategy profile $(x,y)$ such that $xAy = \min\{c, \max A\}$, where $\max A := \max_{s,t}A(s,t)$ and $\min A := \min_{s,t}A(s,t)$. Without loss of generality it can be always considered $\min A \leq c \leq \max A$ since replacing $c$ by $\min\{c, \max A\}$ does not change at all Problem (\ref{def:general_opt_problem_zero_sum}) and for $c \sm \min A$, the problem is infeasible.

\begin{lemma}\label{prop:the_opt_problem_has_always_a_pure_solution}
Given a matrix payoff $A$ and $c \in \mathbb{R}$, with $\min A \leq c \leq \max A$, there always exists $(x,y) \in X \times Y$, such that $xAy = c$, with $x$ or $y$ being a pure strategy.
\end{lemma}

\begin{proof}
Let $s \in S$ be a pure strategy for player $1$ in $G$, such that there exist $t,t' \in T$, with $A(s,t) \leq c \leq A(s,t')$. Then, there exists $\lambda \in [0,1]$ such that $\lambda A(s,t) + (1-\lambda)A(s,t') = c$. Even more, $\lambda$ is explicitly given by
\begin{align}\label{eq:lambda_first_case}
    \lambda = \frac{c - A(s,t)}{A(s,t') - A(s,t)}.    
\end{align}

Suppose that such a pure strategy $s$ does not exist, so for any $s \in S$, either $A(s,t) \leq c, \forall t \in T$, or $A(s,t) \geq c, \forall t \in T$. Let $t \in T$ be any pure strategy of player $2$. Then, since $\min A \leq c \leq \max A$, there exists $s,s' \in S$ such that $A(s,t) \leq c \leq A(s',t)$. Thus, considering $\lambda$ given by,
\begin{align}\label{eq:lambda_second_case}
    \lambda = \frac{c - A(s,t)}{A(s',t) - A(s,t)},    
\end{align}
it holds that $\lambda A(s,t) + (1-\lambda)A(s',t) = c$.
\end{proof}


We are ready to prove the complexity of solving the QCQP problem for a zero-sum game
(\Cref{teo:complexity_of_solving_the_opt_problem_zero_sum_case}).

\begin{proof}[Proof of \Cref{teo:complexity_of_solving_the_opt_problem_zero_sum_case}.]
The complexity of solving the QCQP Problem (\ref{def:general_opt_problem_zero_sum}) corresponds to the one of finding the pure strategies used in the convex combination of \Cref{prop:the_opt_problem_has_always_a_pure_solution}'s proof and then computing the corresponding $\lambda$. Let 
$$S^+ := \{s \in S : \exists t \in T, A(s,t) \geq c\} \text{ and } S^- := \{s \in S : \exists t \in T, A(s,t) \leq c \}.$$
These sets are computed in $|S|\cdot|T|$ comparisons, as in the worst case we have to check all coefficients in $A$. As $\min A \leq c \leq \max A$, both sets are non-empty. If $S^+ \cap S^- \neq \emptyset$, there exist $s \in S$ and $t,t' \in T$ such that $A(s,t) \leq c \leq A(s,t')$, so \Cref{eq:lambda_first_case} gives the solution desired. Otherwise, there exists $t \in T$ and $s,s' \in S$ such that $A(s,t) \leq c \leq A(s',t)$, and \Cref{eq:lambda_second_case} gives the solution desired. Computing the intersection of $S^+$ and $S^-$ has complexity $\mathcal{O}(|S|)$. In either case (the intersection is empty or not), finding the pure strategies needed for the convex combination takes at most $|T|$ comparisons. Finally, computing $\lambda$ requires a constant number of operations on the sizes of the strategy sets. Adding all up, we obtain the stated result.
\end{proof}

\subsection{Deferred-acceptance with competitions algorithm}\label{sec:deferred_acceptance_algo_zero_sum_games}

Suppose $\Gamma$ is a zero-sum additive separable matching game. We aim to prove the following result.

\begin{theorem}[Complexity]\label{teo:complexity_propose_dispose_algo_zero_sum_case}
Let $d \in D$ be a proposer doctor. Let $h$ be the proposed hospital and $d'$ be the doctor that $d$ wants to replace. If $d$ is the winner of the competition, the entire iteration of the DAC algorithm (\Cref{Algo:Propose_dispose_algo_additive_separable_case}) has complexity,
$$\mathcal{O}\left(\left[|H|\cdot|D| + (|S_d| + |S_{d'}|)\cdot\sum_{h' \in H}|T_{h'}|\right]L \right),$$
where $L$ represents the number of bits required to encode all the data.
\end{theorem}

The proof of \Cref{teo:complexity_propose_dispose_algo_zero_sum_case} is split in several results, each of them being a corollary of the complexity result for the general QCQP problem (\Cref{teo:complexity_of_solving_the_opt_problem_zero_sum_case}). 

\begin{corollary}\label{cor:optimal_proposal_complexity}
$d$'s optimal proposal can be computed in
$$\mathcal{O}\left(|H|\cdot|D| + |S_d|\cdot \sum_{h' \in H} |T_{h'}|\right)$$
comparisons.
\end{corollary}

\begin{proof}
$d$'s optimal proposal is computed by solving,
\begin{align}
\label{eq:optimization_problem_proposer_zero_sum_game}
\max\{xA_{d,h'}y \mid xA_{d,h'}y \leq \max_{d' \in \mu(h')} x_{d'}A_{d',h'}y_{d',h'} - \varepsilon, h' \in H_0, x \in X_d, y \in Y_{h'}\},
\end{align}

Problem (\ref{eq:optimization_problem_proposer_zero_sum_game}) can be solved by dividing it in $|H|$ sub-problems (one per hospital) and taking the best of the $|H|$ solutions. Once computed the right-hand side of each subproblem, they get the structure of the general QCQP Problem \ref{def:general_opt_problem_zero_sum} so they need a polynomial number of comparisons to be solved (\Cref{teo:complexity_of_solving_the_opt_problem_zero_sum_case}). Computing the right-hand side for each of them takes $|D|$ comparisons in the worst case. The complexity stated comes from putting it all together.
\end{proof}

\begin{remark}
\textup{$d$'s reservation payoff when competing for $h$ can be computed by solving Problem (\ref{eq:optimization_problem_proposer_zero_sum_game}) leaving $h$ out of the feasible region. Therefore, its complexity is bounded by the one in \Cref{cor:optimal_proposal_complexity}.}
\end{remark}

\begin{corollary}\label{cor:bid_complexity}
The computation of the reservation payoff $\beta_d$ of doctor $d$ plus her bid $\lambda_d$ during a competition takes
$$\mathcal{O}\left(|H|\cdot|D| + |S_d|\cdot \sum_{h' \in H} |T_{h'}|\right)$$
comparisons.
\end{corollary}

\begin{proof}
$d$'s bid is computed by,
\begin{align}
\label{eq:optimization_problem_max_zero_sum_game}
    \min\{xA_{d,h}y \mid xA_{d,h}y \geq \beta_d, x \in X_d, y \in Y_h\},
\end{align}
and takes $\mathcal{O}(|S_d|\cdot|T_h|)$ comparisons (\Cref{teo:complexity_of_solving_the_opt_problem_zero_sum_case}). Adding this to the complexity of computing $\beta_d$, we obtain the stated result.
\end{proof}

Finally, we study the optimization problem solved by the winner.

\begin{corollary}\label{cor:second_price_complexity}
The final strategy profile played by the winner of a competition can be computed in $\mathcal{O}\left(|S_d|\cdot|T_{h}|\right)$ comparisons.
\end{corollary}

\begin{proof}
Let $\lambda_{d'}$ be the bid of $d'$. $d$ solves,
\begin{align}
\label{eq:optimization_problem_new_zero_sum_game}
\max\{xA_{d,h}y \mid xA_{d,h}y \leq \lambda_{d'}, x \in X_d, y \in Y_h\}.
\end{align}

Problem (\ref{eq:optimization_problem_new_zero_sum_game}) has the same structure of Problem (\ref{def:general_opt_problem_zero_sum}). Therefore, we can solved it in $\mathcal{O}(|S_d|\cdot|T_h|)$ comparisons.
\end{proof}

The complexity of an entire iteration of the DAC algorithm (\Cref{teo:complexity_propose_dispose_algo_zero_sum_case}) is obtained by adding up the complexity results given in \Cref{cor:optimal_proposal_complexity,cor:bid_complexity,cor:second_price_complexity}. We omit its formal proof.

\begin{remark}
If there are at most $N$ players in each side and at most $k$ pure strategies per player, \Cref{teo:complexity_propose_dispose_algo_zero_sum_case} proves that each iteration of the DAC algorithm (\Cref{Algo:Propose_dispose_algo_additive_separable_case}) takes $\mathcal{O}((N^2+k^2)L)$ number of elementary operations in being solved, hence is polynomial. As the number of iterations does not depend on the size of the problem but only on $\varepsilon$, we conclude that computing an $\varepsilon$-pairwise stable allocation for a one-to-many zero-sum matching game is a polynomial problem.
\end{remark}

\subsection{Market procedure}\label{sec:market_procedure_zero_sum}

Let $\Gamma$ be a zero-sum roommates matching game. Computing the demand sets of the doctors for a zero-sum roommates matching game becomes particularly easy. Consider $f \in \mathbb{R}^{|D|}$ and two doctors $(d_1,d_2) \in D \times D$. Note that,
\begin{align*}
    d_2 \in P_{d_1}(f) &\Longleftrightarrow f_{d_1} = u_{1,2}(f_{d_2})                       \Longleftrightarrow f_{d_1} = -f_{d_2}
                       \Longleftrightarrow f_{d_1} + f_{d_2} = 0.
\end{align*}

In other words, to compute the demand set of a given doctor $d$, it is enough with checking whether the sum of the payoff of both agents is equal to zero. We directly state the following result.

\begin{theorem}\label{teo:complexity_demand_sets_zero_sum_roommates_games}
Computing the demand sets of all doctors during an iteration of the market procedure, given the current payoff profile $f \in \mathbb{R}^{|D|}$, has complexity $\mathcal{O}\left(|D|^2 L\right)$, where $L$ is the number of bits required to encode all the data.
\end{theorem}

The previous theorem exploits that to compute the doctors' demand sets, it is not needed to compute the strategy profiles but only to compare the agents' payoffs. 

\begin{theorem}\label{teo:complexity_implementation_roommates_problem}
Let $f \in \mathbb{R}^{|D|}$ be the output of the market procedure and suppose there exists a matching $\mu$ such that for any doctor $d \in D$, $f_{d} + f_{\mu(d)} = 0$ if $d$ is matched and $f_d = \underline{f}_d$ if $d$ is unmatched. Then, computing the strategy profile $\vec{x} \in X_D$ such that $f_{d,\mu(d)}(x_d,x_{\mu(d)}) = f_d$, for any $d$ matched, has complexity 
$$\mathcal{O}\biggl(\sum\nolimits_{(d_1,d_2) \in \mu}|S_{d_1}|\cdot|S_{d_2}| L\biggr),$$
where $L$ is the number of bits required to encode all the data.
\end{theorem}

\begin{proof}
Given a couple $(d_1,d_2) \in \mu$, we aim to find $(x_{d_1},x_{d_2}) \in X_{d_1}\times X_{d_2}$ such that, $f_{d_1,d_2}(x_{d_1},x_{d_2}) = f_{d_1} \text{ and } f_{d_2,d_1}(x_{d_2},x_{d_1}) = f_{d_2}$. Since the strategic game of $d_1$ and $d_2$ is a zero-sum game, it is enough with computing $(x_{d_1},x_{d_2})$ such that,
\begin{align*}
    &f_{d_1,d_2}(x_{d_1},x_{d_2}) = f_{d_1} \Longleftrightarrow\ x_{d_1}A_{d_1,d_2}x_{d_2} = f_{d_1} \Longleftrightarrow\ \sum_{s \in S_1}\sum_{s' \in S_2} A_{d_1,d_2}(s,s')x_{d_1}(s)x_{d_2}(s') = f_{d_1},
\end{align*}
which corresponds to a quadratic equation with $|S_1|+|S_2|$ variables. Note that this problem has the exact form $xAy = c$ studied in \Cref{prop:the_opt_problem_has_always_a_pure_solution} that, we have already discussed, can be solved in $\mathcal{O}(|S_1|\cdot|S_2|)$ comparisons. We conclude the proof by considering the sum over all couples in $\mu$ plus the number of bits required to encode the data.
\end{proof}

Checking whether the output of the market procedure is realizable by an allocation (and not a semiallocation) can be easily done in the case of zero-sum matching games. Given $f\in \mathbb{R}^{|D|}$ the output, divide the set of doctors $D$ as,
\begin{align*}
    D_+ &= \{d \in D : f_d \bbi 0\} \cup \{d_0\},\ D_- = \{d \in D : f_d \sm 0\} \cup \{d_0\},\ D_{\sim} = \{d \in D : f_d = 0\}.
\end{align*}

Then, find a correspondence $\mu$ between $D_+$ and $D_-$ such that for any pair of matched agents, the sum of their payoffs in $f$ is zero, or the agents are matched with the empty doctor in case their payoff in $f$ was equal to their IRP. Note that the correspondence can be computed in $|D|^2$ comparisons by checking all possible combinations. Regarding $D_{\sim}$, divide the set into two equal parts and match the agents between them, or the empty doctor in case their IRPs are equal to zero and $D_{\sim}$ has an odd number of doctors. In case this procedure outputs a proper matching $\mu$, the matching game has an pairwise stable allocation. Otherwise, the set of pairwise stable allocations is necessarily empty as the outcome of the market procedure cannot be realizable by an allocation.

\begin{remark}
If there are at most $N$ doctors and $k$ pure strategies per doctor, computing all the demand sets during an iteration of the market procedure has complexity $\mathcal{O}(N^2L)$. Similarly, in case the outcome of the market procedure can be implemented by an allocation, finding the strategy profiles of the agents within the allocation has complexity $\mathcal{O}(Nk^2L)$.
\end{remark}

\subsection{Renegotiation process}\label{sec:str_prof_modif_algo_zero_sum_games}

We focus now on the computation of renegotiation proof allocations. Since the same algorithm can be used for both models, additive separable matching games and roommates matching games, we will only work on the first case. All results can be directly applied to the second model.

Suppose $\Gamma$ is an additive separable matching game in which each strategic game $G_{d,h} = (X_d, Y_h, A_{d,h})$ is a finite zero-sum game in mixed strategies with value $w_{d,h}$, where $A_{d,h}$ is the payoff matrix. We aim to prove the following result.

\begin{theorem}[CNE Complexity]\label{teo:zero_sum_games_are_eps_feasible_and_polynomial}
Let $(d,h)$ be a couple and $G_{d,h} = (X_d, Y_h, A_{d,h})$ be their bi-matrix zero-sum game with value $w_{d,h}$. Let $(f_d,g_{h})$ be a pair of reservation payoffs.
Then, 
\begin{itemize}[leftmargin = *]
    \item[1.] $G_{d,h}$ is $\varepsilon$-feasible,
    \item[2.] For any $(x',y') \in \varepsilon$-CNE($f_d,g_{h}$), it holds 
    $$x'A_{d,h}y' = median\{f_d - 2 \varepsilon, w_{d,h}, g_{h} + 2\varepsilon\}.$$
    \item[3.] Computing an $\varepsilon$-CNE $(x',y')$ has complexity $$\mathcal{O}\left( \max\{|S_d|,|T_h|\}^{2.5} \cdot \min\{|S_d|,|T_h|\} \cdot L_{{d,h}}\right),$$ where $S_d,T_h$ are the pure strategy sets of the players and $L_{{d,h}}$ is the number of bits required to encode the matrix $A_{d,h}$.
\end{itemize}
\end{theorem}

We will make use of the following lemma.

\begin{lemma}\label{lemma:for_computing_t_it_is_enough_to_consider_pure_strategies}
Let $s_1,s_2 \in S_d$ be two pure strategies for player $d$, $(x^*,y^*)$ be the optimal strategies of the players, and $(x,y) \in X_d \times Y_h$ be a strategy profile such that $x$ only has $s_1,s_2$ in its support. Consider $\tau \in (0,1)$ and define $y_{\tau} := (1-\tau)y + \tau y^*$. Suppose that $xA_{d,h}y_{\tau} = f_d$ but $s_1A_{d,h}y_{\tau} \neq f_d \neq s_2A_{d,h}y_{\tau}$. Finally, suppose that $w_{d,h} \sm f_d$. Then, there always exists $\tau' \in (\tau, 1)$, and a pure strategy $s \in S_d$ such that $sA_{d,h}y_{\tau'} = f_d$.
\end{lemma}

\begin{proof} It holds,
$$xA_{d,h}y_{\tau} = x_{s_1}\cdot s_1A_{d,h}y_{\tau} + x_{s_2}\cdot s_2A_{d,h}y_{\tau} = f_d,$$
with $x_{s_1} + x_{s_2} = 1$, $x_{s_1}, x_{s_2} \in [0,1]$. Since $s_1A_{d,h}y_{\tau}$ and $s_2A_{d,h}y_{\tau}$ are both different from $f_d$, we can suppose (without loss of generality) that $s_1A_{d.h}y_{\tau} \bbi f_d$ and $s_2A_{d,h}y_{\tau} \sm f_d$. Then, as $x^*A_{d,h}y^* = w_{d,h} \sm f_d$ and $(x^*,y^*)$ is a saddle point, $s_1A_{d,h}y^* \leq w_{d,h} \sm f_d$. As $y_{\{\tau = 1\}} = y^*$, by continuity, there exists $\tau' \in (\tau, 1)$ such that, $s_1A_{d,h}y_{\{\tau = 1\}} \sm f_d = s_1A_{d,h}y_{\tau'} \sm s_1A_{d,h}y_{\tau}$.
\end{proof}

\Cref{lemma:for_computing_t_it_is_enough_to_consider_pure_strategies} can be easily extended to mixed strategies of any finite support.

\begin{proof}[Proof of \Cref{teo:zero_sum_games_are_eps_feasible_and_polynomial}.]
Let $(x^*,y^*)$ be the players' optimal strategies, i.e., the strategy profile that achieves the value of the game $x^*A_{d,h}y^* = w_{d,h}$. We split the proof into three cases.
\medskip

\noindent \textbf{1.} Suppose that $f_d - 2\varepsilon \leq w_{d,h} \leq g_{h} + 2\varepsilon$. In particular, the value of the game is $\varepsilon$-feasible for both agents. Since it is also a saddle point so agents do not have profitable deviations, $(x^*,y^*)$ is an $\varepsilon$-$(f_d,g_{h})$-CNE.  From Von Neumann's theorem, we know that $(x^*,y^*,w_{d,h})$ can be obtained from the solutions of the pair primal-dual problems,
\begin{align*}
    (P)\ \ &\min \langle c, x\rangle \hspace{2cm} (D)\ \max \langle b,y \rangle \\
    &xA_{d,h} \geq b \hspace{2.8cm} A_{d,h}y \leq c \\ 
    &x \geq 0 \hspace{3.5cm}  y \geq 0,
\end{align*}
where the variables satisfy $x \in X_d$, $y \in Y_h$, and the vectors $c,d$ are given and are equal to $1$ in every coordinate. If $(x',y')$ is the primal-dual solution and $z$ is their optimal value, the optimal strategies of player $d$ and $h$ are given by $(x^*,y^*) = (x'/z, y'/z)$, and they achieve the value of the game $w_{d,h}$. From Vaidya's linear programming complexity result (\Cref{teo:vaidya_complexity_LP}), the number of elementary operations needed to solve the primal-dual problem and computing $(x^*,y^*)$ is 
$$\mathcal{O}\left((|S_d|+|T_h|)^{1.5}\max\{|S_d|,|T_h|\}L_{d,h}\right).$$

\noindent \textbf{2.} Suppose that $w_{d,h} \sm f_d - 2\varepsilon \leq g_{h} + 2\varepsilon$. Let $(x_0,y_0)$ be an $\varepsilon$-feasible strategy profile. Consider the set 
$$\Lambda(f_d) := \{x \in X_d : \exists y \in Y_h, xA_{d,h}y + 2\varepsilon \geq f_d\}.$$
Note that $\Lambda(f_d)$ is non-empty as $(x_0,y_0)$ belongs to it. Consider the problem,
\begin{align}\label{eq:prob_zero_sum_games}
    \sup \bigl[ \inf\{xA_{d,h}y \mid xA_{d,h}y + 2\varepsilon \geq f_d, y \in Y_h \} \mid x \in \Lambda(f_d) \bigr].
\end{align}

Since the set $\{xA_{d,h}y + 2\varepsilon \geq f_d, y \in Y_h\}$, for a given $x$, is bounded, as well as the set $\Lambda(f_d)$, there exists a solution $(x,y)$ of Problem (\ref{eq:prob_zero_sum_games}). Moreover, computing $(x,y)$ has complexity $\mathcal{O}(|T_h|\cdot|S_d|^{2.5}L)$ as Problem (\ref{eq:prob_zero_sum_games}) is equivalent to solve $|T_h|$ linear programming problems, each of them with $|S_d|$ variables and $1$ constraint, and then considering the highest value between them.
By construction, $xA_{d,h}y + 2 \varepsilon \geq f_d$. Suppose $xA_{d,h}y + 2 \varepsilon \bbi f_d$. It follows,
$$xA_{d,h}y \bbi f_d - 2\varepsilon \bbi w_{d,h} = x^*A_{d,h}y^* \geq xA_{d,h}y^*,$$
where the last inequality holds as $(x^*,y^*)$ is a saddle point. Then, there exists $y' \in (y,y^*)$ such that $xA_{d,h}y' = f_d - 2\varepsilon$. This contradicts that $(x,y)$ is solution of Problem (\ref{eq:prob_zero_sum_games}). If $(x,y)$ is an $\varepsilon$-$(f_d, g_{h})$-CNE, the proof is over. Otherwise, consider the problem,
\begin{align}\label{eq:closest_point_to_Nash_eq_that_is_feasible}
  t := \sup\{\tau \in [0,1] : y_{\tau} := (1-\tau)y + \tau y^* \text{ and } \exists x_{\tau} \in X_d, x_{\tau}A_{d,h} y_{\tau} = f_d - 2\varepsilon \}.
\end{align}

The value $t$ exists as for $\tau = 0$, $xA_{d,h}y = f_d - 2 \varepsilon$. In addition, $y_t \neq y^*$ as $x^*A_{d,h}y^* \sm f_d - 2\varepsilon$ and $(x^*,y^*)$ is a saddle point. From \Cref{lemma:for_computing_t_it_is_enough_to_consider_pure_strategies}, if $xA_{d,h}y_{\tau} = f_d$ for some value $\tau \in (0,1)$, then there always exists a pure strategy $s \in S_d$ and $\tau \leq \tau' \sm 1$ such that $sA_{d,h}y_{\tau'} = f_d$. Thus, solving Problem (\ref{eq:closest_point_to_Nash_eq_that_is_feasible}) is equivalent to solve each of the next linear problems,
\begin{align*}
  t_s := \sup\{ \tau \in [0,1] : y_{\tau} := (1-\tau)y + \tau y^* \text{ and }sA_{d,h}y_{\tau} = f_d - 2\varepsilon\}, \text{ for any } s \in S_d,
\end{align*}
and then, considering $t :=\max_{s\in S_d}t_s$. Each $t_s$ can be computed in constant time over $|S_d|$ and $|T_h|$, as the linear programming problem associated has only one variable and one constraint. Finally, computing the maximum of all $t_s$ takes $|S_d|$ comparisons. We claim that $(x_t,y_t)$ is an $\varepsilon$-$(f_d, g_{h})$-CNE. Let $x' \in X_d$ such that $x'A_{d,h}y_t \leq g_{h} + \varepsilon$. We aim to prove that $x'A_{d,h}y_t \leq x_tA_{d,h}y_t + \varepsilon$. Suppose $x'A_{d,h}y_t \bbi x_tA_{d,h}y_t + \varepsilon$. It holds, 
$$x'A_{d,h}y^* \leq w_{d,h} = x^*A_{d,h}y^* \sm f_d - 2\varepsilon = x_tA_{d,h}y_t \sm x_tA_{d,hj}y_t + \varepsilon \sm x'A_{d,h}y_t.$$
Then, there exists $z \in X_d$ and $y' \in (y_t, y^*)$ such that $zA_{d,h}y' = f_d - 2\varepsilon$, contradicting that $t$ is solution of Problem (\ref{eq:closest_point_to_Nash_eq_that_is_feasible}). 

Regarding player $h$, let $y' \in Y_h$ such that $x_tA_{d,h}y' + \varepsilon \geq f_d$. We aim to prove that $x_tA_{d,h}y' \geq x_tA_{d,h}y_t - \varepsilon$, which follows from, 
\begin{align*}
    x_tA_{d,h}y' \geq f_d - \varepsilon = f_d - 2\varepsilon + \varepsilon = x_tA_{d,h}y_t + \varepsilon \bbi x_tA_{d,h}y_t - \varepsilon.
\end{align*}

We conclude that $(x_t,y_t) \in \varepsilon$-CNE($f_d,g_{h}$).

\noindent \textbf{3.} Suppose that $f_d - 2\varepsilon \leq g_{h} + 2\varepsilon \sm w_{d,h}$. Analogously to case 2 (an analogous version of \Cref{lemma:for_computing_t_it_is_enough_to_consider_pure_strategies} has to be proved as well. As the proof follows exactly the same arguments, we prefer to omit it), there exists an $\varepsilon$-$(f_d, g_{h})$-CNE $(x,y)$ satisfying $xA_{d,h}y = g_{h} + 2 \varepsilon$.
\medskip

\noindent Finally, the complexity given at the theorem's state is obtained when taking the maximum complexity between the three cases.
\end{proof}

As a corollary of \Cref{teo:zero_sum_games_are_eps_feasible_and_polynomial} we obtain the following result.

\begin{corollary}\label{cor:complexity_of_computing_all_the_values}
Given an allocation $\pi = (\mu,\vec{x},\vec{y})$, computing all games' values is a polynomial problem and its complexity is bounded by
$$\mathcal{O}\biggl(\sum_{(d,h) \in \mu} (|S_d| + |T_h|)^{1.5}\max\{|S_d|, |T_h|\}L_{d,h}\biggr).$$
\end{corollary}

\begin{proof}
Let $(I,h) \in \mu$ be a matched pair, $d \in I$ a doctor, and $G_{d,h} = (X_d, Y_h, A_{d,h})$ be a zero-sum game. The proof of \Cref{teo:zero_sum_games_are_eps_feasible_and_polynomial} in its first case proves that computing $w_{d,h}$ takes at most $\mathcal{O}((|S_d| + |T_h|)^{1.5}\max\{|S_d|, |T_h|\}L_{d,h})$ elementary operations, where $S_d, T_h$ are the players' strategy sets and $L_{d,h}$ is the number of bits required to encode the matrix $A_{d,h}$. Summing up all the couples, we obtain the stated complexity.
\end{proof}

The complexity of one iteration of the renegotiation process corresponds to the complexity of computing the reservation payoffs and a constrained Nash equilibrium for each couple. As we can have at most $|D|$ couples, the complexity of an entire iteration of the renegotiation process (\Cref{Algo:strategy_profiles_modification_eps_version}) is bounded by,
\begin{align*}
    \mathcal{O}\left( \sum_{d \in D} \left[|H|\cdot|D| + |S_d|\cdot \sum_{h \in H} |T_h| + \max\{|S_d|,|T_{\mu(d)}|\}^{2.5} \cdot \min\{|S_d|,|T_{\mu(d)}|\}\right]\cdot L\right),
\end{align*}
where $L$ is the number of bits required to encode all the problem data.

\begin{remark}
Considering $N$ agents per side and $k$ pure strategies per agent, the complexity of an entire iteration of the renegotiation process (\Cref{Algo:strategy_profiles_modification_eps_version}) is bounded by,
$\mathcal{O}\left(N^4k^{3.5}L\right).$
Hence, it is polynomial.
\end{remark}

The renegotiation process in its original version is known to converge for zero-sum matching games. However, no bound could be given to the number of iterations. For the $\varepsilon$-version, in exchange, we are able to guarantee a bound $T \propto \frac{1}{\varepsilon}$, with $T$ not depending on the problem size.

\begin{theorem}[Convergence]\label{teo:algorithm_2_ends_in_finite_time_for_zero_sum_games}
Let $\Gamma$ be a bi-matrix zero-sum matching game such that each game $G_{d,h}$ has a value $w_{d,h}$. Let $\pi = (\mu,\vec{x},\vec{y})$ be an $\varepsilon$-pairwise stable allocation, input of the $\varepsilon$-renegotiation process (\Cref{Algo:strategy_profiles_modification_eps_version}), the one defines a profile of $\varepsilon$-reservation payoffs $(f_d^{\pi}(\varepsilon), g_h^{\pi}(\varepsilon))_{d \in I,(I,h) \in \mu}$. Then, the number of iterations of \Cref{Algo:strategy_profiles_modification_eps_version} is bounded by
$$\max_{d \in I, (I,h) \in \mu}\frac{\{f_d^{\pi}(\varepsilon) - w_{d,h}, w_{d,h} - g_h^{\pi}(\varepsilon)\}}{\varepsilon}.$$
\end{theorem}

To prove \Cref{teo:algorithm_2_ends_in_finite_time_for_zero_sum_games} we will make use of the following lemma.

\begin{lemma}\label{lemma:eps_outside_options_are_monotone_in_zero_sum_games}
Let $\Gamma$ be a matching game as in \Cref{teo:algorithm_2_ends_in_finite_time_for_zero_sum_games}. Let $\pi = (\mu,\vec{x},\vec{y})$ be an $\varepsilon$-pairwise stable allocation, $(I,h)$ be a matched pair, and $d \in I$. Consider the sequence of reservation payoffs of $(d,h)$ denoted by $(f_d^{\pi(t)}(\varepsilon), g_h^{\pi(t)}(\varepsilon))_t$, with $t$ being the iterations of the renegotiation process (\Cref{Algo:strategy_profiles_modification_eps_version}). If there exists $t^*$ such that $w_{d,h} \leq f_d^{\pi(t)}(\varepsilon) - 2\varepsilon$ (resp. $w_{d,h} \geq g_h^{\pi(t)}(\varepsilon) + 2\varepsilon$), then the subsequence $(f_d^{\pi(t)}(\varepsilon))_{t\geq t^*}$ (resp. $(g_h^{\pi(t)}(\varepsilon))_{t\geq t^*}$) decreases (resp. increases) at least $\varepsilon$ at each step.
\end{lemma}

\begin{proof}
Suppose there exists an iteration $t$ in which $w_{d,h} \leq f_d^{\pi(t)}(\varepsilon) - 2 \varepsilon \leq g_h^{\pi(t)}(\varepsilon) + 2 \varepsilon$, so couple $(d,h)$ switches its payoff to $f_d^{\pi(t)}(\varepsilon) - 2 \varepsilon$ (\Cref{teo:zero_sum_games_are_eps_feasible_and_polynomial}). Let $(\hat{x}_d, \hat{y}_h)$ be the $\varepsilon$-$(f_d^{\pi(t)}(\varepsilon), g_h^{\pi(t)}(\varepsilon))$-CNE played by $(d,h)$ at iteration $t$. Since  $(\hat{x}_d, \hat{y}_h)$ must be 
$${\varepsilon}\text{-}(f_d^{\pi(t+1)}(\varepsilon), g_h^{\pi(t+1)}(\varepsilon))\text{-feasible},$$
in particular, it holds $f_d^{\pi(t+1)}(\varepsilon)\leq \hat{x}_dA_{d,h}\hat{y}_h + \varepsilon = f_d^{\pi(t)}(\varepsilon) - \varepsilon$. Therefore, the sequence of reservation payoffs starting from $t$ decreases at least in $\varepsilon$ at each step.
\end{proof}

Finally, we prove the convergence of the $\varepsilon$-renegotiation process in a $T\propto \frac{1}{\varepsilon}$ number of iterations.

\begin{proof}[Proof of \Cref{teo:algorithm_2_ends_in_finite_time_for_zero_sum_games}.]
At the beginning of the renegotiation process (\Cref{Algo:strategy_profiles_modification_eps_version}), all couples $(d,h)$ belong to one (not necessarily the same) of the following cases: $f_d^{\pi}(\varepsilon) - 2\varepsilon \leq w_{d,h} \leq g_h^{\pi}(\varepsilon) + 2\varepsilon$, $w_{d,h} \leq f_d^{\pi}(\varepsilon) - 2\varepsilon \leq g_h^{\pi}(\varepsilon) + 2\varepsilon$ or $f_d^{\pi}(\varepsilon) - 2\varepsilon \leq g_h^{\pi}(\varepsilon) + 2\varepsilon \leq w_{d,h}$. In the first case, the couple plays their Nash equilibrium and never changes it afterward. In the second case, as $f_d^{\pi}(\varepsilon)$ is strictly decreasing for $d$ (\Cref{lemma:eps_outside_options_are_monotone_in_zero_sum_games}) and bounded from below by $w_{d,h}$, the sequence of reservation payoffs converges in at most $\frac{1}{\varepsilon}(f_d^{\pi}(\varepsilon) - w_{d,h})$ iterations. Analogously, the sequence of reservation payoffs for $h$ converges on the third case in finite time. Therefore, \Cref{Algo:strategy_profiles_modification_eps_version} converges it at most $\frac{1}{\varepsilon} \max_{(d,h) \in \mu} \{f_d^{\pi}(\varepsilon) - w_{d,h}, w_{d,h} - g_h^{\pi}(\varepsilon)\}$ iterations.
\end{proof}  

Let $T := \max\{\max A_{d,h} - \min A_{d,h} : (d,h) \in D \times H\}$. The following table summarizes the complexity results found for zero-sum matching games.

\begin{table}[H]
    \centering
    \begin{tabular}{c|c|c}
    \toprule
    Algorithms & Complexity per It & Max Number of It\\
    \toprule
    DAC & $\mathcal{O}((N^2+k^2)L)$ & $T/\varepsilon$\\
    \hline
    Market procedure  & - & $p(N)$\\ 
    Demand sets & \multirow{1}{*}{$\mathcal{O}(N^2L)$} & \multirow{1}{*}{-} \\
    Implementation & $\mathcal{O}(Nk^2L)$ & -\\
    \hline
    Renegotiation process & $\mathcal{O}(N^4k^{3.5}L)$ & $T/\varepsilon$\\
    \hline
    \end{tabular}
    \caption{Complexity zero-sum matching games: $N$ players per side, $k$ strategies per player, $L$ bits to encode the data, and $p(N)$ a polynomial on $N$.}
    \label{tab:complexity_zero_sum}
\end{table}


\section{Conclusions}\label{sec:conclusions}

We have extended the model of one-to-one matching games under commitment of Garrido-Lucero and Laraki \cite{garrido2025stable} to a general framework in which hospitals can be matched with many doctors at the time and doctors care about their colleagues, capturing many models from the literature including the \textit{matching job market} model of Kelso and Crawford \cite{kelso_jr_job_1982}, the \textit{matching with contracts} model of Hatfield and Milgrom \cite{hatfield_matching_2005}, \textit{hedonic games} \cite{dreze_hedonic_1980}, \textit{roommates problem} of Gale and Shapley \cite{gale_college_1962,irving_efficient_1985,knuth_marriages_1976}, the \textit{roommates problem with transferable utility} \cite{andersson_competitive_2014,chiappori_roommate_2014,eriksson_stable_2001,klaus_consistency_2010,talman_model_2011}, \textit{the roommates problem with non-transferable utility} \cite{alkan_pairing_2014}. 

We have studied the existence of core stable and renegotiation proof allocations for two submodels of general matching games: \textit{additive separable} and \textit{roommates}. For the additive separable matching games submodel, we have adapted the deferred-acceptance with competitions algorithm to compute a core stable allocation whenever all strategy sets are compact, payoff functions are continuous and, in addition, the sets of Pareto-optimal strategy profiles are closed. For the roommates submodel, we have leveraged the work of Alkan and Tuncay \cite{alkan_pairing_2014}. Using their market procedure we are able to compute payoff profiles that, whenever an allocation $\pi$ can implement them, $\pi$ results to be \textit{core stable}. Moreover, the market procedure output $\overline{f}$ is realizable, i.e., we can ensure the existence of core stable allocations, every time that the demand graph at $\overline{f}$ can be decomposed in the disjoint union of even-cycles. 

Regarding renegotiation proofness, we have extended the results in \cite{garrido2025stable} to the two submodels. For the first model the extension is a consequence of the additive separability. For the second model, as agents keep matching in couples, the extension of renegotiation proofness is straightforward. For both submodels we have studied how to compute the agents' \textit{reservation payoffs}, the \textit{constrained Nash equilibria}, and we have adapted the \textit{renegotiation process}. As for one-to-one matching games, we can compute renegotiation proof allocations whenever players play zero-sum games with a value, strictly competitive games with an equilibrium, infinitely repeated games, and potential games.

Subsequently, we have provided the complexity study of the algorithms for three classes of bi-matrix matching games: zero-sum, infinitely repeated (\Cref{sec:complexity_infinitely_repeated_games}), and strictly competitive (\Cref{sec:complexity_strictly_competitive_games}). For the additive separable matching games submodel, we have proved that our algorithms converge to an $\varepsilon$-pairwise-renegotiation proof allocation in the three classes of matching games in a bounded number of iterations, with a bound only depending on $\varepsilon$. Each iteration of the algorithms having a polynomial complexity, the deferred-acceptance with competitions algorithm and renegotiation process are efficient algorithms. For the roommates submodel, we have proved that for zero-sum games, strictly competitive games, and infinitely repeated games, the computations of the demand sets and the allocation that implements the output of the market procedure (if it exists) of Alkan and Tuncay \cite{alkan_pairing_2014} are polynomial problems over the number of doctors and their number of pure strategies. Moreover, we have given a procedure to determine the existence of an allocation implementing the output of the market procedure for zero-sum and strictly competitive matching games. 

Alkan and Tuncay's proofs for the correctness and complexity of the market procedure and direction procedure have been given only in the quasi-linear case. Nevertheless, they claim their procedures achieve the same results in the general non-quasi-linear setting. A deeper analysis of these results is in the list of incoming works.

The $\varepsilon$-renegotiation process works as well for the roommates submodel. Moreover, its complexity results hold for roommates matching games with couples playing zero-sum, strictly competitive, and infinitely repeated games. Therefore, using the complexity results claimed by Alkan and Tuncay together with our results, we conclude that the mechanism designed to compute pairwise stable and renegotiation proof allocations for the roommates submodel is efficient.

A possible future research line is the mix between the two one-to-many matching games submodels, roommates and additive separable matching games, studied in this article. The mixed setting would define a model in which couples of doctors are assigned to hospitals, and agents' utilities within the same triplet depend on the strategies and identities of the partners. Interesting applications may rise from this mix, for example, the siblings' schools allocation problem of Correa et al. \cite{correa_school_2019}.

\bibliographystyle{abbrv}
\bibliography{Zotero.bib}

\begin{thebibliography}{10}

\bibitem{adler_note_2009}
I.~Adler, C.~Daskalakis, and C.~H. Papadimitriou.
\newblock A note on strictly competitive games.
\newblock In {\em International {Workshop} on {Internet} and {Network}
  {Economics}}, pages 471--474. Springer, 2009.

\bibitem{ahani_dynamic_2021}
N.~Ahani, P.~Gölz, A.~D. Procaccia, A.~Teytelboym, and A.~C. Trapp.
\newblock Dynamic {Placement} in {Refugee} {Resettlement}.
\newblock In {\em {ACM} {Conference} on {Economics} and {Computation}}, 2021.

\bibitem{akbarpour_unpaired_2020}
M.~Akbarpour, J.~Combe, Y.~He, V.~Hiller, R.~Shimer, and O.~Tercieux.
\newblock Unpaired kidney exchange: {Overcoming} double coincidence of wants
  without money.
\newblock In {\em Proceedings of the 21st {ACM} {Conference} on {Economics} and
  {Computation}}, pages 465--466, 2020.

\bibitem{akbarpour_centralized_2022}
M.~Akbarpour, A.~Kapor, C.~Neilson, W.~Van~Dijk, and S.~Zimmerman.
\newblock Centralized school choice with unequal outside options.
\newblock {\em Journal of Public Economics}, 210:104644, 2022.
\newblock Publisher: Elsevier.

\bibitem{alkan_pairing_2014}
A.~Alkan and A.~Tuncay.
\newblock Pairing games and markets.
\newblock 2014.
\newblock Publisher: FEEM Working Paper.

\bibitem{andersson_competitive_2014}
T.~Andersson, J.~Gudmundsson, D.~Talman, and Z.~Yang.
\newblock A competitive partnership formation process.
\newblock {\em Games and Economic Behavior}, 86:165--177, 2014.
\newblock Publisher: Elsevier.

\bibitem{anstreicher_convex_2012}
K.~M. Anstreicher.
\newblock On convex relaxations for quadratically constrained quadratic
  programming.
\newblock {\em Mathematical programming}, 136(2):233--251, 2012.
\newblock Publisher: Springer.

\bibitem{aumann_almost_1961}
R.~J. Aumann.
\newblock Almost strictly competitive games.
\newblock {\em Journal of the Society for Industrial and Applied Mathematics},
  9(4):544--550, 1961.
\newblock Publisher: SIAM.

\bibitem{aumann_core_1961}
R.~J. Aumann.
\newblock The core of a cooperative game without side payments.
\newblock {\em Transactions of the American Mathematical Society},
  98(3):539--552, 1961.

\bibitem{aumann_long-term_1994}
R.~J. Aumann and L.~S. Shapley.
\newblock Long-term competition—a game-theoretic analysis.
\newblock In {\em Essays in game theory}, pages 1--15. Springer, 1994.

\bibitem{aygun_matching_2012}
O.~Aygün and T.~Sönmez.
\newblock Matching with {Contracts}: {The} {Critical} {Role} of {Irrelevance}
  of {Rejected} {Contracts}.
\newblock Technical report, Boston College Department of Economics, 2012.

\bibitem{balinski_stable_1997}
M.~Balinski and G.~Ratier.
\newblock Of stable marriages and graphs, and strategy and polytopes.
\newblock {\em SIAM review}, 39(4):575--604, 1997.
\newblock Publisher: SIAM.

\bibitem{balinski_graphs_1998}
M.~Balinski and G.~Ratier.
\newblock Graphs and marriages.
\newblock {\em The American mathematical monthly}, 105(5):430--445, 1998.
\newblock Publisher: Taylor \& Francis.

\bibitem{banerjee_ride_2019}
S.~Banerjee and R.~Johari.
\newblock Ride sharing.
\newblock In {\em Sharing economy}, pages 73--97. Springer, 2019.

\bibitem{baiou_stable_2000}
M.~Baïou and M.~Balinski.
\newblock The stable admissions polytope.
\newblock {\em Mathematical programming}, 87(3):427--439, 2000.
\newblock Publisher: Springer.

\bibitem{baiou_student_2004}
M.~Baïou and M.~Balinski.
\newblock Student admissions and faculty recruitment.
\newblock {\em Theoretical Computer Science}, 322(2):245--265, 2004.
\newblock Publisher: Elsevier.

\bibitem{chambers_revealed_2016}
C.~P. Chambers and F.~Echenique.
\newblock {\em Revealed preference theory}, volume~56.
\newblock Cambridge University Press, 2016.

\bibitem{chiappori_roommate_2014}
P.-A. Chiappori, A.~Galichon, and B.~Salanié.
\newblock The roommate problem is more stable than you think.
\newblock {\em Available at SSRN 2411449}, 2014.

\bibitem{coles_marketplaces_1998}
M.~G. Coles and E.~Smith.
\newblock Marketplaces and {Matching}.
\newblock {\em International Economic Review}, 39:239--254, 1998.

\bibitem{combe_design_2022}
J.~Combe, O.~Tercieux, and C.~Terrier.
\newblock The design of teacher assignment: {Theory} and evidence.
\newblock {\em The Review of Economic Studies}, 89(6):3154--3222, 2022.
\newblock Publisher: Oxford University Press.

\bibitem{comission_commission_2023}
E.~Comission.
\newblock Commission proposes new measures on skills and talent to help address
  critical labour shortages.
\newblock 2023.

\bibitem{correa_school_2019}
J.~Correa, R.~Epstein, J.~Escobar, I.~Rios, B.~Bahamondes, C.~Bonet,
  N.~Epstein, N.~Aramayo, M.~Castillo, A.~Cristi, and {others}.
\newblock School choice in {Chile}.
\newblock In {\em Proceedings of the 2019 {ACM} {Conference} on {Economics} and
  {Computation}}, pages 325--343, 2019.

\bibitem{crawford_job_1981}
V.~P. Crawford and E.~M. Knoer.
\newblock Job matching with heterogeneous firms and workers.
\newblock {\em Econometrica: Journal of the Econometric Society}, pages
  437--450, 1981.
\newblock Publisher: JSTOR.

\bibitem{demange_strategy_1985}
G.~Demange and D.~Gale.
\newblock The strategy structure of two-sided matching markets.
\newblock {\em Econometrica: Journal of the Econometric Society}, pages
  873--888, 1985.
\newblock Publisher: JSTOR.

\bibitem{demange_multi-item_1986}
G.~Demange, D.~Gale, and M.~Sotomayor.
\newblock Multi-item auctions.
\newblock {\em Journal of Political Economy}, 94(4):863--872, 1986.
\newblock Publisher: The University of Chicago Press.

\bibitem{dreze_hedonic_1980}
J.~H. Dreze and J.~Greenberg.
\newblock Hedonic coalitions: {Optimality} and stability.
\newblock {\em Econometrica: Journal of the Econometric Society}, pages
  987--1003, 1980.
\newblock Publisher: JSTOR.

\bibitem{daspremont_relaxations_2003}
A.~d’Aspremont and S.~Boyd.
\newblock Relaxations and randomized methods for nonconvex {QCQPs}.
\newblock {\em EE392o Class Notes, Stanford University}, 1:1--16, 2003.

\bibitem{echenique_ordinal_2017}
F.~Echenique and A.~Galichon.
\newblock Ordinal and cardinal solution concepts for two-sided matching.
\newblock {\em Games and Economic Behavior}, 101:63--77, 2017.
\newblock Publisher: Elsevier.

\bibitem{echenique_core_2002}
F.~Echenique and J.~Oviedo.
\newblock Core many-to-one matchings by fixed-point methods.
\newblock 2002.
\newblock Publisher: California Institute of Technology.

\bibitem{eriksson_stable_2001}
K.~Eriksson and J.~Karlander.
\newblock Stable outcomes of the roommate game with transferable utility.
\newblock {\em International Journal of Game Theory}, 29(4):555--569, 2001.
\newblock Publisher: Springer.

\bibitem{gale_college_1962}
D.~Gale and L.~S. Shapley.
\newblock College admissions and the stability of marriage.
\newblock {\em The American Mathematical Monthly}, 69(1):9--15, 1962.
\newblock Publisher: Taylor \& Francis.

\bibitem{galichon_costly_2019}
A.~Galichon, S.~D. Kominers, and S.~Weber.
\newblock Costly concessions: {An} empirical framework for matching with
  imperfectly transferable utility.
\newblock {\em Journal of Political Economy}, 127(6):2875--2925, 2019.
\newblock Publisher: The University of Chicago Press Chicago, IL.

\bibitem{garrido2025stable}
F.~Garrido-Lucero and R.~Laraki.
\newblock Stable matching games.
\newblock {\em Social Choice and Welfare}, pages 1--39, 2025.

\bibitem{garrido2025two}
F.~Garrido-Lucero, D.~Sokolov, P.~Loiseau, and S.~Mauras.
\newblock Two-sided matching with resource-regional caps.
\newblock {\em arXiv preprint arXiv:2502.14690}, 2025.

\bibitem{gilles_core_2010}
R.~P. Gilles and R.~P. Gilles.
\newblock The core of a cooperative game.
\newblock {\em The cooperative game theory of networks and hierarchies}, pages
  29--70, 2010.
\newblock Publisher: Springer.

\bibitem{gimbert_constrained_2021}
H.~Gimbert, C.~Mathieu, and S.~Mauras.
\newblock Constrained {School} {Choice} with {Incomplete} {Information}.
\newblock {\em arXiv preprint arXiv:2109.09089}, 2021.

\bibitem{hatfield_matching_2008}
J.~W. Hatfield and F.~Kojima.
\newblock Matching with contracts: {Comment}.
\newblock {\em American Economic Review}, 98(3):1189--94, 2008.

\bibitem{hatfield_substitutes_2010}
J.~W. Hatfield and F.~Kojima.
\newblock Substitutes and stability for matching with contracts.
\newblock {\em Journal of Economic theory}, 145(5):1704--1723, 2010.
\newblock Publisher: Elsevier.

\bibitem{hatfield_hidden_2015}
J.~W. Hatfield and S.~D. Kominers.
\newblock Hidden {Substitutes}.
\newblock In {\em {EC}}, page~37, 2015.

\bibitem{hatfield_stability_2013}
J.~W. Hatfield, S.~D. Kominers, A.~Nichifor, M.~Ostrovsky, and A.~Westkamp.
\newblock Stability and competitive equilibrium in trading networks.
\newblock {\em Journal of Political Economy}, 121(5):966--1005, 2013.
\newblock Publisher: University of Chicago Press Chicago, IL.

\bibitem{hatfield_stability_2021}
J.~W. Hatfield, S.~D. Kominers, and A.~Westkamp.
\newblock Stability, strategy-proofness, and cumulative offer mechanisms.
\newblock {\em The Review of Economic Studies}, 88(3):1457--1502, 2021.
\newblock Publisher: Oxford University Press.

\bibitem{hatfield_matching_2005}
J.~W. Hatfield and P.~R. Milgrom.
\newblock Matching with contracts.
\newblock {\em American Economic Review}, 95(4):913--935, 2005.

\bibitem{hildenbrand_core_1982}
W.~Hildenbrand.
\newblock Core of an economy.
\newblock {\em Handbook of mathematical economics}, 2:831--877, 1982.
\newblock Publisher: Elsevier.

\bibitem{irving_efficient_1985}
R.~W. Irving.
\newblock An efficient algorithm for the “stable roommates” problem.
\newblock {\em Journal of Algorithms}, 6(4):577--595, 1985.
\newblock Publisher: Elsevier.

\bibitem{irving_complexity_1986}
R.~W. Irving and P.~Leather.
\newblock The complexity of counting stable marriages.
\newblock {\em SIAM Journal on Computing}, 15(3):655--667, 1986.
\newblock Publisher: SIAM.

\bibitem{karmarkar_new_1984}
N.~Karmarkar.
\newblock A new polynomial-time algorithm for linear programming.
\newblock In {\em Proceedings of the sixteenth annual {ACM} symposium on
  {Theory} of computing}, pages 302--311, 1984.

\bibitem{kelso_jr_job_1982}
A.~S. Kelso~Jr and V.~P. Crawford.
\newblock Job matching, coalition formation, and gross substitutes.
\newblock {\em Econometrica: Journal of the Econometric Society}, pages
  1483--1504, 1982.
\newblock Publisher: JSTOR.

\bibitem{khachiyan_polynomial_1979}
L.~G. Khachiyan.
\newblock A polynomial algorithm in linear programming.
\newblock In {\em Doklady {Akademii} {Nauk}}, volume 244, pages 1093--1096.
  Russian Academy of Sciences, 1979.
\newblock Issue: 5.

\bibitem{klaus_consistency_2010}
B.~Klaus and A.~Nichifor.
\newblock Consistency in one-sided assignment problems.
\newblock {\em Social Choice and Welfare}, 35(3):415--433, 2010.
\newblock Publisher: Springer.

\bibitem{knuth_marriages_1976}
D.~E. Knuth.
\newblock Marriages stables.
\newblock {\em Technical report}, 1976.
\newblock Publisher: Les Presses de l'université de Montréal.

\bibitem{linderoth_simplicial_2005}
J.~Linderoth.
\newblock A simplicial branch-and-bound algorithm for solving quadratically
  constrained quadratic programs.
\newblock {\em Mathematical programming}, 103(2):251--282, 2005.
\newblock Publisher: Springer.

\bibitem{mehta_online_2013}
A.~Mehta.
\newblock Online {Matching} and {Ad} {Allocation}.
\newblock {\em Found. Trends Theor. Comput. Sci.}, 8:265--368, 2013.

\bibitem{papadimitriou_complexity_2007}
C.~H. Papadimitriou.
\newblock The complexity of finding {Nash} equilibria.
\newblock {\em Algorithmic game theory}, 2:30, 2007.
\newblock Publisher: Cambridge Univ. Press Cambridge.

\bibitem{rochford_symmetrically_1984}
S.~C. Rochford.
\newblock Symmetrically pairwise-bargained allocations in an assignment market.
\newblock {\em Journal of Economic Theory}, 34(2):262--281, 1984.
\newblock Publisher: Elsevier.

\bibitem{rothblum_characterization_1992}
U.~G. Rothblum.
\newblock Characterization of stable matchings as extreme points of a polytope.
\newblock {\em Mathematical Programming}, 54(1-3):57--67, 1992.
\newblock Publisher: Springer.

\bibitem{shapley_assignment_1971}
L.~S. Shapley and M.~Shubik.
\newblock The assignment game {I}: {The} core.
\newblock {\em International Journal of game theory}, 1(1):111--130, 1971.
\newblock Publisher: Springer.

\bibitem{shioura_partnership_2017}
A.~Shioura and {others}.
\newblock On the {Partnership} formation problem.
\newblock {\em Journal of Mechanism and Institution Design}, 2:1, 2017.

\bibitem{talman_model_2011}
D.~Talman and Z.~Yang.
\newblock A model of partnership formation.
\newblock {\em Journal of Mathematical Economics}, 47(2):206--212, 2011.
\newblock Publisher: Elsevier.

\bibitem{tan_necessary_1991}
J.~J. Tan.
\newblock A necessary and sufficient condition for the existence of a complete
  stable matching.
\newblock {\em Journal of Algorithms}, 12(1):154--178, 1991.
\newblock Publisher: Elsevier.

\bibitem{vaidya_speeding-up_1989}
P.~M. Vaidya.
\newblock Speeding-up linear programming using fast matrix multiplication.
\newblock In {\em 30th annual symposium on foundations of computer science},
  pages 332--337. IEEE Computer Society, 1989.

\end{thebibliography}

\appendix

\section{Complexity study for Infinitely repeated matching games}\label{sec:complexity_infinitely_repeated_games}

For each potential pair $(d,h) \in D \times H$, let $G_{d,h} = (X_d,Y_h, A_{d,h}, B_{d,h})$ be a finite bi-matrix game in mixed strategies, with $X_d = \Delta(S_d), Y_h = \Delta(T_h)$, where all matrices have only rational entries. Given $K \in \mathbb{N}$, consider the $K$-stages game $G_{d,h}^K$ defined by the payoff functions,
$$f_{d,h}(K,\sigma_d, \sigma_h) = \frac{1}{K} \mathbb{E}_{\sigma}\left[\sum_{k=1}^K A_{d,h}(s_k,t_k)\right],\ \  g_{d,h}(K,\sigma_d, \sigma_h) = \frac{1}{K} \mathbb{E}_{\sigma}\left[\sum_{k=1}^K B_{d,h}(s_k,t_k)\right],$$
where $\sigma_d:\bigcup (S_d\times T_h)_{k=1}^{\infty} \rightarrow X_d$ is a behavioral strategy for player $d$ and $\sigma_h:\bigcup (S_d\times T_h)_{k=1}^{\infty} \rightarrow Y_h$ is a behavioral strategy for player $h$. We define the uniform game $G_{d,h}^{\infty}$ as the limit of $G_{d,h}^K$ when $K$ goes to infinity.

\begin{definition}
\textup{A matching game $\Gamma$ is a \textbf{bi-matrix infinitely repeated matching game} if every strategic game is a uniform game as explained above.}
\end{definition}

To study the complexity of computing pairwise-renegotiation proof allocations in infinitely repeated games, we start by studying the complexity of solving the general QCQP Problem (\ref{def:general_opt_problem}). We state the proof for a pair (doctor, hospital) although it can be straightforwardly applied to roommates.

\begin{proposition}\label{prop:QCQP_problem_is_polynomial_for_inf_repeated_games}
Let $(d,h) \in D \times H$ be a pair, $G_{d,h} = (X_d,Y_h, A_{d,h}, B_{d,h})$ their finite stage game in mixed strategies and $c \in \mathbb{R}$, such that $c \leq \max B_{d,h}$. The complexity of solving the QCQP Problem (\ref{def:general_opt_problem}) in $G_{d,h}^{\infty}$ is $\mathcal{O}\left((|S_d|\cdot|T_h|)^{2.5}L_{d,h} \right)$, where $L_{d,h}$ is the number of bits required to encode the stage game.
\end{proposition}

To prove \Cref{prop:QCQP_problem_is_polynomial_for_inf_repeated_games} we will use the following result.

\begin{lemma}\label{lemma:feasible_payoffs_are_computed_in_pol_time}
Let $(d,h) \in D \times H$ be a pair and let $(\overline{f}, \overline{g}) \in \mathbb{R}^2$ be a payoff vector in the set of feasible payoffs,
$$co(A_{d,h},B_{d,h}) := \text{convex hull}\{(A_{d,h}(s,t) B_{d,h}(s,t)) \in \mathbb{R}^2 : s \in S_d, t \in T_h\}.$$
Then, there exists a pure strategy profile $\sigma$ of $G_{d,h}^{\infty}$ that achieves $(\overline{f}, \overline{g})$. In addition, the number of elementary operations used to compute $\sigma$ is bounded by $\mathcal{O}((|S_d|\cdot|T_h|)^{2.5}L_{d,h})$, where $L_{d,h}$ is the number of bits required to encode the matrices $A_{d,h}$ and $B_{d,h}$.
\end{lemma}

\begin{proof}
Consider the following system with $|S_d|\cdot|T_h|$ variables and three linear equations,
\begin{align}
    \begin{split}\label{def:system_linear_equations_repeated_games}
    \sum_{s,t} A_{d,h}(s,t)\lambda_{s,t} &= \overline{f},\\ \sum_{s,t} B_{d,h}(s,t)\lambda_{s,t} &= \overline{g},\ \ \lambda \in \Delta(S_d \times T_h).
    \end{split}
\end{align}
System $(\ref{def:system_linear_equations_repeated_games})$ can be solved in $\mathcal{O}((|S_d|\cdot|T_h|)^{2.5}L_{d,h})$ elementary operations. Since matrices $A_{d,h}$ and $B_{d,h}$ have rational entries, the solution has the form $(\lambda_{s,t})_{s,t} = (\frac{p_{s,t}}{q_{s,t}})_{s,t}$ with each $p_{s,t}, q_{s,t} \in \mathbb{N}$. Let $N_{\lambda} = \text{lcm}(q_{s,t} : (s,t) \in S_d \times T_h)$ be the least common multiple of all denominators. The number of elementary operations to compute $N_{\lambda}$ is bounded by $\mathcal{O}\left((|S_d|\cdot|T_h|)^{2}\right)$. Enlarge each fraction of the solution so all denominators are equal to $N_{\lambda}$, i.e. $\lambda = (\frac{p'_{s,t}}{N_{\lambda}})_{s,t}$. Suppose that $S_d = \{s_1,s_2,...,s_d\}$ and $T_h = \{t_1, t_2, ..., t_h\}$. Let $\sigma$ be the strategy profile in which players play $(s_1,t_1)$ the first $p'_{s_1,t_1}$-stages, $(s_1,t_2)$ the next $p'_{s_1,t_2}$-stages, $(s_1,t_3)$ the next $p'_{s_1,t_3}$-stages, up to playing $(s_d,t_h)$ during $p'_{s_d,t_h}$-stages, infinitely repeated. By construction, $(f_{d,h}(\sigma), g_{d,h}(\sigma)) = (\overline{f}, \overline{g})$.
\end{proof}

Let us illustrate the previous result with an example.

\begin{example}
\textup{Consider the following prisoners' dilemma $G$ played infinitely many times by a couple $(d,h)$. 
\begin{table}[H]
\centering
\begin{tabular}{cc|c|c}
     & \multicolumn{3}{c}{Agent h}\\\noalign{\vskip 2mm}
     \multirow{3}{*}{Agent d  } & & Cooperate & Betray \\
     \cline{2-4}
     & Cooperate & $2, 2$ & $-1,3$ \\ 
     \cline{2-4}
     & Betray & $3,-1$ & $0,0$
\end{tabular}
\hspace{1.5cm}
\end{table}
\Cref{fig:prisoners_dilemma_2} shows the convex envelope of the pure payoff profiles. 
\begin{figure}[ht]
    \centering
    \includegraphics[scale = 0.8]{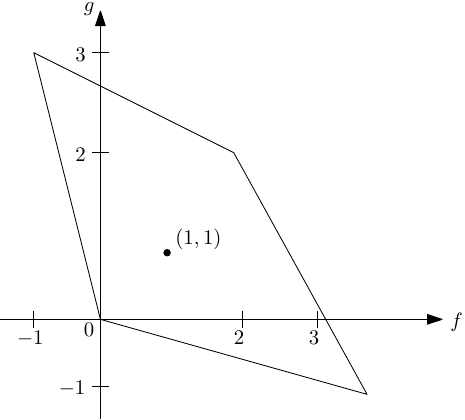}
    \caption{Prisoners' dilemma payoff profiles}
    \label{fig:prisoners_dilemma_2}
\end{figure}
\noindent Consider $(\bar{f},\bar{g}) = (1,1) \in co(A_{d,h}, B_{d,h})$, represented in the figure by the star. Note that $(1,1)$ can be obtained as the convex combination of $\frac{1}{4} (0,0) + \frac{1}{4}(3,-1) + \frac{1}{4}(-1,3) + \frac{1}{4}(2,2)$. Therefore, $(d,h)$ can obtain $(1,1)$ in their infinitely repeated game by playing $(B,B)$ the first four rounds, $(C,B)$ the second four rounds, $(B,C)$ the third four rounds, $(C,C)$ the fourth four rounds, and cycling like this infinitely many times. As every $16$ rounds the couple obtains $(1,1)$, in the limit, their average payoff converges to $(\bar{f}, \bar{g})$.} 
\end{example}

Finally, we prove the complexity result of solving the QCQP problem (\Cref{prop:QCQP_problem_is_polynomial_for_inf_repeated_games}).

\begin{proof}[Proof of \Cref{prop:QCQP_problem_is_polynomial_for_inf_repeated_games}.]
Consider the following optimization problem,
\begin{align}\label{def:opt_problem_1_inf_rep_games}
    \max\biggl\{ \sum_{s \in S_d} \sum_{t \in T_h} A_{d,h}(s,t) \lambda_{s,t} \mid \sum_{s \in S_d} \sum_{t \in T_h} B_{d,h}(s,t)\lambda_{s,t} \geq c, \lambda \in \Delta(S_d \times T_h)\biggr\}.
\end{align}
Problem (\ref{def:opt_problem_1_inf_rep_games}) is a linear programming problem with $|S_d|\cdot|T_h|$ variables and two constraints and its optimal value $(\overline{f},\overline{g})$ coincides with the optimal value of the QCQP Problem (\ref{def:general_opt_problem}). Therefore, any strategy profile $\sigma$ that achieves $(\overline{f},\overline{g})$, is a solution of the QCQP Problem (\ref{def:general_opt_problem}). The stated complexity is obtained from solving Problem (\ref{def:opt_problem_1_inf_rep_games}) and applying \Cref{lemma:feasible_payoffs_are_computed_in_pol_time} to compute $\sigma$.
\end{proof}

\subsection{Deferred-acceptance with competitions algorithm}

The polynomial complexity of solving the QCQP general problem (\Cref{prop:QCQP_problem_is_polynomial_for_inf_repeated_games}) allows us to prove the main result of this section.

\begin{theorem}[Complexity]\label{teo:propose_dispose_algo_is_polynomial_for_inf_rep_games}
Let $d \in D$ be the proposer doctor. Let $h$ be the proposed hospital and $d'$ be the doctor that $d$ wants to replace. If $d$ is the winner of the competition, the entire iteration of the DAC algorithm (\Cref{Algo:Propose_dispose_algo_additive_separable_case}) has complexity,
$$\mathcal{O}\left(|H|\cdot|D| + |S_d|^{2.5}\sum_{h' \in H} |T_{h'}|^{2.5}L_{d,{h'}} + |S_{d'}|^{2.5}\sum_{h' \in W} |T_{h'}|^{2.5}L_{{d'},{h'}} \right),$$
where $L_{i,j}$ is the number of bits required to encode the payoff matrices of $(i,j)$. 
\end{theorem}

\begin{proof}
The optimal proposal problem is split into $|H|$ problems. Each subproblem needs $|D|$ comparisons to compute the right-hand side and then, they have the complexity stated in \Cref{prop:QCQP_problem_is_polynomial_for_inf_repeated_games}. Thus, the optimal proposal computation has complexity, 
$$\mathcal{O}\left(|H|\cdot|D| + \sum_{h' \in H}(|S_d|\cdot|T_{h'}|)^{2.5}L_{d,h'} \right).$$ 
Computing the reservation payoff and the bid of each competitor has exactly the same complexity as the optimal proposal computation, considering the respective set of strategies. Finally, the problem solved by the winner has complexity $\mathcal{O}((|S_d|\cdot|T_{h}|)^{2.5}$ $\cdot L_{d,h})$. Summing up, we obtain the complexity stated in the theorem.
\end{proof}

\begin{remark}
If there are at most $N$ players in each side and at most $k$ pure strategies per player, \Cref{teo:propose_dispose_algo_is_polynomial_for_inf_rep_games} proves that each iteration of the DAC algorithm (\Cref{Algo:Propose_dispose_algo_additive_separable_case}) takes $\mathcal{O}(N^2+Nk^5L)$ number of elementary operations in being solved, hence it is polynomial. As the number of iterations is bounded by $T \propto \frac{1}{\varepsilon}$ (\Cref{teo:propose_dipose_algo_number_of_iterations}), we conclude that computing an $\varepsilon$-pairwise stable allocation for a infinitely repeated matching game is a polynomial problem.
\end{remark}

\subsection{Market procedure}

Let $\Gamma$ be a roommates matching game such that for each couple $(d,d')\in D \times D$, their game $G_{d,d'}$ is an infinitely repeated game with a bi-matrix stage game. Consider a payoff profile $f \in \mathbb{R}^{|D|}$ and a fixed doctor $d \in D$. Note that,
\begin{align*}
    d' \in P_d(f) &\Longleftrightarrow f_d = u_{d,d'}(f_{d'})\\
                  &\Longleftrightarrow (f_d,f_{d'}) \in co(A_{d,d'},A_{d',d})\\
                  &\Longleftrightarrow \exists \lambda = (\lambda_{s,s'})_{s \in S_d,s' \in S_{d'}} \subseteq [0,1], \sum_{s,s'} \lambda_{s,s'} = 1:\\
                  &\qquad f_d = \sum_{s,s'} A_{d,d'}(s,s') \lambda_{s,s'} \text{ and } f_{d'} = \sum_{s,s'} A_{d',d}(s,s') \lambda_{s,s'}.
\end{align*}
Therefore, to determine if a doctor belongs to the demand set of $d$, it is enough with computing the coefficients of their convex combination which, we know, has a polynomial complexity (\Cref{lemma:feasible_payoffs_are_computed_in_pol_time}). We can conclude the following result.

\begin{theorem}
The complexity of computing the demand sets of all doctors during an iteration of the market procedure has complexity
$$\mathcal{O}\left(\sum_{(d,d') \in D \times D} |S_d|\cdot|S_{d'}| L_{d,d'}\right),$$
where $L_{d,d'}$ is the number of bits required to encode the payoff matrices of the pair $(d,d')$.
\end{theorem}

Remark the previous result can be refined as we only need to check once every couple, since belonging to the demand set is a symmetric property. Therefore, once determined that $d' \in P_d(f)$, we directly obtain that $d \in P_{d'}(f)$.

Regarding the implementation of the market procedure's output, once having the good pairs of doctors given by a matching $\mu$, it is enough with solving the linear system of equations from \Cref{lemma:feasible_payoffs_are_computed_in_pol_time} for each of the couples. In particular, we can conclude the following result.

\begin{theorem}
Let $f \in \mathbb{R}^{|D|}$ be the output of the market procedure and suppose there exists a matching $\mu$ such that for any doctor $d \in D$, $(f_d,f_{\mu(d)}) \in co(A_{d,\mu(d)},A_{\mu(d),d})$, if $d$ is matched, and $f_d = \underline{f}_d$, if $d$ is unmatched. Then, computing the strategy profile $\vec{x} \in X_D$ such that $f_{d,\mu(d)} (x_d,x_{\mu(d)}) = f_d$, for any $d$ matched, has complexity
$$\mathcal{O}\left(\sum_{(d,d') \in \mu} |S_{d}|\cdot|S_{d'}|L_{d,d'} \right),$$
where $L_{d,d'}$ is the number of bits required to encode the data of the stage game $G_{d,d'}$.
\end{theorem}

\begin{remark}
If there are at most $N$ doctors and $k$ pure strategies per doctor, computing all the demand sets during an iteration of the market procedure has complexity $\mathcal{O}(N^2k^2L)$. Similarly, in case the outcome of the market procedure can be implemented by an allocation, finding the strategy profiles of the agents within the allocation has complexity $\mathcal{O}(Nk^2L)$.
\end{remark}

\subsection{Renegotiation process}

We introduce some classical definitions from uniform games. Let 
\begin{align*}
    co(A_{d,h},B_{d,h}) :=\text{convex hull}\{(A_{d,h}(s,t) B_{d,h}(s,t)) \in \mathbb{R}^2 : s \in S_d, t \in T_h\},
\end{align*}
to be the set of \textit{feasible payoffs}, and $\alpha_{d,h},\beta_{d,h}$ be, respectively, the \textit{punishment levels}of players $d \in D$ and $h \in H$, defined by,
\begin{align*}
    \alpha_{d,h} &:= \min_{y \in \Delta(T_h)}\max_{x \in \Delta(S_d)} x A_{d,h} y,\\
    \beta_{d,h} &:= \min_{x \in \Delta(S_t)}\max_{y \in \Delta(T_h)} x B_{d,h} y.
\end{align*}

We define the set of \textit{uniform equilibrium payoffs} by,
\begin{align*}
    E_{d,h} := \{(\overline{f},\overline{g}) \in co(A_{d,h},B_{d,h}) : \overline{f}\geq \alpha_{d,h}, \overline{g} \geq \beta_{d,h}\}
\end{align*}

From the Folk theorem of Aumann-Shapley \cite{aumann_long-term_1994}, we know that $E_{d,h}$ is exactly the set of uniform equilibrium payoff of $G_{d,h}^{\infty}$. 

\begin{definition}\label{def:eps_acceptable_payoffs_set}
Let $\pi = (\mu,\Vec{x},\Vec{y})$ be an allocation. For every pair of reservation payoffs $(f_d^{\pi}(\varepsilon),g_{h}^{\pi}(\varepsilon))$ and $\varepsilon \bbi 0$, we define the $\varepsilon$-\textbf{acceptable payoffs set} as
$$E_{d,h}(f_d^{\pi}(\varepsilon),g_{h}^{\pi}(\varepsilon)) := co(A_{d,h},B_{d,h}) \cap \{(\bar{f}, \bar{g}) \in \mathbb{R}^2 : \bar{f} + \varepsilon \geq f^{\pi}_d(\varepsilon), \bar{g} + \varepsilon \geq g_h^{\pi}(\varepsilon)\}.$$
\end{definition}

Finally, we define $\varepsilon$-constrained Nash equilibria for uniform games.

\begin{definition}
\textup{A strategy profile $\sigma = (\sigma_d,\sigma_h)$ is an $\varepsilon$-$(f^{\pi}_d(\varepsilon),g_h^{\pi}(\varepsilon))$-constrained Nash equilibrium of $G_{d,h}^{\infty}$ if,
\begin{itemize}
    \item[1.] For ay $\overline{\varepsilon} \bbi \varepsilon$, there exists $K_0 \in \mathbb{N}$, such that for any $K \geq K_0$ and any $(\tau_d,\tau_h)$,
    \begin{itemize}
        \item[(a)] if $f_{d,h}(K,\tau_d,\sigma_h) \bbi f_{d,h}(K,\sigma) + \overline{\varepsilon}$ then, $g_{d,h}(K,\tau_d, \sigma_h) + \varepsilon \sm g_h^{\pi}(\varepsilon)$,
        \item[(b)] if $g_{d,h}(K,\sigma_d,\tau_h) \bbi g_{d,h}(K,\sigma) + \overline{\varepsilon}$ then, $f_{d,h}(K,\sigma_d, \tau_h) + \varepsilon \sm f^{\pi}_d(\varepsilon)$.
    \end{itemize}
    \item[2.] It holds $(f_{d,h}(K,\sigma), g_{d,h}(K,\sigma)) \xrightarrow{K\to \infty}(f_{d,h}(\sigma), g_{d,h}(\sigma)) \in \mathbb{R}^2$, with $f_{d,h}(\sigma) + \varepsilon \geq f^{\pi}_d(\varepsilon)$, and $g_{d,h}(\sigma) + \varepsilon \geq g_h^{\pi}(\varepsilon)$.
\end{itemize}
The set of $\varepsilon$-$(f^{\pi}_d(\varepsilon),g_h^{\pi}(\varepsilon))$-CNE payoffs is denoted $E^{\infty}_{d,h}(f^{\pi}_d(\varepsilon),g_h^{\pi}(\varepsilon))$.}
\end{definition}

We begin the complexity analysis by studying the computation of $\varepsilon$-CNE.

\begin{theorem}[CNE Complexity]\label{teo:CNE_computation_is_polynomial_for_uniform_games}
Let $G_{d,h}^{\infty}$ be an infinitely repeated game as defined above. Given any players' reservation payoffs $(f^{\pi}_d(\varepsilon),g^{\pi}_h(\varepsilon)) \in \mathbb{R}^2$ such that $E_{d,h}(f^{\pi}_d(\varepsilon),g_h^{\pi}(\varepsilon))$ is non-empty, the complexity of computing an $\varepsilon$-$(f^{\pi}_d(\varepsilon),g^{\pi}_h(\varepsilon))$-CNE is at most,
$\mathcal{O}((|S_d|\cdot|T_h|)^{2.5}L_{d,h})$, where $L_{d,h}$ is the number of bits required to encode the data of the stage game $G_{d,h}$.
\end{theorem}

We split the proof of \Cref{teo:CNE_computation_is_polynomial_for_uniform_games} in the following three lemmas. First, from the Folk theorem of Aumann-Shapley \cite{aumann_long-term_1994}, the following holds.

\begin{lemma}\label{lemma:folk_theorem_CNE}
It holds in $E_{d,h} \cap E_{d,h}(f^{\pi}_d(\varepsilon),g_h^{\pi}(\varepsilon)) \subseteq E^{\infty}_{d,h}(f^{\pi}_d(\varepsilon),g_h^{\pi}(\varepsilon))$.
\end{lemma}

Whenever the intersection in \Cref{lemma:folk_theorem_CNE} is non-empty, there exists a uniform equilibrium payoff profile $(\bar{f}, \bar{g})$ that belongs to $E^{\infty}_{d,h}(f^{\pi}_d(\varepsilon),g_h^{\pi}(\varepsilon))$. Combined with \Cref{lemma:feasible_payoffs_are_computed_in_pol_time} that states the complexity of finding a strategy profile that achieves a given payoff profile, we obtain a uniform equilibrium that achieves $(\bar{f}, \bar{g})$ with the complexity stated in \Cref{teo:CNE_computation_is_polynomial_for_uniform_games}. The following lemma provides sufficient conditions for that intersection to be non-empty.

\begin{lemma}\label{lemma:inf_rep_games_intersection_non_empty}
Let $(f^{\pi}_d(\varepsilon),g^{\pi}_h(\varepsilon))$ be a pair of reservation payoffs such that the set $E_{d,h}(f^{\pi}_d(\varepsilon),g_h^{\pi}(\varepsilon))$ is non-empty. Then, $E_{d,h} \cap E_{d,h}(f^{\pi}_d(\varepsilon),g_h^{\pi}(\varepsilon))$ is non-empty if either $f^{\pi}_d(\varepsilon) - \varepsilon \geq \alpha_{d,h}$ and $g^{\pi}_h(\varepsilon) - \varepsilon \geq \beta_{d,h}$, or $f^{\pi}_d(\varepsilon) - \varepsilon \sm \alpha_{d,h}$ and $g^{\pi}_h(\varepsilon) - \varepsilon \sm \beta_{d,h}$.
\end{lemma}

\begin{proof}
In the first case, $E_{d,h}(f^{\pi}_d(\varepsilon),g_h^{\pi}(\varepsilon))\subseteq E_{d,h}$, thus the intersection between them is equal to $E_{d,h}(f^{\pi}_d(\varepsilon),g_h^{\pi}(\varepsilon))$, which is non-empty. In the second case, $E_{d,h} \subseteq E_{d,h}(f^{\pi}_d(\varepsilon),g_h^{\pi}(\varepsilon))$ and therefore, the intersection is non-empty.
\end{proof}

This yields the two following missing cases.

\begin{lemma}\label{lemma:inf_rep_games_intersection_empty}
Let $(f^{\pi}_d(\varepsilon),g^{\pi}_h(\varepsilon))$ be a pair of reservation payoffs such that the set $E_{d,h}(f^{\pi}_d(\varepsilon),g_h^{\pi}(\varepsilon))$ is non-empty. Then, computing an $\varepsilon$-CNE has complexity $\mathcal{O}((|S_d|\cdot|T_h|)^{2.5}L_{d,h})$ either if $f^{\pi}_d(\varepsilon) - \varepsilon \geq \alpha_{d,h}$ and $g^{\pi}_h(\varepsilon) - \varepsilon \sm \beta_{d,h}$, or $f^{\pi}_d(\varepsilon) - \varepsilon \sm \alpha_{d,h}$ and $g^{\pi}_h(\varepsilon) - \varepsilon \geq \beta_{d,h}$.
\end{lemma}

\begin{proof}
Suppose the first case, $f^{\pi}_d(\varepsilon) - \varepsilon \geq \alpha_{d,h}$ and $g^{\pi}_h(\varepsilon) - \varepsilon \sm \beta_{d,h}$. Let $F  := E_{d,h}(f^{\pi}_d(\varepsilon),g_h^{\pi}(\varepsilon)) \cap E_{d,h}$. If $F$ is non-empty, the result holds from \Cref{lemma:folk_theorem_CNE}. Suppose $F$ is empty and consider the payoff profile $(\bar{f}, \bar{g}) \in co(A_{d,h}, B_{d,h})$ given by 
$$\bar{g} = \max \{ g \in co(A_{d,h},B_{d,h}): \exists f \in \mathbb{R}, (f,g) \in E_{d,h}(f^{\pi}_d(\varepsilon),g_h^{\pi}(\varepsilon))\}.$$ 

Computing $(\bar{f},\bar{g})$ can be done in $\mathcal{O}((|S_d|\cdot|T_h|)^{2.5}L_{d,h})$ elementary operations by solving the system of linear equations with $(\lambda_{s,t})_{s\in S_d,t\in T_h}$ variables (Problem (\ref{def:opt_problem_1_inf_rep_games})) exchanging the roles of the matrices. Shift the payoff profile to $(\bar{f}, \bar{g} + \varepsilon)$, assuming that increasing $\bar{g}$ by $\varepsilon$ does not take the payoff out of the convex envelope (if it does it, $h$ has reached its highest possible payoff, so it does not have any profitable deviation). Let $\sigma$ be a strategy in $G_{d,h}^{\infty}$ that achieves $(\bar{f}, \bar{g} + \varepsilon)$, computable in $\mathcal{O}((|S_d|\cdot|T_h|)^{2.5}L_{d,h})$ (\Cref{lemma:feasible_payoffs_are_computed_in_pol_time}). Consider next $\sigma'$ the strategy profile in which $d$ and $h$ play following $\sigma$ at every stage, such that if $d$ deviates, $h$ punishes her decreasing her payoff to $\alpha_{d,h}$, and if $h$ deviates, $d$ ignores it and keeps playing according to $\sigma$. We claim that $\sigma'$ is an $\varepsilon$-$(f^{\pi}_d(\varepsilon),g^{\pi}_h(\varepsilon))$-CNE. Indeed, it is feasible as their limit payoff profile is $(\bar{f}, \bar{g} + \varepsilon)$. In addition, remark that $s$ does not have profitable deviations as $h$ punishes her and $\bar{f} \geq f^{\pi}_d(\varepsilon) - \varepsilon \geq \alpha_{d,h}$.

Finally, let $K \in \mathbb{N}$ and $\bar{\varepsilon} \bbi \varepsilon$ such that $h$ can deviate at time $K$ and get $g' \geq (\bar{g} + \varepsilon) + \bar{\varepsilon}$. Let $f'$ be the payoff of $d$ until the stage $K$. Note that $(f',g') \in co(A_{d,h}, B_{d,h})$ since $(f',g')$ is an average payoff profile of the $K$-stage game. Suppose that $f' \geq f^{\pi}_d(\varepsilon) - \varepsilon$, so $(f',g') \in E_{d,h}(f^{\pi}_d(\varepsilon),g_h^{\pi}(\varepsilon))$. Then, 
$\bar{f} \geq f' \geq \bar{f} + \varepsilon + \bar{\varepsilon} \bbi \bar{f}$,
which is a contradiction. Therefore, $f' \sm f^{\pi}_d(\varepsilon) - \varepsilon$. Thus, $\sigma'$ is an $\varepsilon$-$(f^{\pi}_d(\varepsilon),g^{\pi}_h(\varepsilon))$-CNE. For the second case in which $f^{\pi}_d(\varepsilon) - \varepsilon \sm \alpha_{d,h}$ and $g^{\pi}_h(\varepsilon) - \varepsilon \geq \beta_{d,h}$, the argument is analogous.
\end{proof}

As all the possible cases are covered by \Cref{lemma:inf_rep_games_intersection_non_empty,lemma:inf_rep_games_intersection_empty}, we conclude the proof of \Cref{teo:CNE_computation_is_polynomial_for_uniform_games} regarding the complexity of computing constrained Nash equilibria. Making a similar computation to the one for zero-sum matching games, we can bound the complexity of an entire iteration of the $\varepsilon$-renegotiation process (\Cref{Algo:strategy_profiles_modification_eps_version}) by,
\begin{align*}
    \mathcal{O}\left(\sum_{d \in D} \left[|H|\cdot|D| + \sum_{h \in H} (|S_d|\cdot|T_h|)^{2.5} + |S_d|\cdot|T_{\mu(d)}|^{2.5}\right]\cdot L \right),
\end{align*}
where the first two terms come from the reservation payoffs computation, the last one from the constrained N as equilibria computation, and $L$ is the number of bits required to encode the entire data. 

\begin{remark}\label{remark:complexity_stra_prof_mod_inf_rept_games}
Considering $N$ agents per side and $k$ pure strategies per player, the complexity of an iteration of the $\varepsilon$-renegotiation process can be bounded by $\mathcal{O}\left(N^3k^5L\right)$.
\end{remark}

Finally, we study the convergence of the algorithm for infinitely repeated games.

\begin{theorem}[Convergence]\label{teo:convergence_algo_2_inf_rep_games}
Let $\pi = (\mu,\sigma_D, \sigma_H)$ be an $\varepsilon$-pairwise stable allocation. Let $(f_d^{\pi}(\varepsilon),g_h^{\pi}(\varepsilon))_{(d,h) \in \mu}$ be the $\varepsilon$-reservation payoffs generated by $\pi$. Then, there exists an oracle for computing $\varepsilon$-CNE such that, starting from $\pi$, the $\varepsilon$-renegotiation process (\Cref{Algo:strategy_profiles_modification_eps_version}) converges in at most $$ \frac{1}{\varepsilon}\left(\max_{(d,h) \in \mu}\{\max\{\alpha_{d,h} - f_d^{\pi}(\varepsilon), \beta_{d,h} - g_h^{\pi}(\varepsilon)\}\}\right),$$
iterations, where $\alpha_{d,h}, \beta_{d,h}$ are the punishment levels of $(d,h)$.
\end{theorem}

\begin{proof}
Let $(d,h) \in \mu$ be a couple and $(f_d^{\pi}(\varepsilon),g_h^{\pi}(\varepsilon))$ be their reservation payoffs at the beginning of \Cref{Algo:strategy_profiles_modification_eps_version}. Note that one of the following four cases must hold:
\smallskip

\noindent 1. $f_d^{\pi}(\varepsilon) - \varepsilon \leq \alpha_{d,h}$ and $g_h^{\pi}(\varepsilon) - \varepsilon \leq \beta_{d,h}$,

\noindent 2. $f_d^{\pi}(\varepsilon) - \varepsilon \geq \alpha_{d,h}$ and $g_h^{\pi}(\varepsilon) - \varepsilon \geq \beta_{d,h}$,

\noindent 3. $f_d^{\pi}(\varepsilon) - \varepsilon \geq \alpha_{d,h}$ and $g_h^{\pi}(\varepsilon) - \varepsilon \sm \beta_{d,h}$,

\noindent 4. $f_d^{\pi}(\varepsilon) - \varepsilon \sm \alpha_{d,h}$ and $g_h^{\pi}(\varepsilon) - \varepsilon \geq \beta_{d,h}$.
\smallskip

\noindent Let $F_{d,h} := E_{d,h} \cap E_{d,h}(f^{\pi}_d(\varepsilon),g_h^{\pi}(\varepsilon))$ and suppose it is non-empty. Then, there exists a feasible uniform equilibrium for $(d,h)$, so the couple changes only once of strategy profile and never again. Suppose $F_{d,h}$ is empty. Necessarily it must hold case (3) or (4). Suppose $f^{\pi}_d(\varepsilon) - \varepsilon \geq \alpha_{d,h}$ and $g_h^{\pi}(\varepsilon) - \varepsilon \sm \beta_{d,h}$ and consider the oracle given in the proof of \Cref{lemma:inf_rep_games_intersection_empty}. Then, the couple passes to gain $(\bar{f}, \bar{g} + \varepsilon)$, where
\begin{align*}
  &\bar{g} = \max\{g : \exists f, (f,g) \in E_{d,h}(f^{\pi}_d(\varepsilon),g_h^{\pi}(\varepsilon))\},\\
  &\bar{f} \in \{f : (f,\bar{g}) \in E_{d,h}(f^{\pi}_d(\varepsilon),g_h^{\pi}(\varepsilon))\}.  
\end{align*}

Let $(f_d^{\pi(1)}(\varepsilon),g_{h}^{\pi(1)}(\varepsilon))$ be the couple's reservation payoffs at the next iteration and consider again $F_{d,h} := E_{d,h} \cap E_{d,h}\bigl(f^{\pi(1)}_d(\varepsilon),g_h^{\pi(1)}(\varepsilon)\bigr)$. If $F_{d,h}$ is non-empty, the couple passes to play a feasible uniform equilibrium. Otherwise, the oracle computes a new payoff profile $(\bar{f}', \bar{g}')$ such that
\begin{align*}
  &\bar{g}' = \max\{g : \exists f, (f,g) \in E_{d,h}\bigl(f^{\pi(1)}_d(\varepsilon),g_h^{\pi(1)}(\varepsilon)\bigr)\},\\
  &\bar{f}' \in \{f : (f,\bar{g}') \in E_{d,h}\bigl(f^{\pi(1)}_d(\varepsilon),g_h^{\pi(1)}(\varepsilon)\bigr)\}.  
\end{align*}
Since $\pi(1)$ is $\varepsilon$-pairwise stable, it holds $f'_d \leq \bar{f} + \varepsilon$, $g' \leq (\bar{g} + \varepsilon) + \varepsilon$. Therefore, $(\bar{f}, \bar{g}+ \varepsilon) \in E_{d,h}\bigl(f^{\pi(1)}_d(\varepsilon),g_h^{\pi(1)}(\varepsilon)\bigr)$ and then, $\bar{g}' \geq \bar{g} + \varepsilon$. We conclude that at each iteration, either the couple changes to play a feasible uniform equilibrium, or player $h$ increases its payoff in at least $\varepsilon$. Since its payoff is bounded by its punishment level, the sequence converges in $T \propto \frac{1}{\varepsilon}$ iterations. If case (4) holds, the conclusion is the same: at each iteration, either the couple plays a feasible uniform equilibrium or player $d$ increases by at least $\varepsilon$ her payoff. Again, we obtain a $T \propto \frac{1}{\varepsilon}$ bound for the number of iterations. Thus, we obtain the number of iterations given in the statement of the theorem by considering the worst possible case.
\end{proof}

\begin{remark}
Adding \Cref{teo:convergence_algo_2_inf_rep_games} to \Cref{remark:complexity_stra_prof_mod_inf_rept_games}, we can conclude that computing an $\varepsilon$-renegotiation proof allocation for an infinitely repeated matching game is a polynomial problem.
\end{remark}

\Cref{tab:complexity_inf_rep_games} summarizes the complexity results found.

\begin{table}[H]
    \centering
    \begin{tabular}{c|c|c|c}
    \toprule
    Algorithms & Complexity per It & Max Nº of It & Constants\\
    \toprule
    DAC & $\mathcal{O}(N^3k^5L)$ & $C_1/\varepsilon$ & $C_1 \leq \max\limits_{d,h}(\max A_{d,h} - \min A_{d,h})$\\
    \hline
    Market procedure:  & - & $p(N)$ & - \\ 
    Demand sets & \multirow{1}{*}{$\mathcal{O}(N^2kL)$} & \multirow{1}{*}{-} & \multirow{1}{*}{-} \\
    Implementation & $\mathcal{O}(Nk^2L)$ & - & - \\
    \hline
    \multirow{3}{*}{Renegotiation process} & \multirow{3}{*}{ $\mathcal{O}(N^2k^{5})L)$ } & \multirow{3}{*}{ $C_2/\varepsilon$ } & $C_2 \leq \max_{d,h}\max \{D_{d},D_h\}$ \\
    & & & $D_{d} := \max A_{d,h} - \min A_{d,h}$\\
    & & & $D_{h} := \max B_{d,h} - \min B_{d,h}$\\
    \hline
    \end{tabular}
    \caption{Complexity infinitely repeated games: $N$ players per side, $k$ strategies per player, $L$ bits to encode the data, $p(N)$ polynomial on $N$, and $A_{d,h}, B_{d,h}$ payoff matrices of the stage games of $(d,h) \in D \times H$.}
    \label{tab:complexity_inf_rep_games}
\end{table}

\section{Complexity study for Strictly competitive matching games}\label{sec:complexity_strictly_competitive_games}

The class $\mathcal{S}$ of strictly competitive games, initially defined by Aumann \cite{aumann_almost_1961}, was fully characterized by Adler, Daskalakis, and Papadimitriou \cite{adler_note_2009} in the bi-matrix case. 

\begin{definition}
\textup{A bimatrix game $G = (S,T,A,-B)$, with $S,T$ finite pure strategy sets and $A,-B$ payoff matrices, is called a \textbf{strictly competitive game} if for any $x,x' \in \Delta(S), y,y' \in \Delta(T)$, $xAy - x'Ay'$ and $xBy - x'By'$ have always the same sign.}
\end{definition}

\begin{definition}
\textup{Given two matrices $A, B \in \mathbb{R}^{m\times n}$, we say that $B$ is an \textbf{affine variant} of $A$ if for some $\lambda \bbi 0$ and unrestricted $\mu \in \mathbb{R}$, $B = \lambda A + \mu U$, where $U$ is $m \times n$ all-ones matrix.}
\end{definition}

Adler et al. proved the following result.

\begin{theorem}\label{teo:affine_transformation_strictly_competitive_games}
If for all $x, x' \in X$ and $y,y' \in Y$, $xAy - x'Ay'$ and $xBy - x'By'$ have the same sign, then $B$ is an affine variant of $A$. Even more, the affine transformation is given by,
\begin{align*}
    A = \frac{a^{max} - a^{min}}{b^{max} - b^{min}} [B - b^{min}U] + a^{min}U, \text{ with }
    \left\{ \begin{array}{c}
     a^{max} := \max A,\ a^{min} := \min A\\ b^{max} := \max B,\ b^{min} := \min B\end{array} \right.
\end{align*}

If $a^{max} = a^{min}$, then it also holds that $b^{max} = b^{min}$ (and vice-versa), in which case clearly $A$ and $B$ are affine variants.
\end{theorem}

\Cref{teo:affine_transformation_strictly_competitive_games} allows us to extend all the results obtained in the previous section for zero-sum matching games to strictly competitive matching games. First of all, we prove that computing the affine transformations is a polynomial problem.


\begin{theorem}\label{teo:complexity_of_affine_transformation}
Let $\Gamma$ be a matching game in which all strategic games $G_{d,h} = (S_d,T_h,A_{d,h},-B_{d,h})$ are bi-matrix strictly competitive games. Let $\Gamma'$ be the affine transformation of $\Gamma$ in which all couples play zero-sum games. Then, computing $\Gamma'$ has complexity
$$\mathcal{O}\left(|D| + |H| + \sum_{d \in D} \sum_{h \in H} |S_d|\cdot|T_h|\right).$$
\end{theorem}

\begin{proof}
In order to obtain $\Gamma'$, besides computing all zero-sum games, we also need to compute all the new individually rational payoffs. Let $(d,h) \in D \times H$ be a potential couple that plays a strictly competitive game $(S_d,T_h,A_{d,h}, -B_{d,h})$. The complexity of computing their affine transformation to a zero-sum game $(S_d,T_h,B_{d,h}, -B_{d,h})$ is $\mathcal{O}|S_d|\cdot|T_h|)$, as we need to compute $a^{max}_{d,h}, a^{min}_{d,h}, b^{max}_{d,h}$, and $b^{min}_{d,h}$. Regarding the individually rational payoffs $(\underline{f}_d, \underline{g}_h)$, set $\alpha_{d,h} := \frac{a_{d,h}^{max} - a_{d,h}^{min}}{b_{d,h}^{max} - b_{d,h}^{min}}$. We take $\alpha_{d,h}$ so it is always lower or equal to $1$ (at least one of the two ways of taking the affine transformation guarantees this). Given $(x,y) \in X_d \times Y_h$ a strategy profile, note that,
\begin{align}
    xA_{d,h}y \geq \underline{f}_d &\Longleftrightarrow xB_{d,h}y \geq \frac{\underline{f}_d - (a_{d,h}^{min} - b_{d,h}^{min}\alpha_{d,h})}{\alpha_{d,h}},\label{eq:IRP_transformation_men}\\
    x(-B_{d,h})y \geq \underline{g}_h &\Longleftrightarrow xB_{d,h}y \leq -\underline{g}_h,\label{eq:IRP_transformation_women}
\end{align}
where we have used that $xUy = 1$. Unlike a ``standard'' matching game in which each player has a unique IRP that works for all possible partners, in the transformed game $\Gamma'$ doctors will have one IRP per hospital, given by \Cref{eq:IRP_transformation_men}. Formally, let 
$$\underline{f}_{d,h}' := \frac{\underline{f}_d - (a_{d,h}^{min} - b_{d,h}^{min}\alpha_{d,h})}{\alpha_{d,h}}, \text{for any } d \in D \text{ and } h \in H.$$
Then, a doctor $d$ accepts to be matched with hospital $h$ if and only her payoff is greater or equal than $\underline{f}_{d,h}'$. Regarding hospitals, it is enough considering $\underline{g}_{h}' := -\underline{g}_h$. Computing each coefficient takes constant time on the size of the agent sets and strategy sets. Thus, the complexity of transforming the IRPs is $\mathcal{O}|D| + |H|)$ plus some factor indicating the number of required bits.
\end{proof}

\subsection{Deferred-acceptance with competitions algorithm}

The analysis of the DAC algorithm (\Cref{Algo:Propose_dispose_algo_additive_separable_case}) complexity is not affected by the fact that doctors may have personalized IRPs for hospitals. Thus, from the complexity results of zero-sum games (\Cref{teo:complexity_propose_dispose_algo_zero_sum_case,teo:complexity_of_affine_transformation}) we conclude the following. 

\begin{corollary}
Computing $\varepsilon$-pairwise stable allocations in bi-matrix strictly competitive matching games is a polynomial problem as the DAC algorithm has a bounded number of iterations, each of them with complexity $\mathcal{O}((N^2+k^2)L)$, where $N$ bounds the number of players in the biggest side, $k$ bounds the number of pure strategies per player and $L$ is the number of bits required to encode all the data.
\end{corollary}

\subsection{Market procedure}

Let $\Gamma$ be a roommates matching game such that for each couple $(d,d') \in D\times D$, their game $G_{d,d'}$ belongs to the class of bi-matrix strictly competitive games in $\mathcal{S}$.

The complexity results obtained for zero-sum games in the computation of the demand sets (\Cref{teo:complexity_demand_sets_zero_sum_roommates_games}) and the mechanism to implement the output of the market procedure (\Cref{teo:complexity_implementation_roommates_problem}) can be extended to $\Gamma$ thanks to the affine transformation result (\Cref{teo:affine_transformation_strictly_competitive_games}). We conclude directly the following result.

\begin{corollary}
Computing the demand sets of all agents during an iteration of the market procedure has complexity $\mathcal{O}(N^2L)$ where $N$ bounds the number of doctors and $L$ is the number of bits required to encode all the data. In addition, the complexity of computing an allocation $\pi$ that implements the output $f$ of the market procedure (if it exists) is $\mathcal{O}(Nk^2L)$, where $k$ bounds the number of pure strategies per player.
\end{corollary}
 
\subsection{Renegotiation process}

As in the zero-sum case, we start with the complexity of computing a constrained Nash equilibrium. Let $G_{d,h} = (X_d,Y_h, A_{d,h}, -B_{d,h})$ be a bi-matrix strictly competitive game in mixed strategies and $(f_d,g_{h})$ be a pair of reservation payoffs. Let $(x,y)$ be an $\varepsilon$-$(f_d,g_{h})$-feasible strategy profile, that is,
\begin{align*}
    xA_{d,h}y + \varepsilon \geq f_d \text{ and } x(-B)_{d,h}y + \varepsilon \geq g_{h} \Longleftrightarrow xA_{d,h}y + \varepsilon \geq f_d \text{ and } xBy \leq - g_{h} - \varepsilon.
\end{align*}
It follows,
\begin{align*}
    xA_{d,h}y + \varepsilon \geq f_d &\Longleftrightarrow x \left(\alpha_{d,h} [B_{d,h} - b_{d,h}^{min}U] + a_{d,h}^{min}U\right) y + \varepsilon \geq f_d \\
    & \Longleftrightarrow \alpha_{d,h} xB_{d,h}y +  (a_{d,h}^{min} - b_{d,h}^{min}\alpha_{d,h})x U y + \varepsilon \geq f_d\\
    & \Longleftrightarrow \alpha_{d,h} xB_{d,h}y +  (a_{d,h}^{min} - b_{d,h}^{min}\alpha_{d,h}) + \varepsilon \geq f_d \\
    & \Longleftrightarrow xB_{d,h}y + \varepsilon \geq \frac{f_d - (a_{d,h}^{min} - b_{d,h}^{min}\alpha_{d,h})}{\alpha_{d,h}} - \varepsilon\cdot \frac{1-\alpha_{d,h}}{\alpha_{d,h}}.
\end{align*}
Recall we have taken $\alpha_{d,h} \in [0,1]$. Thus, in the zero-sum game $G'_{d,h} = (X_d, Y_h, B_{d,h})$, considering the pair $(f_d',g_{h}')$ of reservation payoffs given by,
\begin{align}\label{eq:outside_options_in_the_affine_game}
  f_d' := \frac{f_d -   (a_{d,h}^{min} - b_{d,h}^{min}\alpha_{d,h})}{\alpha_{d,h}}  - \varepsilon \cdot \frac{1-\alpha_{d,h}}{\alpha_{d,h}} \text{ and } g_{h}' := -g_{h},
\end{align}
the sets of feasible strategy profiles, as well as the sets of CNE of $G_{d,h}$ and $G'_{d,h}$, coincide. Therefore, to compute an $\varepsilon$-$(f_d,g_{h})$-constrained Nash equilibrium of the strictly competitive game, we can use the following scheme: 
\begin{itemize}
    \item[1.] Compute the transformation from $A_{d,h}$ to $B_{d,h}$ and define the zero-sum game $G'_{d,h}$.
    \item[2.] Consider the new reservation payoffs $(f_d', g_{h}')$ as in \Cref{eq:outside_options_in_the_affine_game}.
    \item[3.] Compute an $\varepsilon$-$(f_d', g_{h}')$-CNE for the zero-sum game, namely $(x',y')$.
\end{itemize}

\begin{proposition}\label{prop:affine_transformed_CNE_are_CNE}
The scheme above computes an $\varepsilon$-$(f_d',g_{h}')$-CNE of $G_{d,h}$.
\end{proposition}

\begin{proof}
Let $(x',y')$ be an $\varepsilon$-$(f_d', g_{h}')$-CNE of the zero-sum game $G'_{d,h}$. It holds,
\begin{itemize}
    \item[1.] $g_{h}' + \varepsilon \geq x'B_{d,h}y' \geq f_d' - \varepsilon$,
    \item[2.] For any $x \in X_d$ such that $xB_{d,h}y'\leq g_{h} + \varepsilon$, $(x-x')B_{d,h}y' \leq \varepsilon$,
    \item[3.] For any $y \in Y_h$ such that $x'B_{d,h}y + \varepsilon \geq f_d$, $x'B_{d,h}(y'-y) \leq \varepsilon$.
\end{itemize}
From (1) we obtain that $x'(-B_{d,h})y' \geq -g_{h}' - \varepsilon = g_{h} - \varepsilon$, and $x'B_{d,h}y' \geq f_d' - \varepsilon$, which implies that $x'A_{d,h}y' \geq f_d - \varepsilon$, so $(x',y')$ is $(f_d,g_{h})$-feasible in the game $G_{d,h}$. Let $x \in X_d$ such that $x(-B_{d,h})y' + \varepsilon \geq g_{h}$. Then, $xB_{d,h}y' \leq g'_{d,h} - \varepsilon$. From (2), $(x-x')B_{d,h}y' + \varepsilon$. Noticing that $\alpha_{d,h}(x-x')B_{d,h}y' = (x-x')A_{d,h}y'$, we obtain that $ (x-x')A_{d,h}y' \leq \alpha_{d,h} \varepsilon \leq \varepsilon$, as $\alpha_{d,h}$ was taken lower of equal than $1$. Analogously, suppose there is $y \in Y_h$ such that $x'A_{d,h}y + \varepsilon \geq f_d$. Then, $x'B_{d,h}y + \varepsilon \geq f_d'$. From (3), $x'(-B_{d,h})(y-y') \leq \varepsilon$. Therefore, $(x',y')$ is an $\varepsilon$-CNE of $G_{d,h}$.
\end{proof}

From  \Cref{prop:affine_transformed_CNE_are_CNE} and the complexity of computing a constrained Nash equilibrium of a zero-sum game  (\Cref{teo:zero_sum_games_are_eps_feasible_and_polynomial}), we obtain the following result.

\begin{corollary}
Let $G_{d,h} = (S_d,T_h,A_{d,h},-B_{d,h})$ be a bi-matrix strictly competitive game and $(f_d,g_{h})$ be a pair of reservation payoffs. The complexity of computing an $\varepsilon$-$(f_d,g_{h})$-constrained Nash equilibrium is
$$ \mathcal{O}\left( \max\{|S_d|,|T_h|\}^{2.5} \cdot \min\{|S_d|,|T_h|\} \cdot L_{d,h} \right),$$
with $L_{d,h}$ the number of bits required to encode the payoff matrices.
\end{corollary}

Finally, from the bounded number of iterations of the renegotiation process for zero-sum games (\Cref{teo:algorithm_2_ends_in_finite_time_for_zero_sum_games}) we deduce the following.

\begin{corollary}
The $\varepsilon$-renegotiation process (\ref{Algo:strategy_profiles_modification_eps_version}) ends in a finite number of iterations $T\propto \frac{1}{\varepsilon}$ in bi-matrix strictly competitive games. 
\end{corollary}

Let $T := \max\{\max A_{d,h} - \min A_{d,h} : (d,h) \in D \times H\}$. The following table summarizes the complexity results found. 

\begin{table}[H]
    \centering
    \begin{tabular}{c|c|c}
    \toprule
    Algorithms & Complexity per It & Max Number It\\
    \toprule
    DAC  & $\mathcal{O}((N^2+k^2)L)$ & $T/\varepsilon$\\
    \hline
    Market procedure:  & - & $p(N)$\\ 
    Demand sets & \multirow{1}{*}{$\mathcal{O}(N^2L)$} & \multirow{1}{*}{-} \\
    Implementation & $\mathcal{O}(Nk^2L)$ & -\\
    \hline
    Renegotiation process & $\mathcal{O}(N^4k^{3.5})L)$ & $T/\varepsilon$\\
    \hline
    Affine & \multirow{2}{*}{ $\mathcal{O}(N^2k^2)$ } & \multirow{2}{*}{-} \\
    Transformation & &\\
    \hline
    \end{tabular}
    \caption{Complexity strictly competitive matching games: $N$ players per side, $k$ strategies per player, $L$ bits to encode the data, and $p(N)$ a polynomial on $N$.}
    \label{tab:complexity_str_competitive_games}
\end{table}





\end{document}